%% file: main.tex
\documentclass[onecolumn,a4paper,11pt,accepted=2020-01-28]{quantumarticle} 
\pdfoutput=1

\usepackage[utf8]{inputenc}
\usepackage{fullpage,amsfonts,bm}
\usepackage{amsmath,amssymb,amsthm,mathtools}
\usepackage{mleftright}\mleftright
\usepackage[pagebackref,final]{hyperref}
\usepackage{algpseudocode}
\usepackage{float}

\usepackage{verbatim}
\usepackage{ifdraft}
\ifdraft{\newcommand{\authnote}[3]{{\color{#3} {\bf  #1:} #2}}}{\newcommand{\authnote}[3]{}}

\usepackage{tikz}
\usetikzlibrary{calc,arrows,quotes,angles}

\usepackage{subcaption}

\usepackage{thmtools}
\usepackage{thm-restate}

\providecommand{\cx}{0}
\providecommand{\cy}{0}
\providecommand{\scale}{0}
\providecommand{\sfw}{0}

\floatstyle{boxed}

\newfloat{algorithm}{t}{lop}
\floatname{algorithm}{Algorithm}

\newfloat{metaalgorithm}{t}{lop}
\floatname{metaalgorithm}{Meta-Algorithm}

\newtheorem{theorem}{Theorem}

\newtheorem{lemma}[theorem]{Lemma}
\newtheorem{claim}[theorem]{Claim}
\newtheorem{remark}[theorem]{Remark}

\newtheorem{corollary}[theorem]{Corollary}
\newtheorem{definition}[theorem]{Definition}

\mathchardef\mhyphen="2D

\newcommand{\ignore}[1]{}

\newcommand{\ket}[1]{|#1\rangle}
\newcommand{\bra}[1]{\langle#1|}

\newcommand{\ketbra}[2]{|#1\rangle\! \langle #2|}

\newcommand{\Tr}{\mbox{\rm Tr}}
\newcommand{\tr}[1]{\Tr\left(#1\right)}
\newcommand{\opt}{\mbox{\rm OPT}}

\newcommand{\eps}{\varepsilon}
\newcommand{\nrm}[1]{\left\lVert#1\right\rVert}
\newcommand{\maxnrm}[1]{\left\lVert#1\right\rVert_{\max}}
\def\01{\{0,1\}}

\newcommand{\MAJ}{\mbox{\rm MAJ}}
\newcommand{\OR}{\mbox{\rm OR}}
\newcommand{\MOM}{\MAJ_a\mhyphen\OR_b\mhyphen\MAJ_c}
\newcommand{\bigOn}[1]{{\mathcal O}\left(#1\right)}
\newcommand{\bigO}{{\mathcal O}}

\newcommand{\bOt}[1]{\widetilde{\mathcal O}\left(#1\right)}
\newcommand{\cp}{c^{\prime}}
\newcommand{\pt}{\tilde{\mathcal{P}}}

\begin{document}

\title{Quantum SDP-Solvers: Better upper and lower bounds}
\author{Joran van Apeldoorn}
\author{Andr\'as Gily\'en}
\author{Sander Gribling}
\affiliation{QuSoft, CWI, Amsterdam, the Netherlands. $\{${\tt apeldoor,gilyen,gribling}$\}${\tt @cwi.nl}}
\author{Ronald de Wolf}
\affiliation{QuSoft, CWI and University of Amsterdam, the Netherlands.  {\tt rdewolf@cwi.nl}}

\date{}
\maketitle

\begin{abstract}
  Brand\~ao and Svore~\cite{brandao2016QSDPSpeedup} recently gave quantum algorithms for approximately solving semidefinite programs, which in some regimes are faster than the best-possible classical algorithms in terms of the dimension~$n$ of the problem and the number~$m$ of constraints, but worse in terms of various other parameters. In this paper we improve their algorithms in several ways, getting better dependence on those other parameters. To this end we develop new techniques for quantum algorithms, for instance a general way to efficiently implement smooth functions of sparse Hamiltonians, and a generalized minimum-finding procedure.

  We also show limits on this approach to quantum SDP-solvers, for instance for combinatorial optimization problems that have a lot of symmetry.
  Finally, we prove some general lower bounds showing that in the worst case, the complexity of every quantum LP-solver (and hence also SDP-solver) has to scale linearly with $mn$ when $m\approx n$, which is the same as classical.
\end{abstract}

\newpage

\tableofcontents

\newpage

\section{Introduction}

\subsection{Semidefinite programs} \label{subsec:sdp}

In the last decades, particularly since the work of Gr{\"o}tschel, Lov{\'{a}}sz, and Schrijver~\cite{grotschel1988GeomAlgAndConvOpt}, \emph{semidefinite programs} (SDPs) have become an important tool for designing efficient optimization and approximation algorithms. SDPs generalize the better-known \emph{linear} programs (LPs), but (like LPs) they are still efficiently solvable. Thanks to their stronger expressive power, SDPs can sometimes give better approximation algorithms than LPs.

The basic form of an SDP is the following:
\begin{align} \label{eq:SDP}
  \max \quad &\Tr(CX) \\
  \text{s.t.}\ \ \ &\Tr(A_j X) \leq b_j \quad \text{ for all } j \in [m], \notag \\
             &X \succeq 0, \notag
\end{align}
where $[m]:=\{1,\ldots,m\}$.
The input to the problem consists of Hermitian $n\times n$ matrices $C,A_1,\ldots,A_m$ and reals $b_1,\ldots,b_m$.  For normalization purposes we assume $\nrm{C},\nrm{A_j} \leq 1$.
The number of constraints is~$m$ (we do not count the standard $X\succeq 0$ constraint for this). The variable~$X$ of this SDP is an $n\times n$ positive semidefinite (psd) matrix.
LPs correspond to the case where all matrices are diagonal.

A famous example is the algorithm of Goemans and Williamson~\cite{goemans1995MaxCutSDP} for approximating the size of a maximum cut in a graph~$G=([n],E)$: the maximum, over all subsets $S$ of vertices, of the number of edges between $S$ and its complement~$\bar{S}$. Computing MAXCUT$(G)$ exactly is NP-hard. It corresponds to the following integer program
\begin{align*}
  \max \quad &\frac{1}{2}\sum_{\{i,j\}\in E}(1-v_iv_j)\\
  \text{s.t.}\ \ \ & v_j\in\{+1,-1\}\quad \text{ for all } j \in [n],\end{align*}
using the fact that $(1-v_iv_j)/2=1$ if $v_i$ and $v_j$ are different signs, and $(1-v_iv_j)/2=0$ if they are the same.
We can relax this integer program by replacing the signs $v_j$ by unit vectors, and replacing the product $v_iv_j$ in the objective function by the dot product $v_i^T v_j$.
We can implicitly optimize over such vectors (of unspecified dimension) by explicitly optimizing over an $n\times n$ psd matrix $X$ whose diagonal entries are~1.
This $X$ is the Gram matrix of the vectors $v_1,\ldots,v_n$, so $X_{ij}=v_i^T v_j$.
The resulting SDP is
\begin{align*}
  \max \quad &\frac{1}{2}\sum_{\{i,j\}\in E}(1-X_{ij}) \\
  \text{s.t.}\ \ \ &\Tr(E_{jj} X) =1 \quad \text{ for all } j \in [n], \\
             &X \succeq 0,
\end{align*}
where $E_{jj}$ is the $n\times n$ matrix that has a~1 at the $(j,j)$-entry, and 0s elsewhere.

This SDP is a relaxation of a maximization problem, so it may overshoot the correct value, but Goemans and Williamson showed that an optimal solution to the SDP can be rounded to a cut in $G$ whose size is within a factor $\approx 0.878$ of MAXCUT$(G)$ (which is optimal under the Unique Games Conjecture~\cite{khot2007OptimalInapxMAXCUT}). This SDP can be massaged into the form of~\eqref{eq:SDP} by replacing the equality $\Tr(E_{jj} X)=1$ by inequality $\Tr(E_{jj} X)\leq 1$ (so $m=n$) and letting $C$ be a properly normalized version of the Laplacian of~$G$.

\subsection{Classical solvers for LPs and SDPs}

Ever since Dantzig's development of the simplex algorithm for solving LPs in the 1940s~\cite{dantzig1947SimplexAlgorithm}, much work has gone into finding faster solvers, first for LPs and then also for SDPs. The simplex algorithm for LPs (with some reasonable pivot rule) is usually fast in practice, but has worst-case exponential runtime. The ellipsoid method and interior-point methods can solve LPs and SDPs in polynomial time; they will typically \emph{approximate} the optimal value to arbitrary precision~\cite{grotschel1981EllipsoidMethod,nesterov1994InteriorPointAlgs}.
The best known general SDP-solvers~\cite{lee2015FasterCuttingPlaneConvexOpt} approximate the optimal value \opt\ of such an SDP up to additive error $\eps$, with complexity
$$
{\mathcal O}(m(m^2+n^\omega + mns)\, \text{polylog}(m,n,R,1/\eps)),
$$ 
where $\omega\in[2,2.373)$ is the (still unknown) optimal exponent for matrix multiplication; $s$ is the \emph{sparsity}: the maximal number of non-zero entries per row of the input matrices; and $R$ is an upper bound on the trace of an optimal~$X$.\footnote{See Lee, Sidford, and Wong~\cite[Section~10.2 of arXiv version~2]{lee2015FasterCuttingPlaneConvexOpt}, and note that our $m,n$ are their $n,m$, their $S$ is our $mns$, and their $M$ is our~$R$. The bounds for other SDP-solvers that we state later also include another parameter~$r$; it follows from the assumptions of~\cite[Theorem~45 of arXiv version~2]{lee2015FasterCuttingPlaneConvexOpt} that in their setting $r\leq mR$, and hence $r$ is absorbed in the $\log^{O(1)}(mnR/\eps)$ factor.}
The assumption here is that the rows and columns of the matrices of SDP~\eqref{eq:SDP} can be accessed as adjacency lists: we can query, say, the $\ell$th non-zero entry of the $k$th row of matrix $A_j$ in constant time.

Arora and Kale~\cite{arora2016CombPrimDualSDP} gave an alternative way to approximate \opt, using a matrix version of the ``multiplicative weights update'' method.\footnote{See also~\cite{arora2012MultiplicativeWeightsAlg} for a subsequent survey; the same algorithm was independently discovered around the same time in the context of learning theory~\cite{tsuda2005MatrixExpGradUpdates,warmuth2012OnlineVarMinimiz}. In the optimization community, first-order methods for semidefinite programming have been considered for instance in~\cite{renegar2016EfficientSubgradientMethods,renegar2015AcceleratedFirstOrderMethods}.}
In Section~\ref{sec:classicalAK} we will describe their framework in more detail, but in order to describe our result we will start with an overly simplified sketch here.
The algorithm goes back and forth between candidate solutions to the primal SDP and to the corresponding \emph{dual} SDP, whose variables are non-negative reals $y_1,\ldots,y_m$:
\begin{align} \label{eq:SDP2}
  \min \quad & b^T y \\ \notag
  \text{s.t.}\ \ \ &\sum_{j=1}^m y_j A_j - C \succeq 0,\\ \notag
             &y \geq 0.
\end{align}
Under assumptions that will be satisfied everywhere in this paper, strong duality applies:
the primal SDP~\eqref{eq:SDP} and dual SDP~\eqref{eq:SDP2} will have the same optimal value \opt.
The algorithm does a binary search for \opt\ by trying different guesses $\alpha$ for it. Suppose we have fixed some $\alpha$, and want to find out whether $\alpha$ is bigger or smaller than $\opt$.  Start with some candidate solution $X^{(1)}$ for the primal, for example a multiple of the identity matrix ($X^{(1)}$ has to be psd but need not be a \emph{feasible} solution to the primal). This $X^{(1)}$ induces the following polytope:
\begin{align*}
  \mathcal{P}_{\eps}(X^{(1)}) := \{ y \in \mathbb R^m:\ &b^T y  \leq \alpha, \\
                                                        &\tr{\Big(\sum_{j=1}^m y_j A_j - C\Big) X^{(1)}} \geq -\eps,\\
                                                        & y \geq 0 \}.
\end{align*}
This polytope can be thought of as a relaxation of the feasible region of the dual SDP with the extra constraint that $\opt\leq\alpha$: instead of requiring that $\sum_j y_j A_j - C$ is psd, we merely require that its inner product with the particular psd matrix $X^{(1)}$ is not too negative.
The algorithm then calls an ``oracle'' that provides a $y^{(1)}\in \mathcal{P}_{\eps}(X^{(1)})$, or outputs ``fail'' if $\mathcal{P}_{0}(X^{(1)})$ is empty (how to efficiently implement such an oracle depends on the application).
In the ``fail'' case we know there is no dual-feasible $y$ with objective value $\leq\alpha$,
so we can increase our guess $\alpha$ for \opt, and restart.
In case the oracle produced a $y^{(1)}$, this is used to define a Hermitian matrix $H^{(1)}$ and a new candidate solution $X^{(2)}$ for the primal, which is proportional to $e^{-H^{(1)}}$.
Then the oracle for the polytope $\mathcal{P}_{\eps}(X^{(2)})$ induced by this $X^{(2)}$ is called to produce a candidate $y^{(2)}\in\mathcal{P}_{\eps}(X^{(2)})$ for the dual (or ``fail''), this is used to define $H^{(2)}$ and $X^{(3)}$ proportional to $e^{-H^{(2)}}$, and so on.

Surprisingly, the average of the dual candidates $y^{(1)},y^{(2)},\ldots$ converges to a nearly-dual-feasible solution.
Let $R$ be an upper bound on the trace of an optimal $X$ of the primal, $r$ be an upper bound on the sum of entries of an optimal $y$ for the dual, and $w^*$ be the ``width'' of the oracle for a certain SDP: the maximum of $\nrm{\sum_{j=1}^m y_j A_j - C}$ over all psd matrices $X$ and all vectors $y$ that the oracle may output for the corresponding polytope~$\mathcal{P}_{\eps}(X)$. In general we will not know the width of an oracle exactly, but only an upper bound $w \geq w^*$, that may depend on the SDP; this is, however, enough for the Arora-Kale framework. In Section~\ref{sec:classicalAK} we will show that without loss of generality we can assume the oracle returns a $y$ such that $\nrm{y}_1 \leq r$. Because we assumed $\nrm{A_j},\nrm{C} \leq 1$, we have $w^* \leq r+1$ as an easy width-bound.
General properties of the multiplicative weights update method guarantee that after $T=\widetilde{\mathcal O}(w^2R^2/\eps^2)$ iterations\footnote{The $\widetilde{\mathcal O}(\cdot)$ notation hides polylogarithmic factors in all parameters.}, if no oracle call yielded ``fail'', then the vector $\frac{1}{T}\sum_{t=1}^T y^{(t)}$ is close to dual-feasible and satisfies $b^T y\leq\alpha$. This vector can then be turned into a dual-feasible solution by tweaking its first coordinate, certifying that $\opt\leq\alpha+\eps$, and we can decrease our guess $\alpha$ for $\opt$ accordingly.

The framework of Arora and Kale is really a meta-algorithm, because it does not specify how to implement the oracle.  They themselves provide oracles that are optimized for special cases, which allows them to give a very low width-bound for these specific SDPs. For example for the MAXCUT SDP, they obtain a solver with near-linear runtime in the number of edges of the graph. They also observed that the algorithm can be made more efficient by not explicitly calculating the matrix $X^{(t)}$ in each iteration: the algorithm can still be made to work if instead of providing the oracle with $X^{(t)}$, we feed it good estimates of $\Tr(A_j X^{(t)})$ and $\Tr(C X^{(t)})$.  Arora and Kale do not describe oracles for general SDPs, but as we show at the end of Section~\ref{sec:runtime} (using Appendix~\ref{app:trace} to estimate $\Tr(A_j X^{(t)})$ and $\Tr(C X^{(t)})$), one can get a general classical SDP-solver in their framework with complexity
\begin{equation}\label{eq:AKgeneralupperbound}
  \bOt{nms\left(\frac{Rr}{\eps}\right)^{\!\!4}+ns\left(\frac{Rr}{\eps}\right)^{\!\!7\,}}.
\end{equation}
Compared to the complexity of the SDP-solver of~\cite{lee2015FasterCuttingPlaneConvexOpt}, this has much worse dependence on $R$ and~$\eps$, but better dependence on $m$ and $n$. Using the Arora-Kale framework is thus preferable over standard SDP-solvers for the case where $Rr$ is small compared to $mn$, and a rough approximation to \opt\ (say, small constant $\eps$) is good enough.
It should be noted that for many specific cases, Arora and Kale get significantly better upper bounds than~\eqref{eq:AKgeneralupperbound} by designing oracles that are specifically optimized for those cases.

\subsection{Quantum SDP-solvers: the Brand\~ao-Svore algorithm}

Given the speed-ups that \emph{quantum} computers give over classical computers for various problems~\cite{shor1994Factoring,grover1996QSearch,durr2004QQueryCompGraph,ambainis2004QWalkForElementDist,harrow2009QLinSysSolver}, it is natural to ask whether quantum computers can solve LPs and SDPs more efficiently as well. Very little was known about this, until recently when Brand\~ao and Svore~\cite{brandao2016QSDPSpeedup} discovered quantum algorithms that significantly outperform classical SDP-solvers in certain regimes.  Because of the general importance of quickly solving LPs and SDPs, and the limited number of quantum algorithms that have been found so far, this is a very interesting development.

The key idea of the Brand\~ao-Svore algorithm is to take the Arora-Kale approach and to replace two of its steps by more efficient quantum subroutines. First, given a vector $y^{(t-1)}$, it turns out one can use ``Gibbs sampling'' to prepare the new primal candidate $X^{(t)}\propto e^{-H^{(t-1)}}$ \emph{as a $\log(n)$-qubit quantum state $\rho^{(t)}:=X^{(t)}/\Tr(X^{(t)})$} in much less time than needed to compute $X^{(t)}$ as an $n\times n$ matrix.
Second, one can efficiently implement the oracle for $\mathcal{P}_{\eps}(X^{(t)})$ based on a number of copies of $\rho^{(t)}$, using those copies to estimate $\Tr(A_j \rho^{(t)})$ and $\Tr(A_j X^{(t)})$ when needed (note that $\Tr(A\rho)$ is the expectation value of operator~$A$ for the quantum state~$\rho$). This is based on something called ``Jaynes's principle.'' The resulting oracle is weaker than what is used classically, in the sense that it outputs a sample $j\sim y_j/\nrm{y}_1$ rather than the whole vector~$y$. However, such sampling still suffices to make the  algorithm work (it also means we can assume the vector $y^{(t)}$ to be quite sparse).

Using these ideas, Brand\~ao and Svore obtain a quantum SDP-solver of complexity
$$
\widetilde{\mathcal O}(\sqrt{mn}s^2 R^{32} /\delta^{18}),
$$
with \emph{multiplicative} error $1\pm\delta$ for the special case where $b_j\geq 1$ for all $j\in[m]$, and $\opt\geq 1$ (the latter assumption allows them to convert additive error~$\eps$ to multiplicative error~$\delta$)~\cite[Corollary~5 in arXiv version 4]{brandao2016QSDPSpeedup}. They describe a reduction to transform a general SDP of the form~\eqref{eq:SDP} to this special case, but that reduction significantly worsens the dependence of the complexity on the parameters $R$, $r$, and $\delta$.

Note that compared to the runtime~\eqref{eq:AKgeneralupperbound} of our general instantiation of the original Arora-Kale framework, there are quadratic improvements in both $m$ and $n$, corresponding to the two quantum modifications made to Arora-Kale.  However, the dependence on $R,r,s$ and $1/\eps$ is much worse now than in~\eqref{eq:AKgeneralupperbound}. This quantum algorithm thus provides a speed-up only in regimes where $R,r,s,1/\eps$ are fairly small compared to $mn$ (finding good examples of SDPs in such regimes is an open problem).

\subsection{Our results}

In this paper we present two sets of results: improvements to the Brand\~ao-Svore algorithm, and better lower bounds for the complexity of quantum LP-solvers (and hence for quantum SDP-solvers as well).

\subsubsection{Improved quantum SDP-solver}
Our quantum SDP-solver, like the Brand\~ao-Svore algorithm, works by quantizing some aspects of the Arora-Kale algorithm. However, the way we quantize is different and faster than theirs.

First, we give a more efficient procedure to estimate the quantities $\Tr(A_j\rho^{(t)})$ required by the oracle. Instead of first preparing some copies of Gibbs state $\rho^{(t)}\propto e^{-H^{(t-1)}}$ as a mixed state, we coherently prepare a purification of $\rho^{(t)}$, which can then be used to estimate $\Tr(A_j\rho^{(t)})$ more efficiently using amplitude-estimation techniques.
Also, our purified Gibbs sampler has logarithmic dependence on the error, which is exponentially better than the Gibbs sampler of Poulin and Wocjan~\cite{poulin2009GibbsSamplingAndEval} that Brand\~ao and Svore invoke. Chowdhury and Somma~\cite{chowdhury2016QGibbsSampling} also gave a Gibbs sampler with logarithmic error-dependence, but assuming query access to the entries of $\sqrt{H}$ rather than $H$ itself.

Second, we have a different implementation of the oracle, without using Gibbs sampling or Jaynes's principle (though, as mentioned above, we still use purified Gibbs sampling for approximating the $\Tr(A\rho)$ quantities). We observe that the vector $y^{(t)}$ can be made very sparse: \emph{two} non-zero entries suffice.\footnote{Independently of us, Ben-David, Eldar, Garg, Kothari, Natarajan, and Wright (at MIT), and separately Ambainis observed that in the special case where all $b_j$ are at least~1, the oracle can even be made 1-sparse, and the one entry can be found using one Grover search over $m$ points (in both cases personal communication 2017). The same happens implicitly in our Section~\ref{sec:oracle} in this case. However, in general 2 non-zero entries are necessary in~$y$.}
We then show how we can efficiently find such a 2-sparse vector (rather than merely sampling from it) using two applications of a new generalization of the well-known quantum minimum-finding algorithm of D{\"u}rr and H{\o}yer~\cite{durr1996QMinimumFinding}, which is based on Grover search~\cite{grover1996QSearch}.

These modifications both simplify and speed up the quantum SDP-solver, resulting in complexity
$$
\widetilde{\mathcal O}(\sqrt{mn}s^2 (Rr/\eps)^{8}).
$$
The dependence on $m$, $n$, and $s$ is the same as in Brand\~ao-Svore, but our dependence on $R$, $r$, and $1/\eps$ is substantially better.
Note that each of the three parameters $R$, $r$, and $1/\eps$ now occurs with the same 8th power in the complexity. This is no coincidence: as we show in Appendix~\ref{app:reductions}, these three parameters can all be traded for one another, in the sense that we can massage the SDP to make each one of them small at the expense of making the others proportionally bigger.
These trade-offs suggest we should actually think of $Rr/\eps$ as \emph{one} parameter of the primal-dual pair of SDPs, not three separate parameters. For the special case of LPs, we can improve the runtime to
$$
\widetilde{\mathcal O}(\sqrt{mn} (Rr/\eps)^{5}).
$$

Like in Brand\~ao-Svore, our quantum oracle produces very sparse vectors $y$, in our case even of sparsity~2. This means that after $T$ iterations, the final $\eps$-optimal dual-feasible vector (which is a slightly tweaked version of the average of the $T$ $y$-vectors produced in the $T$ iterations) has only $\mathcal O(T)$ non-zero entries.  Such sparse vectors have some advantages, for example they take much less space to store than arbitrary $y\in\mathbb{R}^m$.
In fact, to get a sublinear runtime in terms of~$m$, this is necessary.
However, this sparsity of the algorithm's output also points to a weakness of these methods: if \emph{every} $\eps$-optimal dual-feasible vector $y$ has many non-zero entries, then the number of iterations needs to be large.  For example, if every $\eps$-optimal dual-feasible vector $y$ has $\Omega(m)$ non-zero entries, then these methods require $T=\Omega(m)$ iterations before they can reach an $\eps$-optimal dual-feasible vector. Since $T = \bigO\left(\frac{R^2r^2}{\eps^2} \ln(n)\right)$ this would imply that $\frac{Rr}{\eps} = \Omega(\sqrt{m/\ln(n)})$, and hence many classical SDP-solvers would have a better complexity than our quantum SDP-solver.
As we show in Section~\ref{sec:downside}, this will actually be the case for families of SDPs that have a lot of symmetry.
\subsubsection{Tools that may be of more general interest}
Along the way to our improved SDP-solver, we developed some new techniques that may be of independent interest. These are mostly tucked away in appendices, but here we will highlight two.

\paragraph{Implementing smooth functions of a given Hamiltonian.}
Let $f:\mathbb{R}\to\mathbb{C}$ be a function.
In Appendix~\ref{apx:LowWeight} we describe a general technique to apply a function $f(H)$ of a sparse Hamiltonian~$H$ to a given state $\ket{\phi}$.\footnote{Here a univariate function $f:\mathbb{R}\to\mathbb{C}$ is applied to a Hermitian matrix $H$ in the standard way, by acting on the eigenvalues of $H$: if $H$ has diagonalization $H=U^{-1}DU$, with $D$ the diagonal matrix of $H$'s eigenvalues, then $f(H)$ is the matrix $U^{-1}f(D)U$, where the diagonal matrix $f(D)$ is obtained from $D$ by applying $f$ to its diagonal entries.}
Roughly speaking, what applying $f(H)$ to $\ket{\phi}$ means, is that we want a unitary circuit that maps $\ket{0}\ket{\phi}$ to $\ket{0}f(H)\ket{\phi}+\ket{1}\ket{*}$.
If need be, we can then combine this with amplitude amplification to boost the $\ket{0}f(H)\ket{\phi}$ part of the state.
If the function~$f$ can be approximated well by a low-degree Fourier series, then our preparation will be efficient
in the sense of using few queries to $H$ and few other gates.
The novelty of our approach is that we construct a good Fourier series from the polynomial that approximates~$f$ (for example a truncated Taylor series for~$f$). Our Theorem~\ref{thm:Taylor} can be easily applied to various smooth functions without using involved integral approximations, unlike previous works building on similar techniques. Our most general result, Corollary~\ref{cor:patched}, only requires that the function $f$ can be nicely approximated locally around each possible eigenvalue of $H$, improving on Theorem~\ref{thm:Taylor}.

In this paper we mostly care about the function $f(x)=e^{-x}$, which is what we want for generating a purification of the Gibbs state corresponding to~$H$; and the function $f(x)=\sqrt{x}$, which is what we use for estimating quantities like $\Tr(A\rho)$.
However, our techniques apply much more generally than to these two functions. For example, they also simplify the analysis of the improved linear-systems solver of Childs et al.~\cite{childs2015QLinSysExpPrec}, where the relevant function is $f(x)=1/x$. As in their work, the Linear Combination of Unitaries technique of Childs et al.~\cite{childs2012HamSimLCU,berry2014HamSimTaylor,berry2015HamSimNearlyOpt} is a crucial tool for us.

\paragraph{A generalized minimum-finding algorithm.}
D{\"u}rr and H{\o}yer~\cite{durr1996QMinimumFinding} showed how to find the minimal value of a function $f:[N]\to\mathbb{R}$ using $\bigO(\sqrt{N})$ queries to~$f$,
by repeatedly using Grover search to find smaller and smaller elements of the range of~$f$. In Appendix~\ref{app:genMinFind} we describe a more general
minimum-finding procedure. Suppose we have a unitary $U$ which prepares a quantum state $U\ket{0}=\sum_{k=1}^{N}\ket{\psi_k}\ket{x_k}$, where the $\ket{\psi_k}$ are unnormalized states. Our procedure can find the minimum value $x_{k^*}$ among the $x_k$'s that have support in the second register, using roughly $\bigO(1/\nrm{\psi_{k^*}})$ applications of $U$ and~$U^{-1}$. Also, upon finding the minimal value $k^*$ the procedure actually outputs the normalized state proportional to $\ket{\psi_{k^*}}\ket{x_{k^*}}$. This immediately gives the D{\"u}rr-H{\o}yer result as a special case, if we take $U$ to produce $U\ket{0}=\frac{1}{\sqrt{N}}\sum_{k=1}^{N}\ket{k}\ket{f(k)}$ using one query to~$f$. Unlike D{\"u}rr-H{\o}yer, we need not assume direct query access to the individual values $f(k)$.

More interestingly for us, for a given $n$-dimensional Hamiltonian~$H$, if we combine our minimum-finder with phase estimation using unitary $U=e^{iH}$ on one half of a maximally entangled state, then we obtain an algorithm for estimating the smallest eigenvalue of~$H$ (and preparing its ground state) using roughly $\bigO(\sqrt{n})$ applications of phase estimation with~$U$. A similar result on approximating the smallest eigenvalue of a Hamiltonian was already shown by Poulin and Wocjan~\cite{poulin2009PrepGndStateManyBody}, but we improve on their analysis to be able to apply it as a subroutine in our procedure to estimate~$\Tr(A_j\rho)$.

\subsubsection{Lower bounds}
What about lower bounds for quantum SDP-solvers? Brand\~ao and Svore already proved that a quantum SDP-solver has to make $\Omega(\sqrt{n}+\sqrt{m})$ queries to the input matrices, for some SDPs. Their lower bound is for a family of SDPs where $s,R,r,1/\eps$ are all constant, and is by reduction from a search problem.

In this paper we prove lower bounds that are quantitatively stronger in $m$ and~$n$, but for SDPs with non-constant $R$ and $r$.
The key idea is to consider a Boolean function $F$ on $N=abc$ input bits that is the composition of an $a$-bit majority function with a $b$-bit OR function with a $c$-bit majority function. The known quantum query complexities of majority and OR, combined with composition properties of the adversary lower bound, imply that every quantum algorithm that computes this function requires $\Omega(a\sqrt{b}\,c)$ queries.  We define a family of LPs with non-constant $R$ and $r$ such that a constant-error approximation of \opt\ computes~$F$.
Choosing $a$, $b$, and $c$ appropriately, this implies a lower bound of
$$
\Omega\left(\sqrt{\max\{n,m\}} \left( \min\{n,m\} \right)^{3/2} \right)
$$
queries to the entries of the input matrices for quantum LP-solvers. Since LPs are SDPs with sparsity $s=1$, we get the same lower bound for quantum SDP-solvers. If $m$ and $n$ are of the same order, this lower bound is $\Omega(mn)$, the same scaling with $mn$ as the classical general instantiation of Arora-Kale~\eqref{eq:AKgeneralupperbound}. In particular, this shows that we cannot have an $O(\sqrt{mn})$ upper bound without simultaneously having polynomial dependence on~$Rr/\eps$.
The construction of our lower bound implies that for the case $m\approx n$, this polynomial dependence has to be at least $(Rr/\eps)^{1/4}$.

\paragraph{Subsequent work.}
Following the first version of our paper, improvements in the runtime of quantum SDP-solvers were obtained in~\cite{brandao2017QSDPSpeedupsLearning,apeldoorn2018ImprovedQSDPSolving}, the latter providing a runtime of $\bOt{(\sqrt{m}+\sqrt{n} \frac{Rr}{\epsilon}) s \left(\frac{Rr}{\epsilon}\right)^4}$. 
For the special case of LP-solving, where $s=1$, \cite{apeldoorn2019QAlgorithmsForZeroSumGames} further improved the runtime to $\bOt{(\sqrt{m}+\sqrt{n})\left(\frac{Rr}{\epsilon}\right)^3}$.

In a different algorithmic direction, Kerenidis and Prakash~\cite{kerenidis2018QIntPoint} recently obtained a quantum interior-point method for solving SDPs and LPs. It is hard to compare the latter algorithm to the other SDP-solvers for two reasons. First, the output of their algorithm consists only of almost-feasible solutions to the primal and dual (their algorithm has a polynomial dependence on the distance to feasibility). It is therefore not clear what their output means for the optimal value of the SDPs. Secondly, the runtime of their algorithm depends polynomially on the condition number of the matrices that the interior point method encounters, and no explicit bounds for these condition numbers are given.

Our results on implementing smooth functions of a given Hamiltonian have been extended to more general input models (block-encodings) in~\cite{gilyen2018QSingValTransf}. This recent paper builds on some of our techniques, but achieves slightly improved complexities by directly implementing the transformations without using Hamiltonian simulation as a subroutine.

Recently van Apeldoorn et al.~\cite{apeldoorn2018ConvexOptUsingQuantumOracles} and Chakrabarti et al.~\cite{chakrabarti2018QuantumConvexOpt} developed quantum algorithms for general black-box convex optimization, where one optimizes over a general convex set~$K$, and the access to $K$ is via membership and/or separation oracles. Since we work in a model where we are given access directly to the constraints defining the problem, our results are incomparable to theirs.

\paragraph{Organization.}
The paper is organized as follows.  In Section~\ref{sec:upperbounds} we start with a description of the Arora-Kale framework for SDP-solvers, and then we describe how to quantize different aspects of it to obtain a quantum SDP-solver with better dependence on $R$, $r$, and $1/\eps$ (or rather, on $Rr/\eps$) than Brand\~ao and Svore got.  In Section~\ref{sec:downside} we describe the limitations of primal-dual SDP-solvers using general oracles (not optimized for specific SDPs) that produce sparse dual solutions~$y$: if good solutions are dense, this puts a lower bound on the number of iterations needed.  In Section~\ref{sec:lowerbounds} we give our lower bounds.  A number of the proofs are relegated to the appendices:
how to classically approximate $\Tr(A_j\rho)$ (Appendix~\ref{app:trace}),
how to efficiently implement smooth functions of Hamiltonians on a quantum computer (Appendix~\ref{apx:LowWeight}),
our generalized method for minimum-finding (Appendix~\ref{app:genMinFind}),
upper and lower bounds on how efficiently we can query entries of sums of sparse matrices (Appendix~\ref{app:sparsematrixsum}),
and how to trade off the parameters $R$, $r$, and $1/\eps$ against each other (Appendix~\ref{app:reductions}).

\section{An improved quantum SDP-solver}
\label{sec:upperbounds}

Here we describe our quantum SDP-solver. In Section~\ref{sec:classicalAK} we describe the framework designed by Arora and Kale for solving semidefinite programs. As in the recent work by Brand\~ao and Svore, we use this framework to design an efficient quantum algorithm for solving SDPs. In particular, we show that the key subroutine needed in the Arora-Kale framework can be implemented efficiently on a quantum computer. Our implementation uses different techniques than the quantum algorithm of Brand\~ao and Svore, allowing us to obtain a faster algorithm. The techniques required for this subroutine are developed in Sections~\ref{sec:trCalc} and~\ref{sec:oracle}. In Section~\ref{sec:runtime} we put everything together to prove the main theorem of this section (the notation is explained below):

\begin{restatable}{theorem}{upperbound}
  \label{thm:upperbound}
  Instantiating Meta-Algorithm~\ref{alg:AKSDP} using the trace calculation algorithm from Section~\ref{sec:trCalc} and the oracle from Section~\ref{sec:oracle} (with width-bound $w:=r+1$), and using this to do a binary search for $\opt\in[-R,R]$ (using different guesses $\alpha$ for $\opt$), gives a quantum algorithm for solving SDPs of the form~\eqref{eq:SDP}, which (with high probability) produces a feasible solution $y$ to the dual program which is optimal up to an additive error $\eps$, and uses
  \[
    \bOt{\sqrt{nm} s^2\left( \frac{Rr}{\eps}\right)^{\!\!8\,} }
  \]
  queries to the input matrices and the same order of other gates.
\end{restatable}

\paragraph{Notation/Assumptions.}
We use $\log$ to denote the logarithm in base~$2$.
We denote the all-zero matrix and vector by~$0$ and the all-one vector by~$\mathbf{1}$.
Throughout we assume each element of the input matrices can be represented by a bitstring of size $\text{poly}(\log n,\log m)$.  We use $s$ as the sparsity of the input matrices, that is, the maximum number of non-zero entries in a row (or column) of any of the matrices $C, A_1,\ldots, A_m$ is $s$. Recall that for normalization purposes we assume $\nrm{A_1}, \ldots, \nrm{A_m}, \nrm{C} \leq 1$. We furthermore assume that $A_1 = I$ and $b_1 = R$, that is, the trace of primal-feasible solutions is bounded by $R$ (and hence also the trace of primal-optimal solutions is bounded by $R$).
The analogous quantity for the dual SDP~\eqref{eq:SDP2}, an upper bound on $\sum_{j=1}^m y_j$ for an optimal dual solution~$y$, will be denoted by~$r$. However, we do not add the constraint $\sum_{j=1}^m y_j \leq r$ to the dual.
We will assume $r\geq 1$. For $r$ to be well-defined we have to make the explicit assumption that the optimal solution in the dual is attained.
In Section~\ref{sec:downside} it will be necessary to work with the best possible upper bounds: we let $R^*$ be the smallest trace of an optimal solution to SDP~\eqref{eq:SDP}, and we let $r^*$ be the smallest $\ell_1$-norm of an optimal solution to the dual.  These quantities are well-defined; indeed, both the primal and dual optimum are attained: the dual optimum is attained by assumption, and due to the assumption $A_1 = I$, the dual SDP is strictly feasible, which means that the optimum in~\eqref{eq:SDP} is attained.

Unless specified otherwise, we always consider \emph{additive} error. In particular, an $\eps$-optimal solution to an SDP will be a feasible solution whose objective value is within additive error~$\eps$ of the optimum.

\paragraph{Input oracles:} We assume sparse black-box access to the elements of matrices $C, A_1,\ldots, A_m$ defined in the following way: for input $(j, k,\ell) \in (\{0\}\cup[m]) \times [n] \times [s]$ we can query the location and value of the $\ell$th non-zero entry in the $k$th row of the matrix $A_j$ (where $j=0$ would indicate the $C$ matrix).

Specifically in the quantum case, as described in~\cite{berry2015HamSimNearlyOpt}, we assume access to an oracle $O_I$ which serves the purpose of sparse access. $O_I$ calculates the $\text{index}_{A_j}: [n] \times [s] \to [n]$ function, which for input $(k,\ell)$ gives the column index of the $\ell$th non-zero element in the $k$th row of $A_j$. We assume this oracle computes the index ``in place":
\begin{equation}
  O_I\ket{j,k,\ell} = \ket{j,k,\text{index}_{A_j}(k,\ell)}.
  \label{eq:oracleind}
\end{equation}
(In the degenerate case where the  $k$th row has fewer than $\ell$ non-zero entries, $\text{index}_{A_j}(k,\ell)$ is defined to be $\ell$ together with some special symbol.) We also assume we can apply the inverse of~$O_I$.

We also need another oracle $O_M$, returning a bitstring representation of $(A_j)_{ki}$ for any $j \in \{0\}\cup[m]$ and $ k,i \in [n]$:
\begin{equation}
  O_M\ket{j,k,i,z} = \ket{j,k,i,z \oplus{(A_j)_{ki}}}.
  \label{eq:oraclemat}
\end{equation}
The slightly unusual ``in place'' definition of oracle $O_I$ is not too demanding. In particular, if instead we had an oracle that computed the non-zero entries of a row in order, then we could implement both $O_I$ and its inverse using $\log(s)$ queries (we can compute $\ell$ from $k$ and $\text{index}_{A_j}(k,\ell)$ using binary search)~\cite{berry2015HamSimNearlyOpt}.

\paragraph{Computational model:}
As our computational model, we assume a slight relaxation of the usual quantum circuit model: a classical control system that can run quantum subroutines.
We limit the classical control system so that its number of operations is at most a polylogarithmic factor bigger than the gate complexity of the quantum subroutines, i.e., if the quantum subroutines use $C$ gates, then the classical control system may not use more than $\bigO(C\,\text{polylog}(C))$ operations.

When we talk about gate complexity, we count the number of two-qubit quantum gates needed for implementation of the quantum subroutines.
Additionally, we assume for simplicity that there exists a unit-cost QRAM gate that allows us to store and retrieve qubits in a memory, by means of a swap of two registers indexed by another register:
\[
  QRAM : \ket{i,x,r_1,\dots,r_K} \mapsto \ket{i,r_i,r_1,\dots,r_{i-1},x,r_{i+1},\ldots,r_K},
\]
where the registers $r_1,\dots,r_K$ are only accessible through this gate. The QRAM gate can be seen as a quantum analogue of pointers in classical computing.
The only place where we need QRAM is in Appendix~\ref{app:sparsematrixsum}, for a data structure that allows efficient access to the non-zero entries of a sum of sparse matrices; for the special case of LP-solving it is not needed.

\subsection{The Arora-Kale framework for solving SDPs} \label{sec:classicalAK}

In this section we give a short introduction to the Arora-Kale framework for solving semidefinite programs. We refer to~\cite{arora2016CombPrimDualSDP,arora2012MultiplicativeWeightsAlg} for a more detailed description and omitted proofs.

The key building block is the Matrix Multiplicative Weights (MMW) algorithm introduced by Arora and Kale in~\cite{arora2016CombPrimDualSDP}. The MMW algorithm can be seen as a strategy for you in a game between you and an adversary. We first introduce the game.
There is a number of rounds $T$. In each round you present a density matrix $\rho$ to an adversary, the adversary replies with a loss matrix $M$ satisfying $-I \preceq M \preceq I$. After each round you have to pay $\tr{M \rho}$. Your objective is to pay as little as possible. The MMW algorithm is a strategy for you that allows you to lose not too much, in a sense that is made precise below.
In Algorithm~\ref{alg:MMW} we state the MMW algorithm, the following theorem shows the key property of the output of the algorithm.

\begin{algorithm}[ht]
  \begin{description}
  \item[Input] Parameter $\eta \leq 1$, number of rounds $T$.

  \item[Rules] In each round player $1$ (you) presents a density matrix $\rho$, player $2$ (the adversary) replies with a matrix $M$ satisfying $-I \preceq M \preceq I$.

  \item[Output] A sequence of symmetric $n \times n$ matrices $M^{(1)},\ldots, M^{(T)}$  s.t.\ $-I \preceq M^{(t)} \preceq I$, for $t \in [T]$ and a sequence of $n \times n$ psd matrices $\rho^{(1)},\ldots, \rho^{(T)}$ satisfying $\tr{\rho^{(t)}}=1$ for $t \in [T]$.

  \item[Strategy of player $1$:]
  \end{description}
  \begin{algorithmic}
    \State Take $\rho^{(1)} := I/n$
    \State In round $t$:
    \begin{enumerate}
    \item Show the density matrix $\rho^{(t)}$ to the adversary.
    \item Obtain the loss matrix $M^{(t)}$ from the adversary.
    \item Update the density matrix as follows:
      \[
        \rho^{(t+1)}:= \left.\exp\left(- \eta \sum_{\tau=1}^t M^{(\tau)}\right)\right/\tr{\exp\left(- \eta \sum_{\tau=1}^t M^{(\tau)}\right)}
      \]
    \end{enumerate}
  \end{algorithmic}
  \caption{Matrix Multiplicative Weights (MMW) Algorithm}
  \label{alg:MMW}
\end{algorithm}

\begin{theorem}[{\cite[Theorem~3.1]{arora2016CombPrimDualSDP}}]
  For every adversary, the sequence $\rho^{(1)}, \ldots, \rho^{(T)}$ of density matrices constructed using the Matrix Multiplicative Weights Algorithm~\eqref{alg:MMW} satisfies
  \[
    \sum_{t=1}^T \tr{M^{(t)} \rho^{(t)}} \leq \lambda_{\min}\left(\sum_{t=1}^T M^{(t)}\right) + \eta \sum_{t=1}^T \tr{ (M^{(t)})^2 \rho^{(t)}} + \frac{\ln(n)}{\eta}.
  \]
\end{theorem}

Arora and Kale use the MMW algorithm to construct an SDP-solver.
For that, they construct an adversary who promises to satisfy an additional condition: in each round~$t$, the adversary returns a matrix $M^{(t)}$ whose trace inner product with the density matrix $\rho^{(t)}$ is non-negative.
The above theorem shows that then, after $T$ rounds, the average of the adversary's responses satisfies the stronger condition that its smallest eigenvalue is not too negative: $\lambda_{\min}\left(\frac{1}{T} \sum_{t=1}^T M^{(t)}\right) \geq - \eta - \frac{\ln(n)}{\eta T}$. More explicitly, the MMW algorithm is used to build a vector $y \geq 0$ such that
\[
  \frac{1}{T} \sum_{t=1}^T M^{(t)} \propto \sum_{j=1}^m y_j A_j -C
\]
and $b^T y \leq \alpha$.
That is, the smallest eigenvalue of the matrix $\sum_{j=1}^m y_j A_j -C$ is only slightly below zero and $y$'s objective value is at most $\alpha$. Since $A_1=I$, increasing the first coordinate of $y$ makes the smallest eigenvalue of $\sum_j y_j A_j -C$ bigger, so that this matrix becomes psd and hence dual-feasible. By the above we know how much the minimum eigenvalue has to be shifted, and with the right choice of parameters it can be shown that this gives a dual-feasible vector $\overline{y}$ that satisfies $b^T \overline{y} \leq \alpha + \eps$. In order to present the algorithm formally, we require some definitions.

Given a candidate solution $X \succeq 0$ for the primal problem~\eqref{eq:SDP} and a parameter $\eps \geq 0$, define the polytope
\begin{align}
  \mathcal{P}_{\eps}(X) := \{ y \in \mathbb R^m:\ &b^T y  \leq \alpha, \nonumber\\
                                                  &\tr{\Big(\sum_{j=1}^m y_j A_j - C\Big) X} \geq -\eps,\nonumber\\
                                                  & y \geq 0 \}.\nonumber
\end{align}
One can verify the following:
\begin{lemma}[{\cite[Lemma~4.2]{arora2016CombPrimDualSDP}}] \label{lem:xfeas}
  If for a given candidate solution $X \succeq 0$ the polytope $\mathcal P_0(X)$ is empty, then a scaled version of $X$ is primal-feasible and of objective value at least $\alpha$.
\end{lemma}
The Arora-Kale framework for solving SDPs uses the MMW algorithm where the role of the adversary is taken by an $\eps$-approximate oracle:
\begin{algorithm}[h!]
  \begin{description}
  \item[Input] An $n \times n$ psd matrix $X$, a number
  $\alpha\in[-R,R]$, 
  and the description of an SDP as in~\eqref{eq:SDP}.

  \item[Output] Either the \textsf{Oracle}$_{\eps}$ returns a vector $y$ from the polytope $\mathcal P_\eps(X)$ or it outputs ``fail''.  \linebreak It may only output fail if $\mathcal P_0(X) = \emptyset$.
  \end{description}
  \caption{Definition of an $\eps$-approximate \textsf{Oracle}$_{\eps}$ for maximization SDPs}
  \label{alg:Oracle}
\end{algorithm}

\noindent As we will see later, the runtime of the Arora-Kale framework depends on a property of the oracle called the \emph{width}:
\begin{definition}[\emph{Width} of \textsf{Oracle}$_{\eps}$] \label{def:width}
  The \emph{width} of \textsf{Oracle}$_{\eps}$ for an SDP is the smallest $w^* \geq 0$ such that for every $X \succeq 0$ and $\alpha\in[-R,R]$, the vector $y$ returned by \textsf{Oracle}$_{\eps}$ satisfies $\nrm{\sum_{j=1}^m y_j A_j- C} \leq w^*$.
\end{definition}
In practice, the width of an oracle is not always known. However, it suffices to work with an upper bound $w \geq w^*$: as we can see in Meta-Algorithm~\ref{alg:AKSDP}, the purpose of the width is to rescale the matrix $M^{(t)}$ in such a way that it forms a valid response for the adversary in the MMW algorithm.
\begin{metaalgorithm}[ht]

  \begin{description}
  \item[Input] The input matrices and reals of SDP~\eqref{eq:SDP} and trace bound~$R$. The current guess $\alpha$ of the optimal value of the dual~\eqref{eq:SDP2}. An additive error tolerance $\eps>0$. An $\frac{\eps}{3}$-approximate oracle \textsf{Oracle}$_{\eps/3}$ as in Algorithm~\ref{alg:Oracle} with width-bound~$w$.

  \item[Output] Either ``Lower'' and a vector $\overline{y} \in \mathbb R^{m}_{+}$ feasible for~\eqref{eq:SDP2} with $b^T \overline{y} \leq \alpha+\eps$ \\
    or ``Higher'' and a symmetric $n \times n$ matrix $X$ that, when scaled suitably, is primal-feasible with objective value at least $\alpha$.
  \end{description}
  \begin{algorithmic}
    \State $T := \left\lceil \frac{9 w^2 R^2 \ln(n)}{\eps^2}\right\rceil$.
    \State $\eta := \sqrt{\frac{\ln (n)}{T}}$.
    \State $\rho^{(1)} := I/n$
    \For{ $t  = 1,\dots,T$}
    \State Run \textsf{Oracle}$_{\eps/3}$ with $X^{(t)} = R\rho^{(t)}$.
    \If{\textsf{Oracle}$_{\eps/3}$ outputs ``fail''}
    \State \Return ``Higher'' and a description of $X^{(t)}$.
    \EndIf
    \State Let $y^{(t)}$ be the vector generated by \textsf{Oracle}$_{\eps/3}$.
    \State Set $M^{(t)} = \frac{1}{w} \left( \sum_{j=1}^m y_j^{(t)} A_j - C\right)$.
    \State Define $H^{(t)} = \sum_{\tau=1}^t M^{(\tau)}$.
    \State Update the state matrix as follows: $\rho^{(t+1)}:= \exp\left(- \eta H^{(t)}\right)/\tr{\exp\left(- \eta H^{(t)}\right)}$.
    \EndFor
    \State If \textsf{Oracle}$_{\eps/3}$ does not output ``fail'' in any of the $T$ rounds, then output the dual solution
    $\overline{y} = \frac{\eps}{R} e_1 +  \frac{1}{T} \sum_{t=1}^T y^{(t)}$ where $e_1 = (1,0,\ldots, 0) \in \mathbb R^m$.
  \end{algorithmic}
  \caption{Primal-Dual Algorithm for solving SDPs}
  \label{alg:AKSDP}
\end{metaalgorithm}
The following theorem shows the correctness of the Arora-Kale primal-dual meta-algorithm for solving SDPs, stated in Meta-Algorithm~\ref{alg:AKSDP}:
\begin{theorem}[{\cite[Theorem~4.7]{arora2016CombPrimDualSDP}}]
  Suppose we are given an SDP of the form~\eqref{eq:SDP} with input matrices $A_1 = I, A_2, \ldots, A_m$ and $C$ having operator norm at most $1$, and input reals $b_1=R, b_2, \ldots, b_m$. Assume Meta-Algorithm~\ref{alg:AKSDP} does not output ``fail'' in any of the rounds,
  then the returned vector $\overline{y}$ is feasible for the dual~\eqref{eq:SDP2} with objective value at most $\alpha+\eps$.
  If \textsf{Oracle}$_{\eps/3}$ outputs ``fail'' in the $t$-th round then a suitably scaled version of $X^{(t)}$ is primal-feasible with objective value at least~$\alpha$.
\end{theorem}
The SDP-solver uses $T = \left\lceil \frac{9 w^2 R^2 \ln(n)}{\eps^2}\right\rceil$ iterations. In each iteration several steps have to be taken. The most expensive two steps are computing the matrix exponential of the matrix $- \eta H^{(t)}$ and the application of the oracle.
Note that the only purpose of computing the matrix exponential is to allow the oracle to compute the values $\tr{A_j X}$ for all $j$ and $\tr{CX}$, since the polytope depends on $X$ only through those values.
To obtain faster algorithms it is important to note, as was done already by Arora and Kale, that the primal-dual algorithm also works if we provide a (more accurate) oracle with approximations of $\tr{A_j X}$.
Let $a_j:=\tr{A_j\rho} = \tr{A_jX}/\tr{X}$ and $c:=\tr{C\rho} = \tr{CX}/\tr{X}$. Then, given a list of reals $\tilde{a}_1, \ldots, \tilde{a}_m, \tilde{c}$ and a parameter $\theta \geq 0$, such that $|\tilde{a}_j - a_j| \leq \theta$ for all $j$, and $|\tilde{c}- c| \leq \theta$, we define the polytope
\begin{align*}
  \pt(\tilde{a}_1, \ldots, \tilde{a}_m,\tilde{c}-(r+1)\theta) := \{ y \in \mathbb R^m:\ &  b^T y  \leq \alpha,\\&\sum_{j=1}^m y_j \leq r, \\
                                                                &\sum_{j=1}^m \tilde{a}_j y_j  \geq  \tilde{c} - (r+1)\theta\\&y \geq 0\}.
\end{align*}
For convenience we will denote $\tilde{a} = (\tilde{a}_1, \ldots, \tilde{a}_m)$ and $\cp := \tilde{c}- (r+1) \theta$. Notice that $\pt$ also contains a new type of constraint: $\sum_j y_j \leq r$. Recall that $r$ is defined as a positive real such that there exists an optimal solution $y$ to SDP~\eqref{eq:SDP2} with $\nrm{y}_1 \leq r$. Hence, using that $\mathcal P_0(X)$ is a \emph{relaxation} of the feasible region of the dual (with bound $\alpha$ on the objective value), we may restrict our oracle to return only such~$y$:
\[
  \mathcal P_0(X) \neq \emptyset \Rightarrow \mathcal P_0(X) \cap \{y \in \mathbb R^m: \sum_{j=1}^m y_j \leq r\} \neq \emptyset.
\]
The benefit of this restriction is that an oracle that always returns a vector with bounded $\ell_1$-norm automatically has a width $w^* \leq r+1$, due to the assumptions on the norms of the input matrices. The downside of this restriction is that the analogue of Lemma~\ref{lem:xfeas} does not hold for $\mathcal P_0(X) \cap \{y \in \mathbb R^m: \sum_{j} y_j \leq r\}$.\footnote{Using several transformations of the SDP, from Appendix~\ref{app:reductions} and Lemma 2 of~\cite{brandao2016QSDPSpeedup}, one can show that there is a way to remove the need for this restriction. Hence, after these modifications, if for a given candidate solution $X \succeq 0$ the oracle outputs that the set $\mathcal P_0(X)$ is empty, then a scaled version of $X$ is primal feasible for this new SDP, with objective value at least $\alpha$. This scaled version of $X$ can be modified to a near-feasible solution to the original SDP (it will be psd, but it might violate the linear constraints a little bit) with nearly the same objective value.}

The following shows that an oracle that always returns a vector $y \in \pt(\tilde{a},c')$ if one exists, is a $4Rr\theta$-approximate oracle as defined in Algorithm~\ref{alg:Oracle}.

\begin{lemma} \label{lem:approxP}
  Let $\rho = X/\mathrm{Tr}(X)$ where $\mathrm{Tr}(X)\leq R$. Let $\tilde{a}_1, \ldots, \tilde{a}_m$ and $\tilde{c}$ be $\theta$-approximations of $\tr{A_1 \rho}, \ldots, \tr{A_m \rho}$ and $\tr{C\rho}$ respectively. Then the following holds:
  \[
    \mathcal P_0(X) \cap \{y \in \mathbb R^m: \sum_{j=1}^m y_j \leq r\} \subseteq \pt(\tilde{a},c') \subseteq \mathcal P_{4Rr\theta}(X).
  \]
\end{lemma}
\begin{proof}
  First, suppose $y \in \mathcal P_{0}(X) \cap \{y \in \mathbb R^m: \sum_{j} y_j \leq r\}$.
  We then have $y \in \pt(\tilde{a},c')$ because
  \[
    \sum_{j=1}^m \tilde{a}_j y_j - \tilde{c} 
    \geq \sum_{j=1}^m (\tilde{a}_j - \tr{A_j \rho}) y_j - (\tilde{c} -\tr{C \rho}) 
    \geq - \theta \nrm{y}_1 - \theta \geq - (r+1) \theta,
  \]
  where in the first inequality we subtracted $\sum_{j=1}^m \tr{A_j \rho} y_j -\tr{C \rho}\geq 0$. 

  Next, suppose $y \in \pt(\tilde{a},\cp)$. We show that $y \in \mathcal P_{4Rr \theta}(X)$. Indeed, since $|\tr{A_j\rho} - \tilde{a}_j| \leq \theta$ we have
  \[
    \tr{\!\!\left(\sum_{j=1}^m y_j A_j -C\right)\rho\!} \geq \left(\sum_{j=1}^m \tilde{a}_jy_j + \tilde{c} \right) -(r+1)\theta \geq -(2+r + \nrm{y}_1) \theta \geq\! -4r \theta  \]
  where the last inequality used our assumptions $r\geq 1$ and $\nrm{y}_1\leq r$.
  Hence
  \[
    \tr{\left(\sum_{j=1}^m y_j A_j -C\right)X} \geq -4r\tr{X} \theta \geq -4Rr\theta.
  \]
  For the latter inequality we used $\tr{X} \leq R$.
\end{proof}

We have now seen the Arora-Kale framework for solving SDPs. To obtain a quantum SDP-solver it remains to provide a quantum oracle subroutine. By the above discussion it suffices to set  $\theta = \eps/(12Rr)$ and to use an oracle that is based on $\theta$-approximations of $\tr{A\rho}$ (for $A \in \{A_1, \ldots, A_m, C\}$), since with that choice of $\theta$ we have $\mathcal P_{4Rr \theta}(X)=\mathcal P_{\eps/3}(X)$.
In the section below we first give a quantum algorithm for approximating $\tr{A\rho}$ efficiently (see also Appendix~\ref{app:trace} for a classical algorithm). Then, in Section~\ref{sec:oracle}, we provide an oracle using those estimates. The oracle will be based on a simple geometric idea and can be implemented both on a quantum computer and on a classical computer (of course, resulting in different runtimes). In Section~\ref{sec:runtime} we conclude with an overview of the runtime of our quantum SDP-solver.
We want to stress that our solver is meant to work for any SDP. In particular, our oracle does not use the structure of a specific SDP. As we will show in Section~\ref{sec:downside}, any oracle that works for all SDPs necessarily has a large width-bound. To obtain quantum speedups for a \emph{specific} class of SDPs it will be necessary to develop oracles tuned to that problem, we view this as an important direction for future work. Recall from the introduction that Arora and Kale also obtain fast classical algorithms for problems such as MAXCUT by developing specialized oracles.

\subsection{Approximating the expectation value \texorpdfstring{$\tr{A\rho}$}{traces} using a quantum algorithm}\label{sec:trCalc}

In this section we give an efficient quantum algorithm to approximate quantities of the form $\tr{A\rho}$. We are going to work with Hermitian matrices $A,H\in\mathbb{C}^{n\times n}$, such that $\rho$ is the Gibbs state $e^{-H}/\tr{e^{-H}}$. Note the analogy with quantum physics: in physics terminology $\tr{A\rho}$ is simply called the ``expectation value" of $A$ for a quantum system in a thermal state corresponding to $H$.

The general approach is to separately estimate $\tr{A e^{-H}}$ and $\tr{e^{-H}}$, and then to use the ratio of these estimates as an approximation of $\tr{A\rho}=\tr{A e^{-H}}/\tr{e^{-H}}$. Both estimations are done using state preparation to prepare a pure state with a flag, such that the probability that the flag is~0 is proportional to the quantity we want to estimate, and then to use amplitude estimation to estimate that probability.
Below in Section~\ref{ssec:generalapproach} we first describe the general approach.
In Section~\ref{sec:lptrace} we then instantiate this for the special case where all matrices are diagonal, which is the relevant case for LP-solving. In Section~\ref{sec:estTrArhogeneral} we handle the general case of arbitrary matrices (needed for SDP-solving); the state-preparation part will be substantially more involved there, because in the general case we need not know the diagonalizing bases for $A$ and~$H$, and $A$ and $H$ may not be simultaneously diagonalizable.

\subsubsection{General approach}\label{ssec:generalapproach}

To start, consider the following lemma about the multiplicative approximation error of a ratio of two real numbers that are given by multiplicative approximations:
\begin{lemma}\label{lemma:trTogether}
  Let $0\leq \theta \leq 1$ and let $\alpha, \tilde{\alpha}, Z, \tilde{Z}$ be positive real numbers
  such that $|\alpha - \tilde{\alpha}| \leq \alpha\theta / 3$ and $|Z - \tilde{Z}| \leq Z\theta / 3$. Then
  \[
    \left|\frac{\alpha}{Z}-\frac{\tilde{\alpha}}{\tilde{Z}}\right| \leq \theta\frac{\alpha}{Z}
  \]
\end{lemma}
\begin{proof} The inequality can be proven as follows
  \begin{align*}
    \left|\frac{\alpha}{Z}-\frac{\tilde{\alpha}}{\tilde{Z}}\right|
    &= \left|\frac{\alpha\tilde{Z}\kern-0.2mm-\kern-0.2mm\tilde{\alpha}Z}{Z\tilde{Z}}\right|
      = \left|\frac{\alpha\tilde{Z}\kern-0.2mm-\kern-0.2mm\alpha Z\kern-0.2mm+\kern-0.2mm\alpha Z\kern-0.2mm-\kern-0.2mm\tilde{\alpha}Z}{Z\tilde{Z}}\right|
    \leq \left|\frac{\alpha\tilde{Z}\kern-0.2mm-\kern-0.2mm\alpha Z}{Z\tilde{Z}}\right|+\left|\frac{\alpha Z\kern-0.2mm-\kern-0.2mm\tilde{\alpha}Z}{Z\tilde{Z}}\right|
      \leq\frac{\alpha\theta}{3\tilde{Z}}+\frac{\alpha \theta}{3\tilde{Z}}\leq\theta\frac{\alpha}{Z}
  \end{align*}
where the last step used $\tilde{Z}\geq \frac{2}{3}Z$.
\end{proof}

\begin{corollary}\label{col:split}
  Let $A$ be such that $\nrm{A}\leq 1$. A multiplicative $\frac\theta9$-approximation of both  $\tr{\frac{I}{4}e^{-H}}$ and $\tr{\frac{I+A/2}{4}e^{-H}}$ suffices to get an additive $\theta$-approximation of $\frac{\tr{Ae^{-H}}}{\tr{e^{-H}}}$.
\end{corollary}
\begin{proof}
  According to Lemma~\ref{lemma:trTogether} by dividing the two multiplicative approximations we get \[
    \frac{\theta}{3}\frac{\tr{\frac{I+A/2}{4}e^{-H}}}{\tr{\frac{I}{4}e^{-H}}}
    = \frac{\theta}{3}\left(1+\frac{\tr{\frac{A}{2}e^{-H}}}{\tr{e^{-H}}}\right)
    \leq \frac{\theta}{3}\left(1+\frac{\nrm{A}}{2}\right)
    \leq \theta/2,
  \]
  i.e., an additive $\theta/2$-approximation of
  $$
  1+\frac{\tr{\frac{A}{2}e^{-H}}}{\tr{e^{-H}}},
  $$
  which yields an additive $\theta$-approximation to $\tr{Ae^{-H}}/\tr{e^{-H}}$.
\end{proof}
It thus suffices to separately approximate both quantities from the corollary. Notice that both are of the form $\tr{\frac{I+A/2}{4}e^{-H}}$, the first with the actual $A$, the second with $A=0$.
Furthermore, a multiplicative $\theta/9$-approximation to both can be achieved by approximating both up to an additive error $\theta \tr{e^{-H}} / 72$, since $\tr{\frac{I}{8} e^{-H}}\leq \tr{\frac{I+A/2}{4}e^{-H}}$. 

For now, let us assume we can construct a unitary $U_{A,H}$ such that if we apply it to the state~$\ket{0\ldots 0}$ then we get a probability $\frac{\tr{(I+A/2)e^{-H}}}{4n}$ of outcome~0 when measuring the first qubit. That is:
\[
  \nrm{(\bra{0}\otimes I)U_{A,H}\ket{0\ldots 0}}^2 = \frac{\tr{(I+A/2)e^{-H}}}{4n}.
\]
(To clarify the notation: if $\ket{\psi}$ is a 2-register state, then $(\bra{0}\otimes I)\ket{\psi}$ is the (unnormalized) state in the 2nd register that results from projecting on $\ket{0}$ in the 1st register.)

In practice we will not be able to construct such a $U_{A,H}$ exactly, instead we will construct a $\tilde{U}_{A,H}$ that yields a sufficiently close approximation of the correct probability.
When we have access to such a unitary, the following lemma allows us to use amplitude estimation to estimate the probability and hence $\tr{\frac{I+A/2}{4}e^{-H}}$ up to the desired error.

\begin{lemma}\label{lem:ampest}
  Suppose we have a unitary $U$ acting on $q$ qubits such that $U\ket{0\ldots 0}=\ket{0}\ket{\psi}+\ket{\Phi}$ with $(\bra{0}\otimes I)\ket{\Phi}=0$ and $\nrm{\psi}^2  =p\geq p_{\min}$ for some known bound $p_{\min}$. Let $\mu \in(0,1]$ be the allowed multiplicative error in our estimation of $p$. Then with $\bigO\left(\frac{1}{\mu\sqrt{p_{\min}}}\right)$ uses of $U$ and $U^{-1}$ and using $\bigO\left(\frac{q}{\mu\sqrt{p_{\min}}}\right)$ gates on the $q$ qubits and some ancilla qubits, we obtain a $\tilde{p}$ such that $|p-\tilde{p}|\leq \mu p$ with probability at least $4/5$.
\end{lemma}

\begin{proof}
  We use the amplitude-estimation algorithm of~\cite[Theorem~12]{brassard2002AmpAndEst} with $M$ applications of $U$ and $U^{-1}$. This provides an estimate $\tilde{p}$ of $p$, that with probability at least $8/\pi^2>4/5$ satisfies
  \begin{equation*}
    |p-\tilde{p}|
    \leq2\pi\frac{\sqrt{p(1-p)}}{M}+\frac{\pi^2}{M^2}
    \leq\frac{\pi}{M}\left(2\sqrt{p}+\frac{\pi}{M}\right).
  \end{equation*}
  Choosing $M$ the smallest power of $2$ such that $M \geq 3\pi/(\mu\sqrt{p_{\min}})$, with probability at least $4/5$ we get
  \begin{equation*}
    |p-\tilde{p}|
    \leq\mu\frac{\sqrt{p_{\min}}}{3}\left(2\sqrt{p}
      +\mu\frac{ \sqrt{p_{\min}}}{3}\right)
    \leq\mu\frac{\sqrt{p}}{3}\left(3\sqrt{p}\right)
    \leq \mu p.
  \end{equation*}
  The $q$ factor in the gate complexity comes from the implementation of the amplitude amplification steps needed in amplitude estimation. The gate complexity of the whole amplitude-estimation procedure is dominated by this contribution, proving the final gate complexity.
\end{proof}

\begin{corollary}\label{col:main22}
  Suppose we are given the positive numbers $z\leq \tr{e^{-H}}$, $\theta\in(0,1]$, and unitary circuits $\tilde{U}_{A',H}$ for $A'=0$ and $A'=A$ with $\nrm{A}\leq 1$, each acting on at most $q$ qubits such that
  $$
  \left|\nrm{(\bra{0}\otimes I)\tilde{U}_{A',H}\ket{0\ldots 0}}^2-\frac{\tr{(I+A'/2)e^{-H}}}{4n}\right| \leq \frac{\theta z}{144n}.
  $$
  Applying the procedure of Lemma~\ref{lem:ampest} to $\tilde{U}_{A',H}$ (both for $A'=0$ and for $A'=A$) with $p_{\min}=\frac{z}{9n}$ and $\mu=\theta/19$, and combining the results using Corollary~\ref{col:split} yields an additive $\theta$-approximation of $\tr{A\rho}$ with probability at least $4/5$.
  The procedure uses
  \[
    \bigO\left(\frac{1}{\theta}\sqrt{\frac{n}{z}}\right)
  \]
  applications of $\tilde{U}_{A,H}$, $\tilde{U}_{0,H}$ and their inverses, and $\bigO\left(\frac{q}{\theta}\sqrt{\frac{n}{z}}\right)$ additional gates.
\end{corollary}
\begin{proof}
  First note that since $I+A'/2 \succeq I/2$, we have
  \begin{equation*}
    t:=\frac{\tr{(I+A'/2)e^{-H}}}{4n} \geq \frac{\tr{e^{-H}}}{8n},
  \end{equation*}
  and thus
  \begin{equation}\label{eq:pDiff}
    \left| \nrm{(\bra{0}\otimes I)\tilde{U}_{A',H}\ket{0\ldots 0}}^2 - t\right|
    \leq \frac{\theta z }{144n}
    \leq \frac{\theta}{18}\cdot\frac{\tr{e^{-H}}}{8 n}
    \leq \frac{\theta}{18} t
    \leq \frac{t}{18}.
  \end{equation}
  Therefore
  \[
    \nrm{(\bra{0}\otimes I)\tilde{U}_{A',H}\ket{0\ldots 0}}^2
    \geq \left(1-\frac{1}{18}\right)t
    \geq \left(1-\frac{1}{18}\right)\frac{\tr{e^{-H}}}{8n}
    > \frac{\tr{e^{-H}}}{9 n}
    \geq \frac{z}{9 n}=:p_{\min}.
  \]
  Also by \eqref{eq:pDiff} we have
  \begin{equation*}
    \nrm{(\bra{0}\otimes I)\tilde{U}_{A',H}\ket{0\ldots 0}}^2
    \leq \left(1+\frac{\theta}{18}\right) t
    \leq \frac{19}{18} t.
  \end{equation*}
  Therefore using Lemma~\ref{lem:ampest} and setting $\mu=\theta/19$, with probability at least $4/5$ we get a $\tilde{p}$ satisfying
  \begin{equation}\label{eq:ptDiff}
    \left| \tilde{p} - \nrm{(\bra{0}\otimes I)\tilde{U}_{A',H}\ket{0\ldots 0}}^2 \right|
    \leq \frac{\theta}{19}\cdot\nrm{(\bra{0}\otimes I)\tilde{U}_{A',H}\ket{0\ldots 0}}^2
    \leq \frac{\theta}{18} t.
  \end{equation}
  By combining \eqref{eq:pDiff}-\eqref{eq:ptDiff} and using the triangle inequality we get
  \begin{equation*}
    \left| t - \tilde{p}\right|
    \leq \frac{\theta}{9}t,
  \end{equation*}
  so that Corollary~\ref{col:split} can indeed be applied.
  The complexity statement follows from Lemma~\ref{lem:ampest} and our choices of $p_{\min}$ and $\mu$.
\end{proof}
Notice the $1/\sqrt{z}\geq 1/\sqrt{\tr{e^{-H}}}$ factor in the complexity statement of the last corollary. To make sure this factor is not too large, we would like to ensure $\tr{e^{-H}}=\Omega(1)$. This can be achieved by substituting $H_+ = H-\lambda_{\min}I$, where $\lambda_{\min}$ is the smallest eigenvalue of $H$. It is easy to verify that this will not change the value $\tr{Ae^{-H}/\tr{e^{-H}}}$.

It remains to show how to compute $\lambda_{\min}$ and how to apply $\tilde{U}_{A,H}$. Both of these steps are considerably easier in the case where all matrices are diagonal, so we will consider this case first.

\subsubsection{The special case of diagonal matrices -- for LP-solving}\label{sec:lptrace}
In this section we consider diagonal matrices, assuming oracle access to $H$ of the following form:
\[
  O_H\ket{i}\ket{z} = \ket{i}\ket{z\oplus H_{ii}}
\]
and similarly for $A$. Notice that this kind of oracle can easily be constructed from the general sparse matrix oracle~\eqref{eq:oraclemat} that we assume access to.

\begin{lemma}\label{lem:diagU}
  Let $A,H \in \mathbb{R}^{n\times n}$ be diagonal matrices such that $\nrm{A}\leq 1$ and $H\succeq 0$, and let $\mu > 0$ be an error parameter. Then there exists a unitary $\tilde{U}_{A,H}$ such that
  \[
    \left|  \nrm{ (\bra{0}\otimes I)\tilde{U}_{A,H}\ket{0\ldots0} }^2  - \tr{\frac{I+A/2}{4n}e^{-H}}\right| \leq \mu,
  \]
 which uses $1$ quantum query to $A$ and $H$ and $\bigO(\log^{\bigO(1)}(1/\mu) + \log(n))$ other gates.
\end{lemma}

\begin{proof}
   First we prepare the state $\sum_{i=1}^{n}\ket{i}/\sqrt{n}$ with $\bigOn{\log(n)}$ one- and two-qubit gates. If $n$ is a power of~2 we do this by applying $\log_2(n)$ Hadamard gates on $\ket{0}^{\otimes\log_2(n)}$; in the general case it is still possible to prepare the state $\sum_{i=1}^{n}\ket{i}/\sqrt{n}$ with $\bigOn{\log(n)}$ two-qubit gates, for example by preparing the state $\sum_{i=1}^{k}\ket{i}/\sqrt{k}$ for $k=2^{\lceil\log_2(n)\rceil}$ and then using (exact) amplitude amplification in order to remove the $i>n$ from the superposition.
  
  Then we query the diagonal values of $H$ and $A$ to get the state $\sum_{i=1}^{n}\ket{i}\ket{H_{ii}}\ket{A_{ii}}/\sqrt{n}$. Using these binary values we apply a finite-precision arithmetic circuit to prepare
  $$
  \frac{1}{\sqrt{n}} \sum_{i=1}^{n} \ket{i}\ket{H_{ii}}\ket{A_{ii}}\ket{\beta_i}   \text{, where } \beta_i:=\arcsin\left(\sqrt{\frac{1+A_{ii}/2}{4}e^{-H_{ii}}+\delta_i}\right)/\pi \text{, and } |\delta_i|\leq \mu.
  $$
  The error $\delta_i$ is because we only write down a finite number of bits $b_1.b_2b_3\dots b_{\log(8/\mu)}$. Due to our choice of $A$ and $H$, we know that $\beta_i$ lies in $[0,1]$.
  We proceed by first adding an ancilla qubit initialized to $\ket{1}$ in front of the state, then we apply $\log(8/\mu)$ controlled rotations to this qubit: for each $b_j=1$ we apply a rotation by angle $\pi2^{-j}$. In other words, if $b_1 = 1$, then we rotate $\ket{1}$ fully to $\ket{0}$. If $b_2 = 1$, then we rotate halfway, and we proceed further by halving the angle for each subsequent bit.
  We will end up with the state:
  $$
  \frac{1}{\sqrt{n}}\sum_{i=1}^{n} \left(\sqrt{\frac{1+A_{ii}/2}{4}e^{-H_{ii}}+\delta_i}\,\ket{0}+\sqrt{1-\frac{1+A_{ii}/2}{4}e^{-H_{ii}}-\delta_i}\,\ket{1}\right) \ket{i}\ket{A_{ii}}\ket{H_{ii}}\ket{\beta_i}.
  $$
  It is now easy to see that the squared norm of the $\ket{0}$-part of this state is as required:
  $$
  \nrm{\frac{1}{\sqrt{n}}\sum_{i=1}^{n}\sqrt{\frac{1+A_{ii}/2}{4e^{H_{ii}}}+\delta_i}\,\ket{i}}^2
  =\frac{1}{n}\sum_{i=1}^{n}\left(\frac{1+A_{ii}/2}{4}e^{-H_{ii}}+\delta_i\right)
  =\frac{\tr{(I+A/2)e^{-H}}}{4n}+\sum_{i=1}^{n}\frac{\delta_i}{n},
  $$
  which is an additive $\mu$-approximation since $\left| \sum_{i=1}^{n}\frac{\delta_i}{n} \right| \leq \mu$.
\end{proof}
\begin{corollary}\label{col:traceDiagCalc}
 For $A,H \in \mathbb{R}^{n\times n}$ diagonal matrices with $\nrm{A}\leq 1$, an additive $\theta$-approximation of
  \[
    \tr{A\rho} = \frac{\tr{Ae^{-H}}}{\tr{e^{-H}}}
  \]
  can be computed using $\bigO\left(\frac{\sqrt{n}}{\theta}\right)$ queries to $A$ and $H$, and $\bOt{\frac{\sqrt{n}}{\theta}}$ other gates.
\end{corollary}
\begin{proof}
  Since $H$ is a diagonal matrix, its eigenvalues are exactly its diagonal entries. Using the quantum minimum-finding algorithm of D{\"u}rr and H{\o}yer~\cite{durr1996QMinimumFinding} one can find (with high success probability) the minimum $\lambda_{\min}$ of the diagonal entries using $\bigO(\sqrt{n})$ queries to the matrix elements.
  Applying Lemma~\ref{lem:diagU} and Corollary~\ref{col:main22} to $H_+ = H-\lambda_{\min}I$, with $z=1$, gives the stated bound.
\end{proof}

\subsubsection{General case -- for SDP-solving}\label{sec:estTrArhogeneral}

In this section we will extend the ideas from the last section to non-diagonal matrices. There are a few complications that arise in this more general case. These mostly follow from the fact that we now do not know the eigenvectors of $H$ and $A$, which were the basis states before, and that these eigenvectors might not be the same for both matrices. For example, to find the minimal eigenvalue of $H$, we can no longer simply minimize over its diagonal entries. To solve this, in Appendix~\ref{app:genMinFind} we develop new techniques that generalize minimum-finding.

Furthermore, the unitary $\tilde{U}_{A,H}$ in the LP case could be seen as applying the operator
$$
\sqrt{\frac{I+A/2}{4}e^{-H}}
$$
to a superposition of its eigenvectors. This is also more complicated in the general setting, due to the fact that the eigenvectors are no longer the basis states. In Appendix~\ref{apx:LowWeight} we develop general techniques to apply smooth functions of Hamiltonians to a state. Among other things, this will be used to create an efficient purified Gibbs sampler.

Our Gibbs sampler uses similar methods to the work of Chowdhury and Somma~\cite{chowdhury2016QGibbsSampling} for achieving logarithmic dependence on the precision. However, the result of~\cite{chowdhury2016QGibbsSampling} cannot be applied to our setting, because it implicitly assumes access to an oracle for $\sqrt{H}$ instead of $H$. Although their paper describes a way to construct such an oracle, it comes with a large overhead: they construct an oracle for $\sqrt{H'}=\sqrt{H + \nu I}$, where $\nu\in\mathbb{R}_+$ is some possibly large positive number. This shifting can have a huge effect on $Z'=\tr{e^{-H'}}=e^{-\nu}\tr{e^{-H}}$, which can be prohibitive due to the $\sqrt{1/Z'}$ factor in the runtime, which blows up exponentially in $\nu$.

In the following lemma we show how to implement $\tilde{U}_{A,H}$ using the techniques we developed in Appendix~\ref{apx:LowWeight}.
\begin{lemma}\label{lemma:trPreEst}
  Let $A,H\in \mathbb{C}^{n\times n}$ be Hermitian matrices such that $\nrm{A}\leq 1$ and $I\preceq H\preceq K I$ for a known $K\in \mathbb{R}_+$. Assume $A$ is $s$-sparse and $H$ is $d$-sparse with $s\leq d$. Let $\mu > 0$ be an error parameter. Then there exists a unitary $\tilde{U}_{A,H}$ such that
  \[
    \left|  \nrm{ (\bra{0}\otimes I)\tilde{U}_{A,H}\ket{0\ldots0} }^2  - \tr{\frac{I+A/2}{4n}e^{-H}}\right| \leq \mu
  \]
  that uses $\bOt{Kd}$ queries to $A$ and $H$, and the same order of other gates.
\end{lemma}
\begin{proof}
  The basic idea is that we first prepare a maximally entangled state $\sum_{i=1}^{n}\ket{i}\ket{i}/\sqrt{n}$, and then apply the (norm-decreasing) maps $e^{-H/2}$ and $\sqrt{\frac{I+A/2}{4}}$ to the first register. Note that we can assume without loss of generality that $\mu\leq 1$, otherwise the statement is trivial.

  Let $\tilde{W}_0=\left(\bra{0}\otimes I\right)\tilde{W}\left(\ket{0}\otimes I\right)$ be a $\mu/5$-approximation of the map $e^{-H/2}$ (in operator norm) implemented by using Theorem~\ref{thm:emH}, and let $\tilde{V}_0=\left(\bra{0}\otimes I\right)\tilde{V}\left(\ket{0}\otimes I\right)$ be a $\mu/5$-approximation of the map $\sqrt{\frac{I+A/2}{4}}$ implemented by using Theorem~\ref{thm:Taylor}. We define $\tilde{U}_{A,H}:=\tilde{V}\tilde{W}$, noting that there is a hidden $\otimes I$ factor in both $\tilde{V}$ and $\tilde{W}$ corresponding to their respective ancilla qubit.  As in the linear programming case, we are interested in the probability~$p$ of measuring outcome $00$ in the first register (i.e., the two ``flag'' qubits) after applying $\tilde{U}_{A,H}$.  We will analyze this in terms of these operators below. 
  \begin{align}
    p&:=\nrm{\left(\bra{00}\otimes I\right) \tilde{U}_{A,H}\left(\ket{00}\otimes I\right)\sum_{i=1}^{n}\frac{\ket{i}\ket{i}}{\sqrt{n}}}^2 \nonumber\\
     &=\nrm{\tilde{V}_0\tilde{W}_0\sum_{i=1}^{n}\frac{\ket{i}\ket{i}}{\sqrt{n}}}^2\nonumber\\
     &=\frac{1}{n}\sum_{i=1}^{n}\bra{i}\tilde{W}_0^\dagger\tilde{V}_0^\dagger\tilde{V}_0\tilde{W}_0\ket{i}\nonumber\\
     &=\frac{1}{n}\tr{\tilde{W}_0^\dagger\tilde{V}_0^\dagger\tilde{V}_0\tilde{W}_0}\nonumber\\
     &=\frac{1}{n}\tr{\tilde{V}_0^\dagger\tilde{V}_0\tilde{W}_0\tilde{W}_0^\dagger}\label{eq:exactTrace}
  \end{align}
Now we show that the above quantity is a good approximation of
  \begin{equation}\label{eq:approxTrace}
  \frac{1}{n}\tr{\frac{I+A/2}{4}e^{-H}}. 
  \end{equation}
For this we show that $\tilde{V}_0^\dagger\tilde{V}_0\approx (I+A/2)/4$ and $\tilde{W}_0\tilde{W}_0^\dagger\approx e^{-H}$. To see this, first note that for all matrices $B,\tilde{B}$ with $\nrm{B}\leq 1$, we have
\begin{align*}
  \nrm{B^\dagger B - \tilde{B}^\dagger \tilde{B}}& = \nrm{(B^\dagger-\tilde{B}^\dagger)B + B^\dagger(B-\tilde{B}) - (B^\dagger-\tilde{B}^\dagger)(B-\tilde{B})}\\
  & \leq    \nrm{(B^\dagger-\tilde{B}^\dagger)B} + \nrm{B^\dagger(B-\tilde{B})} +\nrm{(B^\dagger-\tilde{B}^\dagger)(B-\tilde{B})}\\
  & \leq    \nrm{B^\dagger-\tilde{B}^\dagger}\nrm{B} + \nrm{B^\dagger}\nrm{B-\tilde{B}} +\nrm{B^\dagger-\tilde{B}^\dagger}\nrm{B-\tilde{B}}\\
  & \leq 2\nrm{B-\tilde{B}}+\nrm{B-\tilde{B}}^2.
\end{align*}
  Since $\mu\leq 1$, and hence $2\mu/5 + (\mu /5)^2\leq \mu/2$, this implies (with $B=e^{-H/2}$ and $\tilde{B}=\tilde{W}_0^\dagger$) that $\nrm{e^{-H} - \tilde{W}_0\tilde{W}_0^\dagger}\leq \mu / 2$, and also (with $B=\sqrt{(I+A/2)/4}$ and $\tilde{B}=\tilde{V}_0$) $\nrm{(I\!+\!A/2)/4 - \tilde{V}_0^\dagger\tilde{V}_0}\leq \mu/2$. Let $\nrm{\cdot}_1$ denote the trace norm (a.k.a.\ Schatten 1-norm). Note that for all $C,D,\tilde{C},\tilde{D}$:
  \begin{align*}
    \left|\tr{CD}-\tr{\tilde{C}\tilde{D}}\right|
    &\leq \nrm{CD-\tilde{C}\tilde{D}}_1\\
    &= \nrm{ (C-\tilde{C})D + C(D-\tilde{D})-(C-\tilde{C})(D-\tilde{D}) }_1\\
    &\leq \nrm{(C-\tilde{C})D}_1+\nrm{C(D-\tilde{D})}_1+\nrm{(C-\tilde{C})(D-\tilde{D})}_1\\
    & \leq \nrm{C-\tilde{C}}\nrm{D}_1+\nrm{D-\tilde{D}}\left(\nrm{C}_1+\nrm{C-\tilde{C}}_1\right).
  \end{align*}
  Which, in our case (setting $C=(I+A/2)/4$, $D=e^{-H}$, $\tilde{C}=\tilde{V}_0^\dagger\tilde{V}_0$, and $\tilde{D}=\tilde{W}_0\tilde{W}_0^\dagger$) implies that
  \[
    \left|\tr{\left(I+A/2\right) e^{-H}/4}-\tr{\tilde{V}_0^\dagger\tilde{V}_0\tilde{W}_0\tilde{W}_0^\dagger}\right|	\leq (\mu/2) \tr{e^{-H}} +(\mu/2)(1/2+\mu / 2)n.
  \]
  Dividing both sides by $n$ and using equation~\eqref{eq:exactTrace} then implies
  \begin{align*}
    \left|\tr{\left(I+A/2\right) e^{-H}}/(4n)-p\right| & \leq \frac{\mu}{2} \frac{\tr{e^{-H}}}{n}+\frac{\mu}{2}\left(\frac{1}{2}+\frac{\mu}{2}\right) \\
                                                       &\leq \frac{\mu}{2}+ \frac{\mu}{2}\\
                                                       &= \mu.
  \end{align*}
  This proves the correctness of $\tilde{U}_{A,H}$. It remains to show that the complexity statement holds.
  To show this we only need to specify how to implement the map $\sqrt{\frac{I+A/2}{4}}$ using Theorem~\ref{thm:Taylor} (see Appendix~\ref{apx:LowWeight}), since the map $e^{H/2}$ is already dealt with in Theorem~\ref{thm:emH}. To use Theorem~\ref{thm:Taylor}, we choose $x_0:=0$, $K:=1$ and $r:=1$, since $\nrm{A}\leq 1$. Observe that $\sqrt{\frac{1+x/2}{4}}=\frac{1}{2}\sum_{k=0}^{\infty}\binom{1/2}{k}\left(\frac{x}{2}\right)^k$ whenever $|x|\leq 1$. Also let $\delta=1/2$, so $r+\delta=\frac{3}{2}$ and $\frac{1}{2}\sum_{k=0}^{\infty}\left|\binom{1/2}{k}\right|\left(\frac{3}{4}\right)^k\leq 1=:B$.
  Recall that $\tilde{V}$ denotes the unitary that Theorem~\ref{thm:Taylor} constructs.
  Since we choose the precision parameter to be $\mu/5=\Theta(\mu)$, Theorem~\ref{thm:Taylor} shows $\tilde{V}$ can be implemented using $\bigO\left(d\log^2\left(1/\mu\right)\right)$ queries and $\bigO\left(d\log^2\left(1/\mu\right)\left[\log(n)+\log^{2.5}\left(1/\mu\right)\right]\right)$ gates. This cost is negligible compared to the cost of our implementation of $e^{-H/2}$ with $\mu/5$ precision: Theorem~\ref{thm:emH} uses $\bigO\left(Kd\log^2\left(K/\mu\right)\right)$ queries and $\bigO\left(Kd\log^2\left(Kd/\mu\right)\left[\log(n)+\log^{2.5}\left(Kd/\mu\right)\right]\right)$ gates to implement~$\tilde{W}$.
\end{proof}
\begin{corollary}\label{col:traceCalc}
  Let $A,H\in \mathbb{C}^{n\times n}$ be Hermitian matrices such that $\nrm{A}\leq 1$ and $\nrm{H}\leq K$ for a known bound $K\in \mathbb{R}_+$. Assume $A$ is $s$-sparse and $H$ is $d$-sparse with $s\leq d$. An additive $\theta$-approximation of
  \[
    \tr{A\rho} = \frac{Ae^{-H}}{\tr{e^{-H}}}
  \]
  can be computed using $\bOt{\frac{\sqrt{n}dK}{\theta}}$ queries to $A$ and $H$, while using the same order of other gates.
\end{corollary}
\begin{proof}
  Start by computing an estimate $\tilde{\lambda}_{\min}$ of $\lambda_{\min}(H)$, the minimum eigenvalue of $H$, up to additive error $\eps = 1/2$ using Lemma~\ref{lemma:normEst}.
  We define $H_+ := H-(\tilde{\lambda}_{\min}-3/2)I$, so that $I\preceq H_+$ but $2I\nprec H_+$.

  Applying Lemma~\ref{lemma:trPreEst} and Corollary~\ref{col:main22} to $H_+$ with $z=e^{-2}$ gives the stated bound.
\end{proof}

\subsection{An efficient 2-sparse oracle} \label{sec:oracle}
To motivate the problem below, recall from the end of Section~\ref{sec:classicalAK} that $\tilde{a}_j$ is an additive $\theta$-approximation to $\tr{A_j\rho}$, $\tilde{c}$ is an additive  $\theta$-approximation to $\tr{C\rho}$ and $\cp = \tilde{c} - r \theta - \theta$. We first describe a simplified version of our quantum 2-sparse oracle (Lemma~\ref{lem:oracle}) that assumes access to a unitary acting as $\ket{j}\ket{0}\ket{0} \mapsto \ket{j} \ket{\tilde{a}_j} \ket{\psi_j}$, where $\ket{\psi_j}$ is some workspace state depending on~$j$.

At the end of this section we then discuss how to modify the analysis when we are given an oracle that acts as $\ket{j}\ket{0}\ket{0} \mapsto \ket{j} \sum_i \beta_j^{i}\ket{\tilde{a}^{i}_j} \ket{\psi^{i}_j}$ where each $\ket{\tilde{a}^{i}_j}$ is an approximation of $a_j$ and the amplitudes $\beta_j^{i}$ are such that measuring the second register with high probability returns an $\tilde{a}^{i}_j$ which is $\theta$-close to $a_j$.
We do so since the output of the trace-estimation procedure of the previous section is of this more complicated form.

\medskip

Our goal is to find a $y \in \pt(\tilde{a},\cp)$, i.e., a $y$ such that
\begin{align*}
  \nrm{y}_1&\leq r\\
  b^T y &\leq \alpha\\
  \tilde{a}^T y &\geq \cp\\
  y &\geq 0.
\end{align*}
Our first observation is that the polytope $\pt(\tilde{a},\cp)$ is extremely simple: it has only three non-trivial constraints and, if it is non-empty, then it contains a point with at most $2$ non-zero coordinates. The latter follows from general properties of linear programs: any feasible LP with at most $k$ constraints has a $k$-sparse optimal solution (see, e.g.,~\cite[Ch.~7]{schrijver1986TheoryOfLPAndIP}). Note that non-emptiness of $\pt(\tilde a,\cp)$ is equivalent to the optimal value of 
\begin{equation} \label{eq:simple LP}
\min \quad 1^T y \ \text{ s.t. } \  b^T y \leq \alpha, \ \tilde a^T y \geq \cp, \ y \geq 0
\end{equation}
being at most $r$ (the latter being an LP with only $2$ non-trivial constraints). We will give an alternative, constructive proof that we can obtain a $2$-sparse solution in Lemma~\ref{lem:2DArg}. This will also illustrate the intuition behind the definition of our $2$-sparse oracle. 

Before doing so, let us give a first naive approach to find a $y \in \pt(\tilde{a},\cp)$ which will not be sufficiently efficient for our purposes. Using the formulation in Equation~\eqref{eq:simple LP}, we could attempt to find a $y \in \pt(\tilde{a},\cp)$ by solving $\bigO(m^2)$ linear programs with only $2$ variables and $2$ constraints each (these LPs are obtained from~\eqref{eq:simple LP} by setting all but 2 variables equal to zero) and searching for an optimal solution whose value is at most $r$.
Here each linear program is determined by the values $\tilde{a}_j$, $b_j$ and~$\cp$, and thus we can decide if $\pt(\tilde{a},\cp)$ is non-empty using $\bigO(m^2)$ classical time (given these values). 

We use a more sophisticated approach to show that $\bigO(m)$ classical operations (and queries) suffice. Our approach is amenable to a quantum speedup: it also implies that only $\bOt{\sqrt{m}}$ quantum queries suffice. In particular, we now show how to reduce the problem of finding a $y \in \pt(\tilde{a},\cp)$ to finding a convex combination of points $(b_j,\tilde{a}_j)$ that lies within a certain region of the plane.

\medskip

First observe that if $\alpha \geq 0$ and $\cp \leq 0$, then $y=0$ is a solution and our oracle can return it. From now on we will assume that $\alpha < 0$ or $\cp>0$. 
Then for a feasible point $y$ we may write $y = Nq$ with $N=\nrm{y}_1 > 0$ and hence $\nrm{q}_1 = 1$.
So we are looking for an $N$ and a $q$ such that
\begin{align}
  b^T q &\leq \alpha / N\label{eq:Nbound}\\
  \tilde{a}^T q &\geq \cp / N\notag\\\notag
  \nrm{q}_1 &= 1\\\notag
  q &\geq 0\\\notag
  0 &< N \leq r.
\end{align}
We can now view $q \in \mathbb{R}_+^m$ as the coefficients of a convex combination of the points $p_i = (b_i,\tilde{a}_i)$ in the plane. We want such a combination that lies to the upper left of $g_N = (\alpha/N,\cp / N)$ for some $0 < N\leq r$. Let $\mathcal{G}_N$ denote the upper-left quadrant of the plane starting at $g_N$.
\begin{lemma}\label{lem:2DArg}
  If there is a $y \in \pt (\tilde{a},\cp)$, then there is a $2$-sparse $y^{\prime} \in \pt (\tilde a,\cp)$ such that $\nrm{y}_1 = \nrm{y^{\prime}}_1$.
\end{lemma}
\begin{proof}
Let $y\in \pt (\tilde{a},\cp)$, and $N=\nrm{y}_1$. Consider $p_i = (b_i,\tilde{a}_i)$ and $g = (\alpha / N,\cp / N)$ as before, and write $q = y/N$ such that $\sum_{j=1}^m q_j =1$, $q \geq 0$. The vector $q$ certifies that a convex combination of the points $p_i$ lies in $\mathcal{G}_N$. But then there exist $j,k \in \lbrack m \rbrack$ such that the line segment $\overline{p_jp_k}$ intersects~$\mathcal{G}_N$. All points on this line segment are convex combinations of $p_j$ and $p_k$, hence there is a convex combination of $p_j$ and $p_k$ that lies in~$\mathcal{G}_N$. This gives a $2$-sparse $q^{\prime}$, and $y' = N q' \in \pt (\tilde{a},\cp)$.
\end{proof}
We can now restrict our search to $2$-sparse $y$.
Let $\mathcal{G} = \bigcup_{N \in (0, r\rbrack } \mathcal{G}_N$, see Figure~\ref{fig:graph} for the shape of~$\mathcal G$. Then we want to find two points $p_j,p_k$ that have a convex combination in $\mathcal{G}$, since this implies that a scaled version of their convex combination gives a $y \in \pt(\tilde{a},\cp)$ with $\nrm{y}_1 \leq r$ (this scaling can be computed efficiently given $p_j$ and $p_k$).

Furthermore, regarding the possible (non-)emptiness of $\mathcal{G}$ we know the following by Lemma~\ref{lem:approxP} and Lemma~\ref{lem:2DArg}:
\begin{itemize}
\item If $\mathcal P_{0}(X) \cap \{y \in \mathbb R^m \colon \sum_{j} y_j \leq r\}$ is non-empty, then some convex combination of two of the $p_j$'s lies in $\mathcal G$.
\item If $\mathcal P_{4Rr\theta}(X) \cap \{y \in \mathbb R^m \colon \sum_{j} y_j \leq r\}$ is empty, then no convex combination of the $p_j$'s lies in~$\mathcal{G}$.
\end{itemize}
We first prove a simplified version of the main result of this section. The analysis below applies if there are $m$ points $p_j = (b_j,\tilde a_j)$, where $j \in [m]$, and we are given a unitary which acts as $\ket{j}\ket{0}\ket{0} \mapsto \ket{j} \ket{\tilde{a}_j} \ket{\psi_j}$. We later prove the result for more general oracles, which may give superpositions of approximations instead of just the one value~$\tilde a_j$.

 \begin{figure}[H]
	\renewcommand{\scale}{1}
	\renewcommand{\sfw}{.49}
	\centering
	\begin{subfigure}[b]{\sfw \linewidth}
		\centering
		\renewcommand{\cx}{-1}
		\renewcommand{\cy}{-.5}
		\input{Ggraph.txt}
		\caption{$\alpha<0,\cp<0$}
	\end{subfigure}
	\begin{subfigure}[b]{\sfw \linewidth}
		\centering
		\renewcommand{\cx}{-1}
		\renewcommand{\cy}{.5}
		\input{Ggraph.txt}
		\caption{$\alpha<0,\cp\geq 0$}
	\end{subfigure}
	\vskip .5cm
	\begin{subfigure}[b]{\sfw \linewidth}
		\centering
		\renewcommand{\cx}{1}
		\renewcommand{\cy}{-.5}
		\input{Ggraph.txt}
		\caption{$\alpha\geq 0,\cp< 0$}
	\end{subfigure}
	\begin{subfigure}[b]{\sfw \linewidth}
		\centering
		\renewcommand{\cx}{1}
		\renewcommand{\cy}{.5}
		\input{Ggraph.txt}
		\caption{$\alpha\geq 0,\cp\geq 0$}
	\end{subfigure}
	
	\caption{The region $\mathcal G$ in light blue. The borders of two quadrants $\mathcal G_N$ have been drawn by thick dashed blue lines. The red dot at the beginning of the arrow is the point $(\alpha/r,c^{\prime}/r)$.}
	\label{fig:graph}
\end{figure}
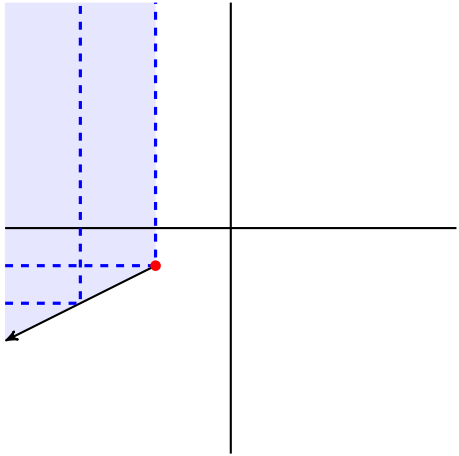
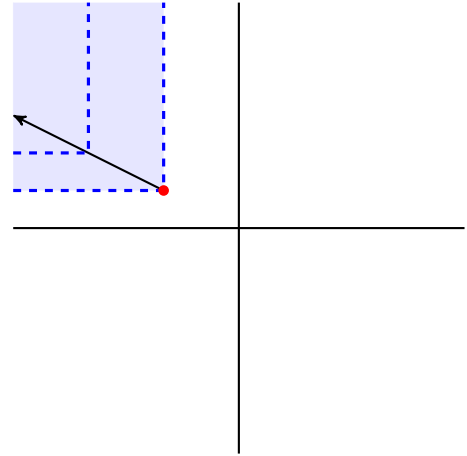
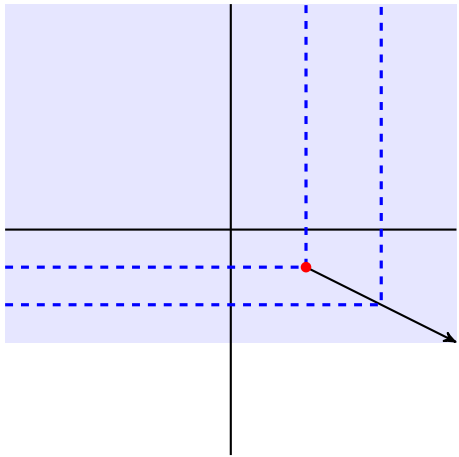
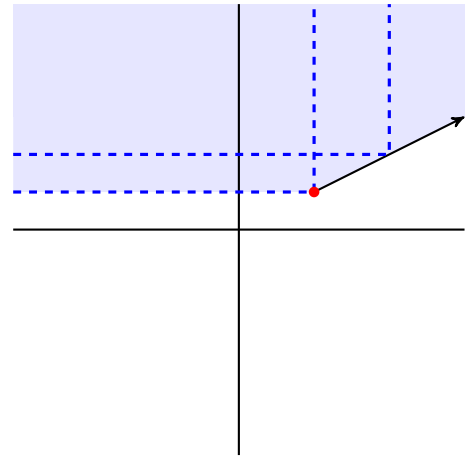

\begin{lemma}[Simple version]\label{lem:oracle}
There is an algorithm that returns a 2-sparse vector $q$ (with $q \geq 0$ and $\|q\|_1=1$) such that $\sum_{j=1}^m q_jp_j\in \mathcal G$, if one exists, using one search and two minimizations over the $m$ points $p_j=(b_j,\tilde{a}_j)$.
This gives a classical algorithm that uses $\bigO(m)$ calls to the subroutine that gives the entries of~$\tilde{a}$, and $\bigO(m)$ other operations; and a quantum algorithm that (in order to solve the problem with high probability) uses $\bigO(\sqrt{m})$ calls to an (exact quantum) subroutine that gives the entries of~$\tilde{a}$, and $\bOt{\sqrt{m}}$ other gates.
\end{lemma}
\addtocounter{lemma}{-1}
\begin{proof}
  The algorithm can be summarized as follows:
  \begin{enumerate}
  \item Check if $\alpha \geq 0$ and $\cp\leq 0$. If so, then return $q=0$. \label{it:step1}
  \item Check if there is a $p_i\in \mathcal G$. If so, then return $q = e_i$ 
  \label{it:step2}
  \item Find $p_j,p_k$ so that the line segment $\overline{p_jp_k}$ goes through $\mathcal G$ and return the corresponding $q$. \label{it:step3}
  \item If the first three steps did not return a vector $q$, then output `Fail'.
    \end{enumerate}

  The main realization is that in step 3 we can search separately for $p_j$ and $p_k$. 
  We explain this in more detail below, but first we will need a better understanding of the shape of $\mathcal G$ (see Figure~\ref{fig:graph} for illustration). The shape of $\mathcal G$ depends on the sign of $\alpha$ and $\cp$. 
  \begin{itemize}
  \item[(a)] If $\alpha < 0$ and $\cp < 0$. The corner point of $\mathcal{G}$ is $(\alpha / r,\cp / r)$. One edge goes up vertically and an other follows the line segment $\lambda \cdot  (\alpha,\cp)$ for $\lambda \in [ 1/r,\infty)$ starting at the corner. 
  \item[(b)] If $\alpha < 0$ and $\cp\geq 0$. Here $\mathcal G_N \subseteq \mathcal G_r$ for $N\leq r$. So $\mathcal G = \mathcal G_r$. The corner point is again $(\alpha / r,\cp / r)$, but now one edge goes up vertically and one goes to the left horizontally.
  \item[(c)] If $\alpha\geq 0$ and $\cp \leq 0$. This is the case where $y=0$ is a solution, $\mathcal G$ is the whole plane and has no corner.
  \item[(d)] If $\alpha\geq 0$ and $\cp > 0$. The corner point of $\mathcal G$ is again $(\alpha / r,\cp / r)$. From there one edge goes to the left horizontally and one edge follows the line segment $\lambda \cdot (\alpha,\cp)$ for $\lambda \in [ 1/r,\infty)$.
  \end{itemize}
  
    Since $\mathcal G$ is always an intersection of at most $2$ halfspaces, steps~\ref{it:step1}-\ref{it:step2} of the algorithm are easy to perform. In step~\ref{it:step1} we handle case (c) by simply returning $y=0$. For the other cases $(\alpha /r ,\cp /r)$ is the corner point of $\mathcal G$ and the two edges are simple lines. Hence in step~\ref{it:step2} we can easily search through all the points to find out if there is one lying in $\mathcal G$; since $\mathcal G$ is a very simple region, this only amounts to checking on which side of the two lines a point lies.

  \begin{figure}[hbt]
    \renewcommand{\scale}{.75}
    \renewcommand{\sfw}{.49}
    \centering
    \centering
    \renewcommand{\cx}{1}
    \renewcommand{\cy}{-1}
    \input{angles.txt}

    \caption{Illustration of $\mathcal G$ with the points $p_j,p_k$ and the angles $\angle \ell_j L_1,\angle L_1 L_2,\angle L_2\ell_k$ drawn in. Clearly the line $\overline{p_jp_k}$ only crosses $\mathcal{G}$ when the total angle is less than $\pi$.}
    \label{fig:angles}
  \end{figure}
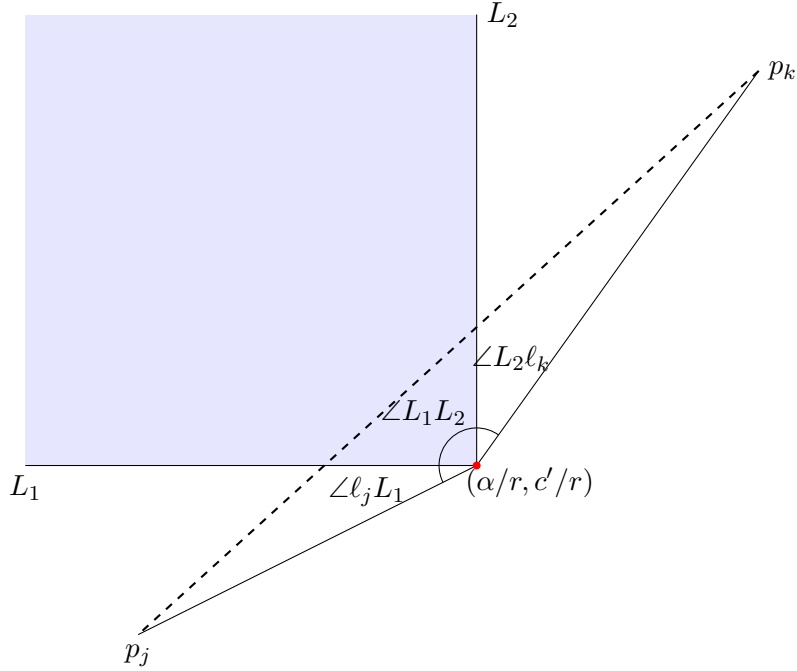

  Now, if we cannot find a single point in $\mathcal{G}$ in step~\ref{it:step2}, then we need a combination of two points in step~\ref{it:step3}. Let $L_1, L_2$ be the edges of $\mathcal G$ and let $\ell_j$ and $\ell_k$ be the line segments from $(\alpha/r,\cp/r)$ to $p_j$ and $p_k$, respectively.
  Then, as can be seen in Figure~\ref{fig:angles}, the line segment $\overline{p_jp_k}$ goes through $\mathcal G$ if and only if (up to relabeling $p_j$ and $p_k$) $\angle \ell_j L_1 + \angle L_1 L_2 + \angle L_2 \ell_k \leq \pi$. Since $\angle L_1 L_2$ is fixed, we can simply look for a $j$ such that $\angle \ell_j L_1$ is minimized and a $k$ such that $\angle L_2\ell_k$ is minimized. If $\overline{p_jp_k}$ does not pass through $\mathcal G$ for this pair of points, then it does not for any of the pairs of points.

  Notice that these minimizations can be done separately and hence can be done in the stated complexity. Given the minimizing points $p_j$ and $p_k$, it is easy to check if they give a solution by calculating the angle between $\ell_j$ and $\ell_k$. The coefficients of the convex combination $q$ are then easy to compute.
\end{proof}

We now consider the more general case where we are given access to a unitary which for each $j$ provides a superposition over different values $\tilde a_j$. We do so because the trace estimation procedure of Corollary~\ref{col:traceCalc} provides an oracle of this form.

\begin{lemma}[General version]\label{lem:oraclefull}
  Assume that we are given an oracle that acts as
  \[
    \ket{j}\ket{0}\ket{0} \mapsto \ket{j} \sum_i \beta_j^{i}\ket{\tilde{a}^{i}_j} \ket{\psi^{i}_j}
  \]
  where each $\ket{\tilde{a}^{i}_j}$ is an approximation of $a_j$ and the amplitudes $\beta_j^{i}$ are such that measuring the second register with high probability returns an $\tilde{a}^{i}_j$ which is $\theta$-close to $a_j$.

	There is a quantum algorithm that uses $\bOt{\sqrt{m}}$ calls to the oracle described above, and the same order of two-qubit gates, and (with high probability) has the following guarantees.
	\begin{itemize}
		\item If $\mathcal P_{0}(X) \cap \{y \in \mathbb R^m \colon \sum_{j} y_j \leq r\}$ is non-empty, then the algorithm returns a $2$-sparse vector in $\mathcal P_{4Rr\theta}(X) \cap \{y \in \mathbb R^m \colon \sum_{j} y_j \leq r\}$.
		\item If $\mathcal P_{4Rr\theta}(X) \cap \{y \in \mathbb R^m \colon \sum_{j} y_j \leq r\}$ is empty, then the algorithm correctly concludes that $\mathcal P_{0}(X) \cap \{y \in \mathbb R^m \colon \sum_{j} y_j \leq r\}$ is empty.
	\end{itemize}
\end{lemma}
\begin{proof}
Since we can exponentially reduce the probability that we obtain an $\tilde{a}^{i}_j$ which is further than~$\theta$ away from $a_j$, we will for simplicity assume that for all $i,j$ we have $|\tilde{a}^{i}_j - a_j| \leq \theta$; the neglected exponentially small probabilities will only affect the analysis in negligible ways.

Note that while we do not allow our quantum algorithm enough time to obtain classical descriptions of all $\tilde a_j$s (we aim for a runtime of $\bOt{\sqrt{m}}$), we do have enough time to compute $\tilde{c}$ once initially (after this measurement, $\mathcal G$ is well-defined). Knowing $\tilde c$, we can compute the angles defined by the points $p_j^i = (b_j,\tilde{a}_j^{i})$ with respect to the corner point of $(\alpha/r,(\tilde{c}-\theta)/ r-\theta)$ and the lines $L_1, L_2$ (see Figure~\ref{fig:angles}). We now apply our generalized minimum-finding algorithm with runtime $\bOt{\sqrt{m}}$ (see Theorem~\ref{thm:genMin}) starting with a uniform superposition over the $j$s to find $k,\ell \in [m]$ and points $p_{k}^i$ and $p_{\ell}^{i'}$ approximately minimizing the respective angles to lines $L_1,L_2$. Here `approximately minimizing' means that there is no $j\in[m]$ such that for all $i''$ the angle of $p_j^{i''} = (b_j,\tilde{a}_j^{i''})$ with $L_1$ is smaller than that of $p_{k}^i$ with $L_1$ (and similar for $\ell$ and $L_2$). From this point on we can simply consider the model in the simple version of this lemma, since by the analysis above there exists an approximation $\tilde{a}\in\mathbb{R}^m$ with $\tilde{a}_k = \tilde{a}_k^i$ and $\tilde{a}_{\ell}= \tilde{a}_{\ell}^{i'}$ and where $k$ and $\ell$ are the correct minimizers.

It follows from Lemma~\ref{lem:approxP} and Lemma~\ref{lem:2DArg} that if $\mathcal P_{0}(X) \cap \{y \in \mathbb R^m: \sum_{j} y_j \leq r\}$ is non-empty, then some convex combination of $(\tilde a_\ell, b_\ell)$ and $(\tilde a_k, b_k)$ lies in $\mathcal G$. On the other hand, if $\mathcal P_{4Rr\theta}(X) \cap \{y \in \mathbb R^m: \sum_{j} y_j \leq r\}$ is empty, then the same lemmas guarantee that we correctly conclude that $\mathcal P_{0}(X) \cap \{y \in \mathbb R^m: \sum_{j} y_j \leq r\}$ is empty. 
\end{proof}

\subsection{Total runtime} \label{sec:runtime}
We are now ready to add our quantum implementations of the trace calculations and the oracle to the classical Arora-Kale framework.
\upperbound*
\begin{proof} Using our implementations of the different building blocks, it remains to calculate what the total complexity will be when they are used together.
  \begin{description}
  \item[Cost of the oracle for $H^{(t)}$.]
    The first problem in each iteration is to obtain access to an oracle for $H^{(t)}$. In each iteration the oracle will produce a $y^{(t)}$ that is at most $2$-sparse, and hence in the $(t+1)$th iteration, $H^{(t)}$ is a linear combination of $2t$ of the $A_j$ matrices and the $C$ matrix.

    We can write down a sparse representation of the coefficients of the linear combination that gives $H^{(t)}$ in each iteration by adding the new terms coming from $y^{(t)}$. This will clearly not take longer than $\bOt{T}$, since there are only a constant number of terms to add for our oracle. As we will see, this term will not dominate the complexity of the full algorithm.

    Using such a sparse representation of the coefficients, one query to a sparse representation of $H^{(t)}$ will cost $\bOt{st}$ queries to the input matrices and $\bOt{st}$ other gates. For a detailed explanation and a matching lower bound, see Appendix~\ref{app:sparsematrixsum}.
  \item[Cost of the oracle for $\tr{A_j\rho}$.]
    In each iteration $M^{(t)}$ is made to have operator norm at most~$1$.
    This means that
    \[
      \nrm{-\eta H^{(t)}}  \leq \eta \sum_{\tau = 1}^t \nrm{M^{(\tau)}} \leq \eta t.
    \]
    Furthermore we know that $H^{(t)}$ is at most $d:=s(2t+1)$-sparse.
    Calculating $\tr{A_j\rho}$ for one index $j$ up to an additive error of $\theta := \eps/(12Rr)$ can be done using the algorithm from Corollary~\ref{col:traceCalc}. This will take
    \[
      \bOt{\sqrt{n}\frac{\nrm{H}d}{\theta}} = \bOt{\sqrt{n} s\frac{\eta t^2 Rr}{\eps}}
    \]
    queries to the oracle for $H^{(t)}$ and the same order of other gates.
    Since each query to $H^{(t)}$ takes $\bOt{st}$ queries to the input matrices, this means that
    \[
      \bOt{\sqrt{n}s^2\frac{\eta t^3 Rr}{\eps}}
    \]
    queries to the input matrices will be made, and the same order of other gates, for each approximation of a $\tr{A_j\rho}$ (and similarly for approximating $\tr{C\rho}$).
  \item[Total cost of one iteration.]
    Lemma~\ref{lem:oracle} tells us that we will use $\bOt{\sqrt{m}}$ calculations of $\tr{A_j\rho}$, and the same order of other gates, to calculate a classical description of a $2$-sparse $y^{(t)}$.
    This brings the total cost of one iteration to
    \[
      \bOt{\sqrt{nm} s^2\frac{\eta t^3 Rr}{\eps}}
    \]
    queries to the input matrices, and the same order of other gates.
  \item[Total quantum runtime for SDPs.]
    Since $w\leq r+1$ we can set $T = \bOt{\frac{R^2r^2}{\eps^2}}$. With $\eta = \sqrt{\frac{\ln(n)}{T}}$,
    summing over all iterations in one run of the algorithm gives a total cost of
    \begin{align*}
      \bOt{\sum_{t=1}^T \sqrt{nm} s^2\frac{\eta t^3 Rr}{\eps}} &=   \bOt{\sqrt{nm} s^2\frac{\eta T^4 Rr}{\eps}} \\
                                                               &=   \bOt{\sqrt{nm} s^2\left( \frac{Rr}{\eps}\right)^{\!\!8} }
    \end{align*}
    queries to the input matrices and the same order of other gates.
  \end{description}
  \vskip -0.5cm
\end{proof}
\paragraph{Total quantum runtime for LPs.}
The final complexity of our algorithm contains a factor $\bOt{sT}$ that comes from the sparsity of the $H^{(t)}$ matrix. This assumes that when we add the input matrices together, the rows become less sparse. This need not happen for certain SDPs. For example, in the SDP relaxation of MAXCUT, the $H^{(t)}$ will always be $d$-sparse, where $d$ is the degree of the graph. A more important class of examples is that of linear programs: since LPs have diagonal $A_j$ and $C$, their sparsity is $s=1$, and even the sparsity of the $H^{(t)}$ is always~$1$. This, plus the fact that the traces can be computed without a factor $\nrm{H}$ in the complexity (as shown in Corollary~\ref{col:traceDiagCalc} in Section~\ref{sec:lptrace}), means that our algorithm solves LPs with
\[
  \bOt{\sqrt{nm} \left( \frac{Rr}{\eps}\right)^{\!\!5} }
\]
queries to the input matrices and the same order of other gates.

\paragraph{Total classical runtime.}
$\kern-0.9mm\text{Using}$ the classical techniques for trace estimation from Appendix~\ref{app:trace}, and the classical version of our oracle (Lemma~\ref{lem:oracle}), we are also able to give a general classical instantiation of the Arora-Kale framework. The final complexity will then be
\[
  \bOt{nms\left(\frac{Rr}{\eps}\right)^{\!\!4}+ns\left(\frac{Rr}{\eps}\right)^{\!\!7}}.
\]
The better dependence on $Rr/\eps$ and $s$, compared to our quantum algorithm, comes from the fact that we now have the time to write down intermediate results explicitly. For example, we do not need to recalculate parts of $H^{(t)}$ for every new query to it, instead we can just calculate it once at the start of the iteration by adding $M^{(t)}$ to $H^{(t-1)}$ and writing down the result.

\paragraph{Further remarks.}
We want to stress again that our solver is meant to work for \emph{all} SDPs. In particular, it does not use the structure of a specific SDP. As we  show in the next section, every oracle that works for all SDPs must have large width. To obtain quantum speedups for a \emph{specific} class of SDPs, it will be necessary to develop oracles tuned to that problem. We view this as an important direction for future work. Recall from the introduction that Arora and Kale also obtain fast classical algorithms for problems such as MAXCUT by doing exactly that: they develop specialized oracles for those problems.

\section{Downside of this method: general oracles are restrictive}\label{sec:downside}

In this section we show some of the limitations of a method that uses sparse or general oracles, i.e., ones that are not optimized for the properties of specific SDPs. We will start by discussing sparse oracles in the next section. We will use a counting argument to show that sparse solutions cannot hold too much information about a problem's solution.
In Section~\ref{sec: general width bounds} we will show that width-bounds that do not depend on the specific structure of an SDP are for many problems not efficient.
As in the rest of the paper, we will assume the notation of Section~\ref{sec:upperbounds}, in particular of Meta-Algorithm~\ref{alg:AKSDP}.

\subsection{Sparse oracles are restrictive}
\begin{lemma}
    If, for some specific SDP of the form~\eqref{eq:SDP}, every $\eps$-optimal dual-feasible vector has at least $\ell$ non-zero elements, then the width $w$ of any $k$-sparse \textsf{Oracle}$_{\eps/3}$ for this SDP is such that $\frac{Rw}{\eps} = \Omega\left(\sqrt{\frac{\ell}{k\ln(n)}}\right)$.
\end{lemma}
\begin{proof}
  The vector $\bar{y}$ returned by Meta-Algorithm~\ref{alg:AKSDP} is, by construction, the average of $T$ vectors $y^{(t)}$ that are all $k$-sparse, plus one extra $1$-sparse term of $\frac{\eps}{R}e_1$, and hence $\ell \leq kT+1$. The stated bound on $\frac{Rw}{\eps}$ then follows directly by combining this inequality with $T = \bigO(\frac{R^2w^2}{\eps^2}\ln(n))$.
\end{proof}
The oracle presented in Section~\ref{sec:oracle} always provides a $2$-sparse vector $y$. This implies that if an SDP requires an $\ell$-sparse dual solution, we must have $\frac{Rw}{\eps} = \Omega(\sqrt{\ell / \ln(n)})$. This in turn means that the upper bound on the runtime of our algorithm will be of order $\ell^4 \sqrt{nm} s^2$. This is clearly bad if $\ell$ is of the order $n$ or $m$.

Of course it could be the case that almost every SDP of interest has a sparse approximate dual solution (or can easily be rewritten so that it does), and hence sparseness might be not a restriction at all. However, as we will see below, this is not the case. We will prove that for certain kinds of SDPs, no ``useful'' dual solution can be very sparse. Intuitively, a dual solution to an SDP is ``useful'' if it can be turned into a solution of the problem that the SDP is trying to solve. We make this more precise in the definition below. 

\begin{definition}\label{def:problem}
    A problem is defined by a function $f$ that, for every element $p$ of the problem domain $\mathcal D$, gives a subset of the solution space $\mathcal S$, consisting of the solutions that are considered correct.
   We say a family of SDPs, $\{SDP^{(p)}\}_{p\in \mathcal D}$, solves the problem via the dual if there is an $\eps \geq 0$ and a function $g$ such that for every $p \in \mathcal D$ and every $\eps$-optimal dual-feasible vector $y^{(p)}$ to $SDP^{(p)}$:
  \[
    g(y^{(p)}) \in f(p).
  \]
  In other words, an $\eps$-optimal dual solution can be converted into a correct solution of the original problem without more knowledge of $p$.
\end{definition}
For these kinds of SDP families we will prove a lower bound on the sparsity of the dual solutions.
The idea for this bound is as follows. If you have a lot of different instances that require different solutions, but the SDPs are equivalent up to permuting the constraints and the coordinates of $\mathbb{R}^n$, then a dual solution should have a lot of unique permutations and hence cannot be too sparse.

\begin{theorem}
  Consider a problem and a family of SDPs as in Definition~\ref{def:problem}. Let $\mathcal T \subseteq \mathcal D$ be such that for all $p,q \in \mathcal T$:
  \begin{itemize}
  \item $f(p) \cap f(q) = \emptyset$. That is, a solution to $p$ is not a solution to $q$ and vice versa.
  \item The number of constraints $m$ and the primal variable size $n$ are the same for $SDP^{(p)}$ and $SDP^{(q)}$.
  \item Let $A_j^{(p)}$ be the constraints of $SDP^{(p)}$ and $A_j^{(q)}$ those from $SDP^{(q)}$ (and define $C^{(p)}$, $C^{(q)}$, $b_j^{(p)}$, and $b_j^{(q)}$ in the same manner). Then there exist $\sigma \in S_n$, $\pi \in S_m$ s.t.\ $\sigma^{-1} A^{(p)}_{\pi(j)} \sigma = A^{(q)}_j$ (and $\sigma^{-1} C^{(p)} \sigma = C^{(q)}$, $b^{(p)}_{\pi(j)} = b^{(q)}_j$). That is, the SDPs are the same up to permutations of the labels of the constraints and permutations of the coordinates of $\mathbb{R}^n$.
  \end{itemize}
  If $y^{(p)}$ is an $\eps$-optimal dual-feasible vector to $SDP^{(p)}$ for some $p\in \mathcal T$, then $y^{(p)}$ is at least $\frac{\log(|\mathcal T|)}{\log m}$-dense (i.e., has at least that many non-zero entries).
\end{theorem}
\begin{proof}
  We first observe that, with $SDP^{(p)}$ and $SDP^{(q)}$ as in the lemma, if $y^{(p)}$ is an $\eps$-optimal dual-feasible vector of $SDP^{(p)}$, then $y^{(q)}$ defined by
  \[
    y^{(q)}_j := y^{(p)}_{\pi(j)} = \pi(y^{(p)})_j
  \]
  is an $\eps$-optimal dual vector for $SDP^{(q)}$. Here we use the fact that a permutation of the $n$~coordinates in the primal does not affect the dual solutions. Since $f(p)\cap f(q) = \emptyset$ we know that $g(y^{(p)}) \neq g(y^{(q)})$ and so $y^{(p)} \neq y^{(q)}$. Since this is true for every $q$ in $\mathcal T$, there should be at least $|\mathcal T|$ different vectors $y^{(q)} = \pi (y^{(p)})$.

  A $k$-sparse vector can have $k$ different non-zero entries and hence the number of possible unique permutations of that vector is at most
  \[
    \binom{m}{k} k! = \frac{m!}{(m-k)!} = \prod_{t = m-k+1}^m t \leq m^k
  \]
  so
  \[
    \frac{\log |\mathcal{T}|}{\log m} \leq k.
  \]
  \vskip -5mm
\end{proof}

\paragraph{Example.} Consider the $(s,t)$-mincut problem, i.e., the dual of the $(s,t)$-maxflow. Specifically, consider a simple instance of this problem: the union of two complete graphs of size $z+1$, where $s$ is in one subgraph and $t$ in the other. Let the other vertices be labeled by $\{1,2,\dots,2z\}$.
Every assignment of the labels over the two halves gives a unique mincut, in terms of which labels fall on which side of the cut. There is exactly one partition of the vertices in two sets that cuts no edges (namely the partition consists of the two complete graphs), and every other partition cuts at least $z$ edges. Hence a $z/2$-approximate cut is a mincut. This means that there are $\binom{2z}{z}$ problems that require a different output. So for every family of SDPs that is symmetric under permutation of the vertices and for which a $z/2$-approximate dual solution gives an $(s,t)$-mincut, the sparsity of a $z/2$-approximate dual solution is at least\footnote{Here $m$ is the number of constraints, not the number of edges in the graph.}
\[
  \frac{\log {\binom{2z}{z}}}{\log m} \geq \frac{z}{ \log{m}},
\]
where we used that $\binom{2z}{z} \geq \frac{2^{2z}}{2\sqrt{z}}$. In particular this holds for the standard linear programming formulation of the $(s,t)$-maxflow/$(s,t)$-mincut problem. 

\subsection{General width-bounds are restrictive for certain SDPs} \label{sec: general width bounds}
In this section we will show that width-bounds can be restrictive when they do not consider the specific structure of an SDP.

\begin{definition} An algorithm $O$ is called a \emph{general oracle} if, when provided with an error parameter $\eps$, it correctly implements an \textsf{Oracle}$_{\eps}$ (as in Algorithm~\ref{alg:Oracle}) for all inputs. We use $O_\eps$ to denote the algorithm provided by $O$ with error parameter $\eps$ fixed. A function $w(n,m,s,r,R,\eps)$ is called a \emph{general width-bound} for a general oracle if, for every $0<\eps<1/2$, the value $w(n,m,s,r,R,\eps)$ is a correct width-bound (see Definition~\ref{def:width}) for $O_\eps$ for every SDP with parameters $n,m,s,r$, and $R$. In particular, the function $w$ may not depend on the structure of the input $A_1,\ldots,A_m$, $C$, $b$ or on the value of $\alpha$.
\end{definition}

We will show that general width-bounds need to scale with $r^{*}$ (recall that $r^{*}$ denotes the smallest $\ell_1$-norm of an optimal solution to the dual). We then go on to show that if two SDPs in a class can be combined to get another element of that class in a natural manner, then, under some mild conditions, $r^{*}$ will be of the order $n$ and $m$ for some instances of the class.

We start by showing, for specifically constructed LPs, a lower bound on the width of any oracle. Although these LPs will not solve any useful problem, every general width-bound should also apply to these LPs. This gives a lower bound on general width-bounds.
\begin{lemma}\label{lem:genisrstar}
For every $n \geq 3$, $m \geq 3$, $s \geq 1$, $R^* \geq 1$, $r^{*}>0$, there is an SDP with these parameters such that for any $0 \leq \eps \leq 1/2$ any \textsf{Oracle}$_{\eps}$ for this SDP has width at least $r^*/2$.
\end{lemma}
\begin{proof}
  We will construct an LP for $n=m=3$. This is enough to prove the lemma since LPs are a subclass of SDPs and we can increase $n$, $m$, and $s$ by adding more dimensions and $s$-dense SDP constraints that do not influence the analysis below.
  For some $k> 0$, consider the following LP
  \begin{align*}
    \max \ \ \ & (1,0,0) x\\
    \text{s.t.} \ \ \ & \begin{bmatrix}
      1 & 1 & 1\\
      1/k & 1 & 0\\
      -1 & 0 & -1
    \end{bmatrix} x \leq \begin{bmatrix}
      R\\
      0\\
      -R
    \end{bmatrix}\\
               & x \geq 0
  \end{align*}
  where the first row is the primal trace constraint.
  Notice that $x_1 = x_2 = 0$ due to the second constraint. This implies that $\opt = 0$ and, due to the last constraint, that $x_3\geq R$. Notice that $(0,0,R)$ is actually an optimal solution, so $R^{*} = R$.

  To calculate $r^{*}$, look at the dual of the LP:
  \begin{align*}
    \min \ \ \ & (R,0,-R) y\\
    \text{s.t.} \ \ \ & \begin{bmatrix}
      1 & 1/k & -1\\
      1 & 1 & 0\\
      1 & 0 & -1
    \end{bmatrix} y \geq \begin{bmatrix}
      1\\
      0\\
      0
    \end{bmatrix}\\
               & y \geq 0,
  \end{align*}
  due to strong duality its optimal value is $0$ as well. This implies $y_1 = y_3$, so the first constraint becomes $y_2 \geq k$. This in turn implies $r^{*}\geq k$, which is actually attained (by $y = (0,k,0)$) so $r^{*} = k$.

  Since the oracle and width-bound should work for every $x\in \mathbb{R}^3_+$ and every $\alpha$, they should in particular work for $x = (R,0,0)$ and $\alpha = 0$. In this case the polytope for the oracle becomes
  \begin{align*}
    \mathcal{P}_{\eps}(x) := \{ y \in \mathbb R^m:\ & y_1 - y_3 \leq 0, \\
                                                    & y_1-y_3+y_2/k\geq 1-\eps/R,\\
                                                    & y \geq 0 \}.
  \end{align*}
  since $b^T y = y_1-y_3$, $c^T x = 1$, $a_1^T x = 1$, $a_2^T x = 1/k$ and $a_3^T x = -1$. This implies that for every $y\in \mathcal{P}_{\eps}(x)$, we have $y_2 \geq k(1-\eps/R) \geq k/2 = r^{*}/2$, where the second inequality follows from the assumptions that $\eps \leq 1/2$ and $R\geq 1$.

  Notice that the term
  \[
    \nrm{\sum^m_{j=1} y_j A_j - C}
  \]
  in the definition of width for an SDP becomes
  \[
    \nrm{ A^T y - c}_{\infty}
  \]
  in the case of an LP. In our case, due to the second constraint in the dual, we know that
  \[
    \nrm{ A^T y - c}_{\infty} \geq y_1 + y_2 \geq \frac{r^{*}}{2}
  \]
  for every vector $y$ from $\mathcal{P}_{\eps}(x)$. This shows that any oracle has width at least $r^{*}/2$ for this LP.
\end{proof}
\begin{corollary}
  For every general width-bound $w(n,m,s,r,R,\eps)$, if $n,m\geq 3$, $s\geq 1$, $r>0$, $R>1$, and $\eps\leq 1/2$, then
  \[
    w(n,m,s,r,R,\eps) \geq \frac{r}{2}.
  \]
\end{corollary}
\begin{proof}
Consider the LP given by Lemma~\ref{lem:genisrstar} with $r^* = r$. Using a general oracle with general width-bound $w$ for this LP implies the corollary.
\end{proof}
Note that this bound applies to both our algorithm and the one given by Brand\~ao and Svore.
It turns out that for many classes of SDPs it is natural to assume that $m$, $r^{*}$, and $R^{*}$ grow linearly with $n$, and that the ``logical'' choice of $\eps$ also scales linearly with $n$. We now argue that this is for instance the case when SDPs in a class combine in a natural manner. As an example, consider the MAXCUT problem. For, e.g., $d$-regular graphs the MAXCUT value grows linearly with the number of vertices $n$. Therefore, one is generally interested in finding a constant \emph{multiplicative} approximation to the optimal value. For $d$-regular graphs this would thus translate to an \emph{additive} error which scales linearly with the number of vertices. We argue below that for the SDP-relaxation it is also natural to assume that $r$ and $R$ grow linearly in $n$. 
Take for example two SDP-relaxations for the MAXCUT problem on two graphs $G^{(1)}$ and $G^{(2)}$ (on $n^{(1)}$ and $n^{(2)}$ vertices, respectively):

\noindent\begin{minipage}{\textwidth}
  \begin{minipage}{.49\textwidth}
    \begin{align*}
      \max \quad & \tr{L(G^{(1)})X^{(1)}} \\
      \text{s.t.}\ \ \ & \tr{X^{(1)}} \leq n^{(1)}\\
                 & \tr{E_{jj} X^{(1)}} \leq 1 \text{ for }j=1,\dots,n^{(1)}\\
                 &X^{(1)} \succeq 0
    \end{align*}
  \end{minipage}
  \begin{minipage}{.49\textwidth}
    \begin{align*}
      \max \quad & \tr{L(G^{(2)})X^{(2)}} \\
      \text{s.t.}\ \ \ & \tr{X^{(2)}} \leq n^{(2)}\\
                 & \tr{E_{jj} X^{(2)}} \leq 1 \text{ for }j=1,\dots,n^{(2)}\\
                 &X^{(2)} \succeq 0
    \end{align*}
  \end{minipage}
\end{minipage}
\noindent Where $L(G)$ is the Laplacian of a graph. Note that this is not normalized to operator norm $\leq 1$, but for simplicity we ignore this here. If we denote the direct sum of two matrices by $\oplus$, that is
\[
  A\oplus B = \begin{bmatrix}A &0\\0&B\end{bmatrix},
\]
then, for the disjoint union of the two graphs, we have
\[
  L(G^{(1)} \cup G^{(2)}) = L(G^{(1)}) \oplus L(G^{(2)}).
\]
This, plus the fact that the trace distributes over direct sums of matrices, means that the SDP relaxation for MAXCUT on $G^{(1)} \cup G^{(2)}$ is the same as a natural combination of the two separate maximizations:
\begin{align*}
  \max \quad & \tr{L(G^{(1)})X^{(1)}} + \tr{L(G^{(2)})X^{(2)}} \\
  \text{s.t.}\ \ \ & \tr{X^{(1)}}+ \tr{X^{(2)}} \leq n^{(1)}+n^{(2)}\\
             & \tr{E_{jj} X^{(1)}} \leq 1 \text{ for }j=1,\dots,n^{(1)}\\
             & \tr{E_{jj} X^{(2)}} \leq 1 \text{ for }j=1,\dots,n^{(2)}\\
             &X^{(1)},X^{(2)} \succeq 0.
\end{align*}
It is easy to see that the new value of $n$ is $n^{(1)}+n^{(2)}$, the new value of $m$ is $m^{(1)}+m^{(2)}-1$ and the new value of $R^{*}$ is $n^{(1)}+n^{(2)} = R^{* (1)}+R^{* (2)}$. Likewise, it is natural to assume that the desired approximation error for the combined SDP is the sum of the desired errors for the individual SDPs: starting with feasible solutions $X^{(i)}$ that are $\eps^{(i)}$-approximate solutions to the two SDP-relaxations ($i = 1,2$), the matrix $X^{(1)} \oplus X^{(2)}$ is an $(\eps^{(1)}+\eps^{(2)})$-approximate solution to the combined SDP.
It remains to see what happens to $r^{*}$, and so, for general width-bounds, what happens to~$w$. As we will see later in this section, under some mild conditions, these kind of combinations imply that there are MAXCUT-relaxation SDPs for which $r^{*}$ also increases linearly, but this requires a bit more work.
\begin{definition} \label{def:combining}
  We say that a class of SDPs (each with an associated allowed approximation error) is \emph{combinable} if there is a $k\geq 0$ so that for every two elements in this class, $(SDP^{(a)},\eps^{(a)})$ and $(SDP^{(b)},\eps^{(b)})$, there is an instance in the class, $(SDP^{(c)},\eps^{(c)})$, that is a combination of the two in the following sense:
  \begin{itemize}
  \item $C^{(c)} = C^{(a)} \oplus C^{(b)}$.
  \item $A^{(c)}_j = A^{(a)}_j \oplus A^{(b)}_j$ and  $b_j^{(c)} = b_j^{(a)}+b_j^{(b)}$ for $j \in \lbrack k \rbrack$.
  \item $A^{(c)}_{j} = A^{(a)}_{j}\oplus \mathbf{0}$ and $b^{(c)}_{j} = b^{(a)}_{j}$ for $j = k+1,\dots,m^{(a)}$.
  \item $A^{(c)}_{m^{(a)}+j-k} = \mathbf{0} \oplus A^{(b)}_{j}$ and $b^{(c)}_{m^{(a)}+j-k} = b^{(b)}_{j}$ for $j = k+1,\dots,m^{(b)}$.
  \item $\eps^{(c)} = \eps^{(a)}+\eps^{(b)}$.
  \end{itemize}
  In other words, some fixed set of constraints are summed pairwise, and the remaining constraints get added separately.
\end{definition}
The motivation for the above definition reflects the following: if $X^{(a)}$ and $X^{(b)}$ are feasible solutions to $SDP^{(a)}$ and $SDP^{(b)}$ that are $\eps^{(a)}$-approximate and $\eps^{(b)}$-approximate solutions respectively, then $X^{(a)} \oplus X^{(b)}$ is an $(\eps^{(a)}+\eps^{(b)})$-approximate solution to $SDP^{(c)}$.

Furthermore, note that this is a natural generalization of the combining property of the MAXCUT relaxations (in that case $k=1$ to account for the constraint giving the trace bound).
\begin{theorem}
  If a class of SDPs is combinable and there is an element $SDP^{(1)}$ for which every optimal dual solution has the property that
  \[
    \sum_{j=k+1}^m y_m \geq \delta
  \]
  for some $\delta >0$, then there is a sequence $(SDP^{(t)},\eps^{(t)})_{t\in \mathbb{N}}$ in the class such that $\frac{R^{* (t)}r^{* (t)}}{\eps^{(t)}}$ increases linearly in $n^{(t)}$, $m^{(t)}$ and $t$, and $SDP^{(t)}$ is the $t$-fold combination of $SDP^{(1)}$ with itself.
\end{theorem}
\begin{proof}
  The sequence we will consider is the $t$-fold combination of $SDP^{(1)}$ with itself.
  If $SDP^{(1)}$ is
  \noindent\begin{minipage}{\textwidth}
    \begin{minipage}{.49\textwidth}
      \begin{align*}
        \max \quad &\Tr(CX) \\
        \text{s.t.}\ \ \ &\Tr(A_j X) \leq b_j \quad \text{ for } j \in [m^{(1)}], \\
                   &X \succeq 0
      \end{align*}\end{minipage}
    \begin{minipage}{.49\textwidth}
      \begin{align*}
        \min \quad &\sum_{j=1}^{m^{(1)}} b_j y_j \\
        \text{s.t.}\ \ \ &\sum_{j=1}^{m^{(1)}} y_j A_j - C \succeq 0,\\
                   &y \geq 0
      \end{align*}\end{minipage}
  \end{minipage}
  then $SDP^{(t)}$ is
  \begin{align*}
    \max \quad &\ \sum_{i=1}^t \tr{CX_i} & &\\
    \text{s.t.}\ \ \ & \sum_{i=1}^t \Tr(A_j X_i) \leq t b_j &&\text{ for } j \in [k], \\
               & \Tr(A_j X_i) \leq  b_j &&\text{ for } j = k+1 , \dots , m^{(1)} \text{ and } i = 1 , \dots , t \\
               &X_i \succeq 0 &&\text{ for all } i = 1, \dots,t
  \end{align*}
  with dual
  \begin{align*}
    \min \quad &\sum_{j=1}^k t b_j y_j + \sum_{i=1}^t \sum_{j=k+1}^{m^{(1)}} b_j y_j^i\\
    \text{s.t.}\ \ \ &\sum_{j=1}^k y_j A_j  + \sum_{j=k+1}^{m^{(1)}} y_j^i A_j \succeq C \text{ for } i = 1,\dots,t\\
               &y,y^i \geq 0.
  \end{align*}
  First, let us consider the value of $\opt^{(t)}$. Let $X^{(1)}$ be an optimal solution to $SDP^{(1)}$ and for all $i \in \lbrack t \rbrack$ let $X_i = X^{(1)}$. Since these $X_i$ form a feasible solution to $SDP^{(t)}$, this shows that $\opt^{(t)} \geq t\cdot \opt^{(1)}$.
  Furthermore, let $y^{(1)}$ be an optimal dual solution of $SDP^{(1)}$, then $(y^{(1)}_1,\dots,y^{(1)}_k)\oplus \left( y^{(1)}_{k+1},\cdots,y^{(1)}_{m^{(1)}}\right)^{\oplus t}$ is a feasible dual solution for $SDP^{(t)}$ with objective value $t\cdot \opt^{(1)}$, so $\opt^{(t)} = t \cdot \opt^{(1)}$.

  Next, let us consider the value of $r^{* (t)}$. Let $\tilde{y} \oplus y^1 \oplus \dots \oplus y^t$ be an optimal dual solution for $SDP^{(t)}$, split into the parts of $y$ that correspond to different parts of the combination. Then $\tilde{y} \oplus y^i$ is a feasible dual solution for $SDP^{(1)}$ and hence $b^T ( \tilde{y}\oplus y^i) \geq \opt^{(1)}$. On the other hand we have
  \[
    t \cdot \opt^{(1)} = \opt^{(t)} = \sum_{i=1}^t b^T (\tilde{y} \oplus y^i),
  \]
  this implies that each term in the sum is actually equal to $\opt^{(1)}$. But if $(\tilde{y}\oplus y^i)$ is an optimal dual solution of $SDP^{(1)}$ then $\nrm{(\tilde{y}\oplus y^i)}_1 \geq r^{* (1)}$ by definition and $\nrm{y^i}_1 \geq \delta$. We conclude that $r^{* (t)} \geq r^{* (1)} - \delta + t\delta$.

  Now we know the behavior of $r^{*}$ under combinations, let us look at the primal to find a similar statement for $R^{*(t)}$. Define a new SDP, $\widehat{SDP}^{(t)}$, in which all the constraints are summed when combining, that is, in Definition~\ref{def:combining} we take $k = n^{(1)}$, however, contrary to that definition, we even sum the psd constraints: 
  \vskip-5mm
  \begin{align*}
    \max \quad &\ \sum_{i=1}^t \tr{CX_i} \\
    \text{s.t.}\ \ \ & \sum_{i=1}^t \Tr(A_j X_i) \leq t b_j \quad \text{ for } j \in [m^{(1)}], \\
               &\sum_{i=1}^t X_i \succeq 0.
  \end{align*}
  This SDP has the same objective function as $SDP^{(t)}$ but a larger feasible region: every feasible $X_1,\dots,X_t$ for $SDP^{(t)}$ is also feasible for $\widehat{SDP}^{(t)}$. However, by a change of variables, $X := \sum_{i=1}^t X_i$, it is easy to see that $\widehat{SDP}^{(t)}$ is simply a scaled version of $SDP^{(1)}$. So, $\widehat{SDP}^{(t)}$ has optimal value $t\cdot \opt^{(1)}$. Since optimal solutions to $\widehat{SDP}^{(t)}$ are scaled optimal solutions to $SDP^{(1)}$, we have $\hat{R}^{*(t)} = t\cdot R^{* (1)}$. Combining the above, it follows that every optimal solution to $SDP^{(t)}$ is optimal to $\widehat{SDP}^{(t)}$ as well, and hence has trace at least $t\cdot R^{*(1)}$, so $R^{* (t)}\geq t\cdot R^{* (1)}$.

  We conclude that
  \[
    \frac{
      R^{*(t)} r^{*(t)}
    }{
      \eps^{(t)}
    }
    \geq
    \frac{
      tR^{* (1)}
      (r^{* (1)}+(t-1)\delta)
    }{
      t\eps^{(1)}}
    = \Omega\left( t \right)
  \]
  and $n^{(t)} = t n^{(1)}$, $m^{(t)} = t ( m^{(1)}-k)+k$.
\end{proof}
This shows that for many natural SDP formulations for combinatorial problems, such as the MAXCUT relaxation or LPs that have to do with resource management, $R^{*}r^{*}/\eps$ increases linearly in $n$ and $m$ for some instances. Hence, using $R^{*}\leq R$ and Lemma~\ref{lem:genisrstar}, $Rw/\eps$ grows at least linearly when a general width-bound is used.
\section{Lower bounds on the quantum query complexity}\label{sec:lowerbounds}

In this section we will show that every LP-solver (and hence every SDP-solver) that can distinguish two optimal values with high probability needs $\Omega\left(\sqrt{\max\{n,m\}} \left( \min\{n,m\} \right)^{3/2} \right)$ quantum queries in the worst case.

For the lower bound on LP-solving we will give a reduction from a composition of Majority and OR functions.

\begin{definition}
  Given input bits $Z_{ij\ell} \in \{0,1\}^{a\times b\times c}$ the problem of calculating
  \begin{align*}
    MAJ_a(&\\
          &OR_b(MAJ_c(Z_{111},\dots,Z_{11c}),\dots,MAJ_c(Z_{1b1},\dots,Z_{1bc})),\\
          &\dots,\\
          &OR_b(MAJ_c(Z_{a11},\dots,Z_{a1c}),\dots,MAJ_c(Z_{ab1},\dots,Z_{abc}))\\
    )&
  \end{align*}
  with the promise that
  \begin{itemize}
  \item Each inner $\MAJ_c$ is a boundary case, in other words $\sum_{\ell=1}^c \!Z_{ij\ell} \in\! \{c/2,c/2+1\}$ for all $i,j$.
  \item The outer $\MAJ_a$ is a boundary case, in other words, if $\tilde{Z} \in \{0,1\}^a$ is the bitstring that results from all the OR calculations, then $|\tilde{Z}| \in \{a/2,a/2+1\}$.
  \end{itemize}
  is called the promise $\MOM$ problem.
\end{definition}
\pagebreak[3]
\begin{lemma}\label{lem:knownlow}
  It takes at least $\Omega(a\sqrt{b}\,c)$ queries to the input to solve the promise $\MOM\!$ problem.
\end{lemma}
\begin{proof}
  The promise version of $\MAJ_k$ is known to require $\Omega(k)$ quantum queries. Likewise, it is known that the OR$_k$ function requires $\Omega(\sqrt{k})$ queries. Furthermore, the adversary bound tells us that query complexity is multiplicative under composition of functions; Kimmel~\cite[Lemma~A.3 (Lemma~6 in the arXiv version)]{kimmel2011QAdversatryUpperBound} showed that this even holds for promise functions. 
\end{proof}

\begin{lemma} \label{lem:calculation}
  Determining the value
  \[
    \sum_{i=1}^a \max_{j \in \lbrack b\rbrack} \sum_{\ell=1}^c Z_{ij\ell},
  \]
  for a $Z$ from the promise $\MOM$ problem up to additive error $\eps = 1/3$, solves the promise $\MOM$ problem.
\end{lemma}
\begin{proof}
  Notice that due to the first promise, $\sum_{\ell=1}^c Z_{ij\ell} \in \{c/2,c/2+1\}$ for all $i\in [a],j\in[b]$. This implies that
  \begin{itemize}
  \item If the $i$th OR is $0$, then all of its inner MAJ functions are $0$ and hence
    \[
      \max_{j \in \lbrack b\rbrack} \sum_{\ell=1}^c Z_{ij\ell} = \frac{c}{2}.
    \]
  \item If the $i$th OR is $1$, then at least one of its inner MAJ functions is $1$ and hence
    \[
      \max_{j \in \lbrack b\rbrack} \sum_{\ell=1}^c Z_{ij\ell} = \frac{c}{2} + 1.
    \]
  \end{itemize}
  Now, if we denote the string of outcomes of the OR functions by $\tilde{Z}\in \{0,1\}^a$, then
  \[
    \sum_{i=1}^a \max_{j \in \lbrack b\rbrack} \sum_{\ell=1}^c Z_{ij\ell} = a \frac{c}{2} + |\tilde{Z}|.
  \]
  Hence determining the left-hand side will determine $|\tilde{Z}|$; this Hamming weight is either $\frac{a}{2}$ if the full function evaluates to $0$, or $\frac{a}{2}+1$ if it evaluates to $1$.
\end{proof}
\begin{lemma} \label{lem:generallp}
  For an input $Z\in \{0,1\}^{a\times b\times c}$ there is an LP with $m = c+a$ and $n = c+ab$ for which the optimal value is
  \[
    \sum_{i=1}^a \max_{j \in \lbrack b\rbrack} \sum_{\ell=1}^c Z_{ij\ell}.
  \]
  Furthermore, a query to an entry of the input matrix or vector costs at most $1$ query to $Z$.
\end{lemma}
\begin{proof}
  Let $Z^{(i)}$ be the matrix one gets by fixing the first index of $Z$ and putting the entries in a $c\times b$ matrix, so $Z^{(i)}_{\ell j} = Z_{ij\ell}$.
  We define the following LP:
  \begin{align*}
    \opt = \text{max} \ \ \ & \sum_{k=1}^c w_k\\
    \text{s.t.} \ \ \ &
                        \begin{bmatrix}
                          I & -Z^1 & \cdots & -Z^a \\
                          0 & \mathbf{1}^T & &\\
                          0 &  & \ddots &\\
                          0 &  & & \mathbf{1}^T
                        \end{bmatrix}
                                   \begin{bmatrix}
                                     w\\
                                     v^{(1)}\\
                                     \vdots\\
                                     v^{(a)}\\
                                   \end{bmatrix}
    \leq
    \begin{bmatrix}
      0\\
      {1}\\
      \vdots\\
      {1}
    \end{bmatrix}\\
                            & v^1,\dots,v^a \in \mathbb{R}_{+}^b ,w \in \mathbb{R}_{+}^c.
  \end{align*}
  Notice every $Z^{(i)}$ is of size $c\times b$, so that indeed $m = c+a$ and $n = c + ab$.

  For every $i \in [a]$ there is a constraint that says
  \[
    \sum_{j=1}^b v^{(i)}_j \leq 1.
  \]
  The constraints involving $w$ say that for every $k\in[c]$
  \[
    w_k \leq \sum_{i=1}^a \sum_{j=1}^b v^{(i)}_j Z^{(i)}_{k j} = \sum_{i=1}^a  ( Z^{(i)} v^{(i)} )_k,
  \]
  where $( Z^{(i)} v^{(i)} )_k$ is the $k$th entry of the matrix-vector product $Z^{(i)} v^{(i)}$.
  Clearly, for an optimal solution these constraints will be satisfied with equality, since in the objective function $w_k$ has a positive weight.
  Summing over $k$ on both sides, we get the equality
  \begin{align*}
    \opt &= \sum_{k = 1}^c w_{k}\\
         &= \sum_{k = 1}^c \sum_{i=1}^a  ( Z^{(i)} v^{(i)} )_{k}\\
         &= \sum_{i=1}^a \sum_{k = 1}^c   ( Z^{(i)} v^{(i)} )_{k}\\
         &= \sum_{i=1}^a  \nrm{ Z^{(i)} v^{(i)}}_1,
  \end{align*}
  so in the optimum $\nrm{ Z^{(i)} v^{(i)}}_1$ will be maximized.
  Note that we can use the $\ell_1$-norm as a shorthand for the sum over vector elements since all elements are positive.
  In particular, the value of $\nrm{ Z^{(i)} v^{(i)}}_1$ is given by
  \begin{align*}
    \text{max} \ \ \ & \nrm{ Z^{(i)} v^{(i)}}_1\\
    \text{s.t.} \ \ \ &	\nrm{v^{(i)}}_1\leq 1\\
                     & v^{(i)}\geq 0.
  \end{align*}
  Now $\nrm{\smash{ Z^{(i)} v^{(i)}}}_1$ will be maximized by putting all weight in $v^{(i)}$ on the index that corresponds to the column of $Z^{(i)}$ that has the highest Hamming weight.
  In particular in the optimum $\nrm{ \smash{Z^{(i)} v^{(i)}}}_1 = \max_{j\in \lbrack b \rbrack} \sum_{\ell = 1}^c Z^{(i)}_{\ell j}$.
  Putting everything together gives:
  \[
    \opt = \sum_{i=1}^a  \nrm{ Z^{(i)} v^{(i)}}_1 =  \sum_{i=1}^a  \max_{j\in \lbrack b \rbrack} \sum_{\ell = 1}^c Z^{(i)}_{\ell j}  =  \sum_{i=1}^a  \max_{j\in \lbrack b \rbrack} \sum_{\ell = 1}^c Z_{i j \ell}.
  \]
  \vskip-5mm
\end{proof}
\begin{theorem}\label{thm:lowbound}
  There is a family of LPs, with $m \leq n$ and two possible integer optimal values, that require at least $\Omega(\sqrt{n}\,m^{3/2})$ quantum queries to the input to distinguish those two values.
\end{theorem}
\begin{proof}
  Let $a = c = m/2$ and $b = \frac{n-c}{a} = \frac{2n}{m} - 1$, so that $n =c+ab$ and $m=c+a$.
  By Lemma~\ref{lem:generallp} there exists an LP with $n =c+ab$ and $m=c+a$ that calculates
  \[
    \sum_{i=1}^a \max_{j \in \lbrack b\rbrack} \sum_{\ell=1}^c Z_{ij\ell}
  \]
  for an input $Z$ to the promise $\MOM$ problem. By Lemma~\ref{lem:calculation}, calculating this value will solve the promise $\MOM$ problem.
  By Lemma~\ref{lem:knownlow} the promise $\MOM$ problem takes $\Omega(a\sqrt{b}\,c)$ quantum queries in the worst case.
  This implies a lower bound of
  \[
    \Omega\left(m^2\sqrt{\frac{n}{m}}\right) = \Omega(m^{3/2}\sqrt{n})
  \]
  quantum queries on solving these LPs.
\end{proof}

\begin{corollary}
  Distinguishing two optimal values of an LP (and hence also of an SDP) with additive error $\eps<1/2$ requires \[\Omega\left(\sqrt{\max\{n,m\}} \left( \min\{n,m\} \right)^{3/2} \right)\] quantum queries to the input matrices in the worst case.
\end{corollary}
\begin{proof}
  This follows from Theorem~\ref{thm:lowbound} and LP duality.
\end{proof}
It is important to note that the parameters $R$ and $r$ from the Arora-Kale algorithm are not constant in this family of LPs ($R,r = \Theta(\min\{n,m\}^2)$ here), and hence this lower bound does not contradict the scaling with $\sqrt{mn}$ of the complexity of our SDP-solver or Brand\~{a}o and Svore's. Since we show in the appendix that one can always rewrite the LP (or SDP) so that $2$ of the parameters $R,r,\eps$ are constant, the lower bound implies that any algorithm with a sub-linear dependence on $m$ or $n$ has to depend at least polynomially on $Rr/\eps$. For example, the above family of LPs shows that an algorithm with a $\sqrt{mn}$ dependence has to have an $(Rr/\eps)^\kappa$ factor in its complexity with $\kappa\geq 1/4$. It remains an open question whether a lower bound of $\Omega(\sqrt{mn})$ can be proven for a family of LPs/SDPs where $\eps$, $R$ and $r$ all constant.

\section{Conclusion}

In this paper we gave better algorithms and lower bounds for quantum SDP-solvers, improving upon recent work of Brand\~ao and Svore~\cite{brandao2016QSDPSpeedup}. Here are a few directions for future work:
\begin{itemize}
\item {\bf Better upper bounds.}
  The runtime of our algorithm, like the earlier algorithm of Brand\~ao and Svore, has better dependence on $m$ and $n$ than the best classical SDP-solvers, but worse dependence on $s$ and on $Rr/\eps$. In subsequent work (see the introduction), these dependencies have been improved, but there is room for further improvement.
\item {\bf Applications of our algorithm.}
  As mentioned, both our and Brand\~ao-Svore's quantum SDP-solvers only improve upon the best classical algorithms for a specific regime of parameters, namely where $mn\gg Rr/\eps$. Unfortunately, we don't know particularly interesting problems in combinatorial optimization in this regime. As shown in Section~\ref{sec:downside}, many natural SDP formulations will not fall into this regime. However, it would be interesting to find useful SDPs for which our algorithm gives a significant speed-up.
\item {\bf New algorithms.} As in the work by Arora and Kale, it might be more promising to look at oracles (now quantum) that are designed for specific SDPs. Such oracles could build on the techniques developed here, or develop totally new techniques. It might also be possible to speed up other classical SDP solvers, for example those based on interior-point methods.
\item {\bf Better lower bounds.}
  Our lower bounds are probably not optimal, particularly for the case where $m$ and $n$ are not of the same order. The most interesting case would be to get lower bounds that are simultaneously tight in the parameters $m$, $n$, $s$, and $Rr/\eps$.
\end{itemize}

\paragraph{Acknowledgments.}
An earlier version of this paper appeared in the proceedings of the FOCS'17 conference~\cite{apeldoorn2017QSDPSolvers}.

We thank Fernando Brand\~ao for sending us several drafts of~\cite{brandao2016QSDPSpeedup} and for answering our many questions about their algorithms. We thank Stacey Jeffery for pointing us to~\cite{kimmel2011QAdversatryUpperBound}, and Andris Ambainis and Robin Kothari for useful discussions and comments. Thanks to Monique Laurent for finding a bug in an earlier version, which we fixed. We also thank the anonymous referees of FOCS'17 and Quantum for very helpful comments that improved the presentation.

JvA and SG are supported by the Netherlands Organization for Scientific Research (NWO), grant number 617.001.351.
AG and RdW are supported by ERC Consolidator Grant 615307-QPROGRESS. 
RdW is also partially supported by NWO through Gravitation-grant Quantum Software Consortium - 024.003.037, and through QuantERA project Quant\-Algo 680-91-034.

\bibliographystyle{alphaUrlePrint}
\bibliography{Bibliography}
\appendix


\section{Classical estimation of the expectation value \texorpdfstring{$\tr{A\rho}$}{traces}}\label{app:trace}

To provide contrast to Section~\ref{sec:trCalc}, here we describe a \emph{classical} procedure to efficiently estimate $\mathrm{Tr}(A \rho)$ where $A$ is a Hermitian matrix such that $\nrm{A} \leq 1$, and $\rho = \exp(-H)/\tr{\exp(-H)}$ for some Hermitian matrix $H$. The results in this section can be seen as a generalization of~\cite[Section 7]{arora2016CombPrimDualSDP}. The key observation is that if we are given a Hermitian matrix $B\succeq 0$, and if we take a random vector $u = (u_1, \ldots, u_n)$ where $u_i \in \{\pm 1\}$ is uniformly distributed, then, using $\mathbb E[u_i] = 0$, $\mathbb E[u_i^2] = 1$, we have
\begin{align*}
  \mathbb E \left[ u^T \sqrt{B} A \sqrt{B} u \right] &= \mathbb E \left[ \tr{\sqrt{B} A \sqrt{B} u u^T} \right] = \tr{\sqrt{B} A \sqrt{B} \mathbb E[uu^T]} \\
                                                     &=  \tr{\sqrt{B} A \sqrt{B} I} = \mathrm{Tr}(A B).
\end{align*}
We now show that $u^T \sqrt{B} A \sqrt{B} u$ is highly concentrated around its mean by Chebyshev's inequality.
\begin{lemma} \label{lem:JLherm0}
  Given a Hermitian matrix $A$, with $||A|| \leq 1$, a psd matrix $B$, and a parameter $0 < \theta \leq 1$. With probability $1-1/16$, the average of $k = \bigO \left(1/\theta^2\right)$ independent samples from the distribution $u^T \sqrt{B} A \sqrt{B} u$ is at most $\theta \tr{B}$ away from $\tr{AB}$. Here $u = (u_i)$ and each $u_i \in \{\pm 1\}$ is i.i.d.\ uniformly distributed.
\end{lemma}
\begin{proof}
  We let $F_k$ be the random variable $\frac{1}{k} \sum_{i=1}^k (u^{(i)})^T \sqrt{B} A \sqrt{B} u^{(i)}$, where each of the vectors $u^{(i)} \in \{\pm 1\}^n$ is sampled from the distribution described above. By the above it is clear that $\mathbb{E}[F_k] = \mathrm{Tr}(AB)$.

  We will use Chebyshev's inequality which, in our setting, states that for every $t >0$
  \begin{equation} \label{eq:chebyshev}
    \mathrm{Pr}\big( |F_k - \tr{AB}| \geq t \sigma_k \big) \leq \frac{1}{t^2},
  \end{equation}
  here $\sigma_k^2$ is the variance of $F_k$. We will now upper bound the variance of $F_k$. First note that $\mathrm{var}(F_k) = \frac{1}{k} \mathrm{var}(u^T \sqrt{B} A \sqrt{B} u)$.
  It therefore suffices to upper bound the variance $\sigma^2$ of $u^T \sqrt{B} A \sqrt{B} u$. We first write
  \begin{align*}
    \sigma^2 &= \mathrm{var}(u^T \sqrt{B} A \sqrt{B} u) = \mathbb{E} \left[(u^T \sqrt{B} A \sqrt{B} u)^2\right] - \mathbb{E}\left[u^T \sqrt{B} A \sqrt{B} u\right]^2 \\
             &=\underbrace{\mathbb{E} \left[ \sum_{i,j,k,l = 1}^n u_i u_j u_k u_l (\sqrt{B} A \sqrt{B})_{ij}  (\sqrt{B} A \sqrt{B})_{kl}  \right]}_{(*)}  - \mathrm{Tr}(\sqrt{B} A \sqrt{B})^2.
  \end{align*}
  We then calculate $(*)$ using $\mathbb{E}[u_i] = 0$, $\mathbb{E}[u_i^2]=1$, and the independence of the $u_i$'s:
  \begin{align*}
    (*) &= \sum_{i \neq j} (\sqrt{B} A \sqrt{B})_{ij} \left((\sqrt{B} A \sqrt{B})_{ij} + (\sqrt{B} A \sqrt{B})_{ji}\right)  + \sum_{i,k=1}^n (\sqrt{B} A \sqrt{B})_{ii} (\sqrt{B} A \sqrt{B})_{kk}   \\
        &= \sum_{i \neq j} 2(\sqrt{B} A \sqrt{B})_{ij}^2 + \mathrm{Tr}(\sqrt{B} A \sqrt{B})^2.
  \end{align*}
  Therefore, using Cauchy-Schwarz, we have
  \begin{align*}
    \sigma^2 &= (*) -  \mathrm{Tr}(\sqrt{B} A \sqrt{B})^2 = \sum_{i \neq j} 2(\sqrt{B} A \sqrt{B})_{ij}^2 \leq  \sum_{i,j} 2(\sqrt{B} A \sqrt{B})_{ij}^2\\
             &= 2 \mathrm{Tr}((\sqrt{B} A \sqrt{B})^2) = 2 |\langle ABA, B\rangle | \leq 2 |\langle ABA,ABA\rangle|^{1/2} |\langle B,B\rangle|^{1/2} \\
             &= 2 \tr{A^2 B A^2 B}^{1/2} \tr{B^2}^{1/2} \leq 2 \tr{B A^2 B}^{1/2} \tr{B^2}^{1/2} \leq 2 \tr{B^2} \leq 2 \tr{B}^2,
  \end{align*}
  where on the last line we use $\nrm{A} \leq 1$ and $\tr{A Y} \leq \nrm{A} \tr{Y}$ for any $Y \succeq 0$, in particular for $BA^2B$ and $B^2$.

  It follows that $\sigma_k^2 \leq 2 \tr{B}^2/k$. Chebyshev's inequality~\eqref{eq:chebyshev} therefore shows that for $k = \lceil 32 /\theta^2\rceil $ and $t = 4$,
  \[
    \mathrm{Pr}\big( |F_k - \tr{AB}| \geq \theta \tr{B} \big) \leq \frac{1}{16}.
  \]
  \vskip -5mm
\end{proof}
A simple computation shows that the success probability in the above lemma can be boosted to $1-\delta$ by picking the median of $\bigO(\log(1/\delta))$ repetitions. To show this, let $K = \lceil \log(1/\delta)\rceil$ and for each $i \in [K]$ let $(F_k)_i$ be the average of $k$ samples of $u^T \sqrt{B} A \sqrt{B} u$. Let $z_K$ denote the median of those $K$ numbers. We have
\[
  \mathrm{Pr}\big(|z_K - \tr{AB}| \geq \theta \tr{B}\big) = \underbrace{\mathrm{Pr}\big( z_K \geq \theta \tr{B} + \tr{AB} \big)}_{(*)} +  \mathrm{Pr}\big( z_K \leq   - \theta \tr{B} + \tr{AB} \big)
\]
We upper bound $(*)$:
\begin{align*}
  (*) &\leq  \sum_{I \subseteq [K]: |I| \geq K/2} \prod_{i \in I} \mathrm{Pr}\big( (F_k)_i \geq \theta \tr{B} + \tr{AB} \big) \\
      &\leq  (|\{I \subseteq [K]: |I| \geq K/2 \}| ) \left(\frac{1}{16}\right)^{K/2} \\
      &=  2^{K-1}\left(\frac{1}{4}\right)^K \\
      &\leq \frac{1}{2}\left(\frac{1}{2}\right)^{\log(1/\delta)} = \frac{1}{2} \delta.
\end{align*}
Analogously, one can show that $\mathrm{Pr}\big( z_K \leq   - \theta \tr{B} + \tr{AB} \big) \leq \frac{1}{2} \delta$. Hence
\[\mathrm{Pr}\big(|z_K - \tr{AB}| \geq \theta \tr{B}\big) \leq \delta. \]
This proves the following lemma:
\begin{lemma} \label{lem:JLherm}
  Given a Hermitian matrix $A$, with $||A|| \leq 1$, a psd matrix $B$, and parameters $0 < \delta \leq 1/2$ and $0 < \theta \leq 1$. Using $k = \bigO \left(\log(\frac{1}{\delta})/\theta^2\right)$ samples from the distribution $u^T \sqrt{B} A \sqrt{B} u$, one can find an estimate of $\tr{AB}$ that, with probability $1-\delta$, has additive error at most $\theta \tr{B}$. Here $u = (u_i)$ and the $u_i \in \{\pm 1\}$ are i.i.d.\ uniformly distributed.
\end{lemma}
Looking back at Meta-Algorithm~\ref{alg:AKSDP}, we would like to apply the above lemma to $B =\exp(- H)$.\footnote{For ease of notation, we write $H$ for $\eta H$.} Since it is expensive to compute the exponent of a matrix, it is of interest to consider samples from $v^T A v$, where $v$ is an approximation of $\sqrt{B} u$.

Say $\nrm{\sqrt{B} u -v} \leq \kappa$ and, as always, $\nrm{A} \leq 1$. Then
\begin{align*}
  |u^T \sqrt{B} A \sqrt{B} u - v^T Av| &= | u^T \sqrt{B} A \sqrt{B} u - u^T \sqrt{B} A v + u^T \sqrt{B} A v - v^T A v| \\
                                       &\leq | u^T \sqrt{B} A \sqrt{B} u - u^T \sqrt{B} A v| + |u^T \sqrt{B} A v - v^T A v| \\
                                       &=| u^T \sqrt{B} A (\sqrt{B} u-v)| + |(\sqrt{B}u-v)^T A v| \\
                                       &\leq \nrm{u^T \sqrt{B} A} \nrm{\sqrt{B} u -v} + \nrm{\sqrt{B} u -v} \nrm{Av} \\
                                       &\leq \nrm{ \sqrt{B} u} \nrm{A} \kappa + \kappa \nrm{A} \nrm{v} \\
                                       &\leq \kappa ( \nrm{ \sqrt{B} u} + \nrm{v}) \\
                                       &\leq \kappa ( \nrm{ \sqrt{B} u} + \nrm{\sqrt{B} u + v - \sqrt{B}u }) \\
                                       &\leq \kappa \left(\nrm{ \sqrt{B} u}+ \nrm{ \sqrt{B} u} + \nrm{\sqrt{B} u -v}\right) \\
                                       &\leq 2\kappa \nrm{ \sqrt{B} u} + \kappa^2
\end{align*}
Now observe that we are interested in
\[
  \frac{\tr{A \exp(- H)}}{\tr{\exp(-H)}} = \frac{\tr{A \exp(- H + \gamma I)}}{\tr{\exp(-H + \gamma I)}}.
\]
Suppose an upper bound $K$ on $\nrm{H}$ is known, then we can consider $H' = H -K I$ which satisfies $H' \preceq 0$.
It follows that $\nrm{\exp(- H')} \geq 1$ and, with $B = \exp(-H')$, therefore $\nrm{\smash{\sqrt{B}}} \leq \nrm{B} \leq \tr{B}$. Hence, taking $\kappa \leq \min\{\theta/\nrm{u},\nrm{\smash{\sqrt{B}}u}\} = \theta/ \nrm{u}$,\footnote{Here we assume that $\theta \leq 1$. Then, since $\lambda_{\min}(B) \geq 1$, we trivially have $\theta/\nrm{u} \leq 1/\sqrt{n} \leq \sqrt{n} \leq \nrm{\smash{\sqrt{B}} u}$.} we find
\[
  |u^T \sqrt{B} A \sqrt{B} u - v^T Av| \leq 2\kappa \nrm{ \sqrt{B} u} + \kappa^2  \leq 3 \kappa \nrm{\sqrt{B} u} \leq 3 \kappa \nrm{\sqrt{B}} \nrm{u} \leq 3 \theta \nrm{B} \leq 3 \theta \tr{B}.
\]
This shows that the additional error incurred by sampling from $v^T Av$ is proportional to $\theta \tr{B}$.
Finally, a $\kappa$-approximation of $\smash{\sqrt{B}}u$, with $\kappa = \theta/\nrm{u}$ can be obtained by using the truncated Taylor series of $\exp(-H'/2)$ of degree $p = \max\{2e \nrm{H'}, \log\left(\frac{\sqrt{n}}{\theta}\right) \}$:
\begin{align*}
  \nrm{\exp(-H'/2)  - \sum_{i = 0}^p \frac{(H'/2)^i}{i!}   } &= \nrm{\sum_{i = p+1}^\infty \frac{(H'/2)^i}{i!} } \leq \sum_{j=p+1}^\infty \frac{\nrm{H'}^j}{j!} \leq \sum_{j=p+1}^\infty \left(\frac{e \nrm{H'}}{j}\right)^j \\
                                                             &\leq \left(\frac{e \nrm{H'}}{p+1}\right)^{p+1} \frac{1}{1-\left(\frac{e\nrm{H'}}{(p+1)}\right)} \leq \left(\frac{1}{2}\right)^p = \theta/\sqrt{n},
\end{align*}
\begin{lemma} \label{lem:23}
  Given a Hermitian $s$-sparse matrix $A$, with $||A|| \leq 1$, a psd matrix $B = \exp(-H)$ with $H \preceq 0$, for a $d$-sparse $H$, and parameters $0 < \delta \leq 1/2$ and $0 < \theta \leq 1$. With probability $1-\delta$, using $k = \bigO \left(\log(\frac{1}{\delta})/\theta^2\right)$ samples from the distribution $v^T A v$, one can find an estimate that is at most $\theta \tr{B}$ away from $\tr{AB}$. Here
  \[
    v = \sum_{i=0}^p \frac{(H/2)^i}{i!} u
  \]
  where $p = \bigO(\max\{\nrm{H},\log\left(\frac{\sqrt{n}}{\theta}\right) \})$, and  $u = (u_j)$ where the $u_j \in \{\pm 1\}$ are i.i.d. uniformly distributed.
\end{lemma}

\begin{lemma}
  Given $m$ Hermitian $s$-sparse $n \times n$ matrices $A_1=I,A_2, \ldots, A_m$, with $\nrm{A_j} \leq 1$ for all $j$, a Hermitian $d$-sparse $n \times n$ matrix $H$ with $\nrm{H} \leq K$, and parameters $0 < \delta \leq 1/2$ and $0 < \theta \leq 1$.
  With probability $1-\delta$, we can compute $\theta$-approximations $a_1, \ldots, a_m$ of $\tr{A_1 B}/\tr{B}, \ldots, \tr{A_m B}/\tr{B}$ where $B = \exp(-H)$, using
  \[
    \bigO\left(\frac{\log(\frac{m}{\delta})}{\theta^2} \max\left\{K,\log\left(\frac{\sqrt{n}}{\theta}\right)\right\} d n  + \frac{\log(\frac{m}{\delta})}{\theta^2}  m s n\right)
  \]
  queries to the entries of $A_1, \ldots, A_m$, $H$ and arithmetic operations.

\end{lemma}
\begin{proof}
  As observed above, for every matrix $A$,
  \[
    \frac{\tr{A \exp(- H)}}{\tr{\exp(- H)}} = \frac{\tr{A \exp(- H + \gamma I)}}{\tr{\exp(- H + \gamma I)}}.
  \]
  Lemma~\ref{lem:23} states that for $B' = \exp(- H + K I)$ and $A \in \{A_1, \ldots, A_m\}$ using $k = \bigO \left(\log(\frac{m}{\delta})/\theta^2\right)$ samples from the distribution $v^T A v$, one can find an estimate that is at most $\theta \tr{B}$ away from $\tr{AB}$ with probability $1-\delta/m$. Here
  \[
    v = \sum_{i=0}^p \frac{((H-KI)/2)^i}{i!} u
  \]
  where $p = \bigO(\max\{K,\log\left(\frac{\sqrt{n}}{\theta}\right) \})$, and  $u = (u_j)$ with the $u_j \in \{\pm 1\}$ i.i.d. uniformly distributed.

  Observe that the $k$ samples from $v^T A v$ are really obtained from $k$ samples of vectors $u = (u_j)$ combined with some post-processing, namely obtaining $v = \sum_{i=0}^p \frac{((H-KI)/2)^i}{i!} u$ and two more sparse matrix vector products.

  We can therefore obtain $k$ samples from each of $v^T A_1 v, \ldots, v^T A_m v$ by \emph{once} calculating $k$ vectors $v = \sum_{i=0}^p \frac{((H-KI)/2)^i}{i!} u$, and then, for each of the $m$ matrices $A_j$ computing the $k$ products $v^T A_j v$.
  The $k$ vectors $v$ can be constructed using
  \[
    \bigO\left(\frac{\log(\frac{m}{\delta})}{\theta^2} \max\left\{K,\log\left(\frac{\sqrt{n}}{\theta}\right)\right\}d n\right)
  \]
  queries to the entries of $H$ and arithmetic operations.
  The $mk$ matrix vector products can be computed using
  \[
    \bigO\left( \frac{\log(\frac{m}{\delta})}{\theta^2}  m s n \right)
  \]
  arithmetic operations and queries to the entries of $A_1, \ldots, A_m$ and $H$. This leads to total complexity
  \[
    \bigO\left(\frac{\log(\frac{m}{\delta})}{\theta^2} \max\left\{K,\log\left(\frac{\sqrt{n}}{\theta}\right)\right\} d n  + \frac{\log(\frac{m}{\delta})}{\theta^2}  m s n\right)
  \]
  for computing $k$ samples from each of $v^TA_1 v, \ldots, v^TA_m v$.

  The results of Lemma~\ref{lem:23} say that for each $j$, using those $k$ samples of $v^T A_j v$ we can construct a $\theta \tr{B'}/4$-approximation $a_j'$ of $\tr{A_j B'}$, with probability $1- \delta/(2m)$. Therefore, by a union bound, with probability $1-\delta/2$ we can construct $\theta \tr{B'}/4$-approximations $a_1', \ldots, a_m'$ of $\tr{A_1 B'}, \ldots, \tr{A_m B'}$.
  Therefore, for each $j$, with probability at least $1-\delta$, by Lemma~\ref{lemma:trTogether} we have that $a_j = a_j'/a_1'$ is a $\theta$-approximation of $\tr{A_j B'}/\tr{B'}$, and hence it is a $\theta$-approximation of $\tr{A_j B}/\tr{B}$.
\end{proof}

\section{Implementing smooth functions of Hamiltonians}\label{apx:LowWeight}

In this appendix we show how to efficiently implement smooth functions of a given Hamiltonian.
First we explain what we mean by a function of a Hamiltonian $H\in\mathbb{C}^{n\times n}$, i.e., a Hermitian matrix. Since Hermitian matrices are diagonalizable using a unitary matrix, we can write $H=U^\dagger \text{diag}(\lambda) U$, where $\lambda\in\mathbb{R}^n$ is the vector of eigenvalues. Then for a function $f:\mathbb{R}\rightarrow\mathbb{C}$ we define $f(H):=U^\dagger \text{diag}(f(\lambda)) U$ with a slight abuse of notation, where we apply $f$ to the eigenvalues in $\lambda$ one-by-one. Note that if we approximate $f$ by $\tilde{f}$, then $\nrm{\tilde{f}(H)-f(H)}=\nrm{\text{diag}(\tilde{f}(\lambda))-\text{diag}(f(\lambda))}$. Suppose $D\subseteq\mathbb{R}$ is such that $\lambda\in D^n$, then we can upper bound this norm by the maximum of $|\tilde{f}(x)-f(x)|$ over $x\in D$. Finally we note that $D=[-\nrm{H},\nrm{H}]$ is always a valid choice.

\medskip

The main idea of the method presented below, is to implement a map $\tilde{f}(H)$, where $\tilde{f}$ is a good (finite) Fourier approximation of $f$ for all $x\in[-\nrm{H},\nrm{H}]$. The novelty in our approach is that we construct a Fourier approximation based on some polynomial approximation. In the special case, when $f$ is analytic and $\nrm{H}$ is less than the radius of convergence of the Taylor series, we can obtain good polynomial approximation functions simply by truncating the Taylor series, with logarithmic dependence on the precision parameter. Finally we implement the Fourier series using Hamiltonian simulation and the Linear Combination of Unitaries (LCU) trick~\cite{childs2012HamSimLCU,berry2014HamSimTaylor,berry2015HamSimNearlyOpt}.

This approach was already used in several earlier papers, particularly in~\cite{childs2015QLinSysExpPrec,chowdhury2016QGibbsSampling}.
There the main technical difficulty was to obtain a good truncated Fourier series.
This is a non-trivial task, since on top of the approximation error, one needs to optimize two other parameters of the Fourier approximation that determine the complexity of implementation, namely:

$\bullet$ the largest time parameter $t$ that appears in some Fourier term $e^{-itH}$, and

$\bullet$ the total weight of the coefficients, by which we mean the $1$-norm of the vector of coefficients.
\\ \noindent Earlier works used clever integral approximations and involved calculus to construct a good Fourier approximation for a specific function~$f$. We are not aware of a general result.

In contrast, our  Theorem~\ref{thm:Taylor} and Corollary~\ref{cor:patched} avoids the usage of any integration. It obtains a low-weight Fourier approximation function using the Taylor series. The described method is completely general, and has the nice property that the maximal time parameter $t$ depends logarithmically on the desired approximation precision. Since it uses the Taylor series, it is easy to apply to a wide range of smooth functions.

The circuit we describe for the implementation of the linear operator $f(H): \mathbb{C}^n \rightarrow \mathbb{C}^n$ is going to depend on the specific function~$f$, but not on $H$; the $H$-dependence is only coming from Hamiltonian simulation.
Since the circuit for a specific $f$ can be constructed in advance, we do not need to worry about the (polynomial) cost of constructing the circuit, making the analysis simpler.
When we describe gate complexity, we count the number of two-qubit gates needed for a quantum circuit implementation, just as in Section~\ref{sec:upperbounds}.

Since this appendix presents stand-alone results, here we will deviate slightly from the notation used throughout the rest of the paper, to conform to the standard notation used in the literature (for example, $\eps$, $r$, $\theta$ and $a$ have a different meaning in this appendix). For simplicity we also assume, that the Hamiltonian $H$ acts on $\mathbb{C}^{n}$, where $n$ is a power of 2. Whenever we write $\log(\text{formula})$ in some complexity statement we actually mean $\log_2(2+\text{formula})$ in order to avoid incorrect near-$0$ or even negative expressions in complexity bounds that would appear for small values of the formula.

\paragraph{Hamiltonian simulation.}
We implement each term in a Fourier series using a Hamiltonian simulation algorithm, and combine the terms using the LCU Lemma. Specifically we use~\cite{berry2015HamSimNearlyOpt}, but in fact our techniques would work with any kind of Hamiltonian simulation algorithm.\footnote{For example there is a more recent method for Hamiltonian simulation~\cite{low2016HamSimQubitization,low2016HamSimQSignProc} that could possibly improve on some of the log factors we get from~\cite{berry2015HamSimNearlyOpt}, but one could even consider completely different input models allowing different simulation methods.} The following definition describes what we mean by controlled Hamiltonian simulation.

\begin{definition}\label{def:controlledSim}
  Let $M=2^J$ for some $J\in \mathbb{N}$, $\gamma\in\mathbb{R}$ and $\epsilon\geq0$. We say that the unitary
  $$
  W:=\sum_{m=-M}^{M-1}\ketbra{m}{m}\otimes e^{im\gamma H}
  $$
  implements controlled $(M,\gamma)$-simulation of the Hamiltonian $H$, where $\ket{m}$ denotes a (signed) bitstring $\ket{b_Jb_{J-1}\ldots b_0}$ such that $m=-b_J2^J+\sum_{j=0}^{J-1}b_j2^j$. The unitary $\tilde{W}$ implements controlled $(M,\gamma,\eps)$-simulation of the Hamiltonian $H$, if
  $$
  \nrm{\tilde{W}-W}\leq \eps.
  $$
\end{definition}

Note that in this definition we assume that both positive and negative powers of $e^{iH}$ are simulated. This is necessary for our Fourier series, but sometimes we use only positive powers, e.g., for phase estimation; in that case we can simply ignore the negative powers.

The following lemma is inspired by the techniques of~\cite{childs2015QLinSysExpPrec}. It calculates the cost of such controlled Hamiltonian simulation in terms of queries to the input oracles~\eqref{eq:oracleind}-\eqref{eq:oraclemat} as described in Section~\ref{sec:upperbounds}.

\begin{lemma}\label{lemma:controlledHamsin}
  Let $H\in\mathbb{C}^{n\times n}$ be a $d$-sparse Hamiltonian. Suppose we know an upper bound $K\in\mathbb{R_+}$ on the norm of $H$, i.e., $\nrm{H}\leq K$, and let $\tau:=M\gamma K$. If $\eps>0$ and $\gamma=\Omega(1/(Kd))$, then a controlled $(M,\gamma,\eps)$-simulation of $H$ can be implemented using $\bigO(\tau d\log(\tau /\eps)/\log\log(\tau /\eps))$ queries and $\bigO\left(\tau d\log(\tau /\eps)/\log\log(\tau /\eps)\left[\log(n)+\log^{\frac{5}{2}}(\tau /\eps)\right]\right)$
  gates.
\end{lemma}
\begin{proof}
  We use the results of~\cite[Lemma 9-10]{berry2015HamSimNearlyOpt}, which tell us that a $d$-sparse Hamiltonian $H$ can be simulated for time $t$ with $\eps$ precision in the operator norm using
  \begin{equation}\label{eq:hamSimQueries}
    \bigO\left(\left(t\maxnrm{H}d+1\right)\frac{\log\left(t\nrm{H}/\eps\right)}{\log\log\left(t\nrm{H}/\eps\right)}\right)
  \end{equation}
  queries and gate complexity
  \begin{equation}\label{eq:hamSimGates}
    \bigO\left(\left(t\maxnrm{H}d+1\right)\frac{\log\left(t\nrm{H}/\eps\right)}{\log\log\left(t\nrm{H}/\eps\right)}\left[\log(n)+\log^{\frac{5}{2}}\left(t\nrm{H}/\eps\right)\right]\right).
  \end{equation}

  Now we use a standard trick to remove log factors from the implementation cost, and write the given unitary $W$ as the product of some increasingly precisely implemented controlled Hamiltonian simulation unitaries.
  For $b\in \{0,1\}$ let us introduce the projector $\ketbra{b}{b}_j:=I_{2^j}\otimes\ketbra{b}{b}\otimes I_{2^{J-j}}$,
  where $J=\log(M)$. Observe that
  \begin{equation}\label{eq:cleverW}
    W=\left(\ketbra{1}{1}_{J}\otimes e^{-i2^{J}\gamma H}+\ketbra{0}{0}_{J}\otimes I\right)\prod_{j=0}^{J-1}\left(\ketbra{1}{1}_{j}\otimes e^{i2^{j}\gamma H}+\ketbra{0}{0}_{j}\otimes I\right).
  \end{equation}

  The $j$-th operator $e^{\pm i2^{j}\gamma H}$ in the product~\eqref{eq:cleverW} can be implemented with $2^{j-J-1}\epsilon$ precision using
  $\bigO\left(2^j\gamma Kd\log\left(\frac{2^j\gamma K}{\eps 2^{j-J-1}}\right)/\log\log\left(\frac{2^j\gamma K}{\eps 2^{j-J-1}}\right)\right)
  =\bigO(2^j\gamma Kd\log(\tau/\epsilon)/\log\log(\tau/\epsilon))$ queries by~\eqref{eq:hamSimQueries} and using $\bigO(2^jd\log(\tau/\epsilon)/\log\log(\tau/\epsilon)[\log(n)+\log^{\frac{5}{2}}(\tau/\epsilon)])$ gates by~\eqref{eq:hamSimGates}.
  Let us denote by $\tilde{W}$ the concatenation of all these controlled Hamiltonian simulation unitaries.
  Adding up the costs we see that our implementation of $\tilde{W}$ uses $\bigO(\tau d\log(\tau/\eps)/\log\log(\tau/\eps))$ queries and has gate complexity $\bigO\left(\tau d\log(\tau /\eps)/\log\log(\tau/\eps)\left[\log(n)+\log^{\frac{5}{2}}(\tau /\eps)\right]\right)$.
  Using the triangle inequality repeatedly, it is easy to see that $\nrm{W\!-\!\tilde{W}}\leq\sum_{j=0}^{J}2^{j-J-1}\eps \leq \eps$.
\end{proof}

\subsection{Implementation of smooth functions of Hamiltonians: general results}

The first lemma we prove provides the basis for our approach. It shows how to turn a polynomial approximation of a function $f$ on the interval $[-1,1]$ into a nice Fourier series in an efficient way, while not increasing the weight of coefficients. This is useful, because we can implement a function given by a Fourier series using the LCU Lemma, but only after scaling it down with the weight of the coefficients.

\begin{lemma}\label{lemma:LowWeightAPX}
  Let $\delta,\eps\in\!(0,1)$ and $f:\mathbb{R}\rightarrow \mathbb{C}$ s.t. $\left|f(x)\!-\!\sum_{k=0}^K a_k x^k\right|\leq \eps/4$ for all $x\in\![-1+\delta,1-\delta]$.
  Then $\exists\, c\in\mathbb{C}^{2M+1}$ such that
  $$
  \left|f(x)-\sum_{m=-M}^M c_m e^{\frac{i\pi m}{2}x}\right|\leq \eps
  $$
  for all $x\in\![-1+\delta,1-\delta]$, where $M=\max\left(2\left\lceil \ln\left(\frac{4\nrm{a}_1}{\eps}\right)\frac{1}{\delta} \right\rceil,0\right)$ and $\nrm{c}_1\leq \nrm{a}_1$. Moreover $c$ can be efficiently calculated on a classical computer in time $\text{poly}(K,M,\log(1/\eps))$.
\end{lemma}

\begin{proof} Let us introduce the notation $\nrm{f}_\infty=\sup\{|f(x)|: x\in\![-1+\delta,1-\delta]\}$.
  First we consider the case when $\nrm{a}_1<\eps/2$. Then $\nrm{f}_\infty\leq \nrm{f(x)-\sum_{k=0}^K a_k x^k}_\infty + \nrm{\sum_{k=0}^K a_k x^k}_\infty < \eps/4 + \eps/2 < \eps$. So in this case the statement holds with $M=0$ and $c=0$, i.e., even with an empty sum.

  From now on we assume $\nrm{a}_1\geq \eps/2$. We are going to build up our approximation gradually. Our first approximate function $\tilde{f}_1(x):=\sum_{k=0}^K a_k x^k$ satisfies $\nrm{f-\tilde{f}_1}_\infty\!\!\leq \eps/4$ by assumption.
In order to construct a Fourier series, we will work towards a linear combination of sines. To that end, note that $\forall x\in[-1,1]$: $\tilde{f}_1(x)\!=\!\sum_{k=0}^K a_k \left(\frac{\arcsin\left(\sin(x\pi/2 )\right)}{\pi/2}\right)^{\!k}\!$. Let $b^{(k)}$ denote the series of coefficients such that $\left(\frac{\arcsin(y)}{\pi/2}\right)^k=\sum_{\ell=0}^{\infty}b_\ell^{(k)} y^\ell$ for all $y\in[-1,1]$. For $k=1$ the coefficients are just $\frac{2}{\pi}$ times the coefficients of the Taylor series of $\arcsin$ so we know that $b^{(1)}_{2\ell}=0$ while $b^{(1)}_{2\ell+1}=\binom{2\ell}{\ell}\frac{2^{-2\ell}}{2\ell+1}\frac{2}{\pi}$. Since $\left(\frac{\arcsin(y)}{\pi/2}\right)^{\!k+1}\!\!=\left(\frac{\arcsin(y)}{\pi/2}\right)^{\!k}\left(\sum_{\ell=0}^{\infty}b_\ell^{(1)} y^\ell\right)$, we obtain the formula $b^{(k+1)}_{\ell}=\sum_{\ell'=0}^{\ell}b^{(k)}_{\ell'}b^{(1)}_{\ell-\ell'}$, so one can recursively calculate each $b^{(k)}$. As $b^{(1)}\geq 0$ one can use the above identity inductively to show that $b^{(k)}\geq 0$.
  Therefore $\nrm{b^{(k)}}_1=\sum_{\ell=0}^{\infty}b_\ell^{(k)} 1^\ell=\left(\frac{\arcsin(1)}{\pi/2}\right)^k=1$.
  Using the above definitions and observations we can rewrite
  $$\forall x\in[-1,1]: \tilde{f}_1(x)=\sum_{k=0}^K a_k \sum_{\ell=0}^{\infty} b^{(k)}_\ell \sin^\ell(x\pi/2).$$
  To obtain the second approximation function, we want to truncate the summation over $\ell$ at $L=\ln\left(\frac{4\nrm{a}_1}{\eps}\right)\frac{1}{\delta^2}$ in the above formula. We first estimate the tail of the sum. We are going to use that for all $\delta\in [0,1]$: $\sin((1-\delta)\pi/2)\leq 1 - \delta^2$. For all $k\in\mathbb{N}$ and $x\in\![-1+\delta,1-\delta]$ we have:
  \begin{align*}
    \left|\sum_{\ell=\lceil L\rceil}^{\infty} b^{(k)}_\ell \sin^\ell(x\pi/2)\right|
    &\leq \sum_{\ell=\lceil L\rceil}^{\infty} b^{(k)}_\ell \left|\sin^\ell(x\pi/2)\right| \\
    &\leq \sum_{\ell=\lceil L\rceil}^{\infty} b^{(k)}_\ell \left|1-\delta^2\right|^\ell\\
    &\leq \left(1-\delta^2\right)^L\sum_{\ell=\lceil L\rceil}^{\infty} b^{(k)}_\ell \\
    &\leq \left(1-\delta^2\right)^L\\
    &\leq e^{-\delta^2L} \\
    &= \frac{\eps}{4\nrm{a}_1}.
  \end{align*}
  Thus we have $\nrm{\tilde{f}_1-\tilde{f}_2}_\infty\leq \eps/4$ for
  $$\tilde{f}_2(x):=\sum_{k=0}^K a_k \sum_{\ell=0}^{\lfloor L \rfloor} b^{(k)}_\ell \sin^\ell(x\pi/2).$$
  To obtain our third approximation function, we will approximate $\sin^\ell(x\pi/2)$. First observe that
  \begin{equation}\label{eq:sinl}
    \sin^\ell(z)=\left(\frac{e^{-iz}-e^{iz}}{-2i}\right)^{\!\ell}
    =\left(\frac{i}{2}\right)^{\!\ell}\sum_{m=0}^{\ell}(-1)^{m}\binom{\ell}{m}e^{iz(2m-\ell)}
  \end{equation}
  which, as we will show (for $M'$ much larger than $\sqrt{\ell}$)
  is very well approximated by
  $$\left(\frac{i}{2}\right)^{\!\ell}\sum_{m=\lceil\ell/2\rceil-M'}^{\lfloor\ell/2\rfloor+M'}(-1)^{m}\binom{\ell}{m}e^{iz(2m-\ell)}. $$
  Truncating the summation in~\eqref{eq:sinl} based on this approximation reduces the maximal time evolution parameter (i.e., the maximal value of the parameter $t$ in the $\exp(izt)$ terms) quadratically.
  To make this approximation precise, we use Chernoff's inequality~\cite[A.1.7]{alon2008ProbMethod} for the binomial distribution, or more precisely its corollary for sums of binomial coefficients, stating
  $$ \sum_{m=\lceil\ell/2+M'\rceil}^{\ell}2^{-\ell}\binom{\ell}{m}
  \leq e^{-\frac{2 (M')^2}{\ell}}.$$
  Let $M'=\left\lceil \ln\left(\frac{4\nrm{a}_1}{\eps}\right)\frac{1}{\delta} \right\rceil$ and suppose $\ell\leq L$, then this bound implies that
  \begin{equation}\label{eq:binomBound}
    \sum_{m=0}^{\lfloor\ell/2\rfloor-M'}2^{-\ell}\binom{\ell}{m}
    =\sum_{m=\lceil\ell/2\rceil+M'}^{\ell}2^{-\ell}\binom{\ell}{m}
    \leq e^{-\frac{2 (M')^2}{\ell}}
    \leq e^{-\frac{2 (M')^2}{L}}
    \leq \left(\frac{\eps}{4\nrm{a}_1}\right)^2
    \leq \frac{\eps}{4\nrm{a}_1},
  \end{equation}
  where for the last inequality we use the assumption $\eps \leq 2 \nrm{a}_1$.
  By combining~\eqref{eq:sinl} and~\eqref{eq:binomBound} we get that for all $\ell \leq L$
  \begin{equation*}
    \nrm{\sin^\ell(z)-\left(\frac{i}{2}\right)^{\!\ell}\sum_{m=\lceil\ell/2\rceil-M'}^{\lfloor\ell/2\rfloor+M'}(-1)^{m}\binom{\ell}{m}e^{iz(2m-\ell)}}_\infty\leq \frac{\eps}{2\nrm{a}_1}.
  \end{equation*}
  Substituting $z=x\pi/2$ into this bound we can see that $\nrm{\tilde{f}_2-\tilde{f}_3}_\infty\leq \eps/2$, for
  \begin{equation}\label{eq:finalFourier}
    \tilde{f}_3(x):=\sum_{k=0}^K a_k \sum_{\ell=0}^{\lfloor L \rfloor} b^{(k)}_\ell \left(\frac{i}{2}\right)^\ell\sum_{m=\lceil\ell/2\rceil-M'}^{\lfloor\ell/2\rfloor+M'}(-1)^{m}\binom{\ell}{m}e^{\frac{i\pi x}{2}(2m-\ell)},
  \end{equation}
  using $\sum_{k=0}^K |a_k| \sum_{\ell=0}^{\lfloor L \rfloor} \left|b^{(k)}_\ell\right|\leq \sum_{k=0}^K |a_k|=\nrm{a}_1$. Therefore we can conclude that $\tilde{f}_3$ is an $\eps$-approximation to $f$:
  \[
    \nrm{f-\tilde{f}_3}_\infty\leq \nrm{f-\tilde{f}_1}_\infty + \nrm{\tilde{f}_1-\tilde{f}_2}_\infty + \nrm{\tilde{f}_2-\tilde{f}_3}_\infty \leq \eps.
  \]
  Observe that in~\eqref{eq:finalFourier} the largest value of $|m-\ell|$ in the exponent is upper bounded by $2M'=M$.
  So by rearranging the terms in $\tilde{f}_3$ we can write $\tilde{f}_3(x)=\sum_{m=-M}^{M}c_{m} e^{\frac{i\pi m}{2}x}$.
  Now let us fix a value $k$ in the first summation of~\eqref{eq:finalFourier}. Observe that after taking the absolute value of each term, the last two summations still yield a value $\leq 1$, since $\nrm{b^{(k)}}_1=1$ and $\sum_{m=0}^{\ell}\binom{\ell}{m}=2^\ell$.
  It follows that $\nrm{c}_1\leq \nrm{a}_1$. From the construction of the proof, it is easy to see that (an $\eps$-approximation of) $c$ can be calculated in time $\text{poly}(K,M,\log(1/\eps))$.
\end{proof}

Now we present the Linear Combination of Unitaries (LCU) Lemma~\cite{childs2012HamSimLCU,berry2014HamSimTaylor,berry2015HamSimNearlyOpt}, which we will use for combining the Fourier terms in our quantum circuit. Since we intend to use LCU for implementing non-unitary operations, we describe a version without the final amplitude amplification step. We provide a short proof for completeness.

\begin{lemma}[LCU Lemma~\cite{childs2012HamSimLCU,berry2014HamSimTaylor,berry2015HamSimNearlyOpt}] \label{lemma:LCU}
  Let $U_1,U_2,\ldots,U_m$ be unitaries on a Hilbert space $\mathcal{H}$, and $L=\sum_{i=1}^{m}a_i U_i$, where $a\in\mathbb{R}_+^{m}\setminus\{0\}$.
  Let	$V=\sum_{i=1}^{m}\ket{i}\!\bra{i}\otimes U_i$ and $A\in\mathbb{C}^{m\times m}$ be a unitary such that $A\ket{0}=\sum_{i=1}^{m}\sqrt{\frac{a_i}{\nrm{a}_1}}\ket{i}$.
  Then $\frac{L}{\nrm{a}_1}=\left(\bra{0}\otimes I\right) \left(A^\dagger\otimes I\right) V \left(A\otimes I\right)\left(\ket{0}\otimes I\right)$, i.e., for every $\ket{\psi}\in \mathcal{H}$ we have $\left(A^\dagger\otimes I\right) V \left(A\otimes I\right)\ket{0}\ket{\psi}=\ket{0}\frac{L}{\nrm{a}_1}\ket{\psi}+\ket{\Phi^\perp}$, where the vector $\ket{\Phi^\perp}$ satisfies $\left(\ketbra{0}{0}\otimes I\right)\ket{\Phi^\perp}=0$.
\end{lemma}
\begin{proof}
  \begin{align*}
    \left(\bra{0}\otimes I\right)(A^\dagger\otimes I) V \left(A\otimes I\right)\ket{0}\ket{\psi}
    &= \left(\left(\sum_{i=1}^{m}\sqrt{\frac{a_i}{\nrm{a}_1}}\bra{i}\right)\otimes I\right)V \sum_{i=1}^{m}\sqrt{\frac{a_i}{\nrm{a}_1}}\ket{i}\ket{\psi}\\
    &= \left(\left(\sum_{i=1}^{m}\sqrt{\frac{a_i}{\nrm{a}_1}}\bra{i}\right)\otimes I\right) \sum_{i=1}^{m}\sqrt{\frac{a_i}{\nrm{a}_1}}\ket{i}U_i\ket{\psi}\\
    &= \sum_{i=1}^{m}\frac{a_i}{\nrm{a}_1}U_i\ket{\psi}\\
    &= \frac{L}{\nrm{a}_1}\ket{\psi}
  \end{align*}
  \vskip-5mm
\end{proof}

The next result summarizes how to efficiently implement a Fourier series of a Hamiltonian.
\begin{lemma}\label{lemma:LCUApplied}
  Suppose $f(x)=\sum_{m=-M}^{M-1} c_m e^{im\gamma x}$, for a given $c\in\mathbb{C}^{2M}\setminus\{0\}$. We can construct a unitary $\tilde{U}$ which implements the operator $\frac{f(H)}{\nrm{c}_1}=\sum_{m=-M}^{M-1} \frac{c_m}{\nrm{c}_1} e^{im\gamma H}$ with $\eps$ precision, i.e., such that $$\nrm{(\bra{0}\otimes I)\tilde{U}(\ket{0}\otimes I)-\frac{f(H)}{\nrm{c}_1}}\leq\eps,$$
  using $\bigO(M(\log(M)+1))$ two-qubit gates and a single use of a circuit implementing controlled $(M,\gamma,\eps)$-simulation of $H$.
\end{lemma}
\begin{proof}
  This is a direct corollary of Lemma~\ref{lemma:LCU}.
  To work out the details, note that we can always extend $c$ with some $0$ values, so we can assume without loss of generality that $M$ is a power of $2$. This is useful, because then we can represent each $m\in [-M,M-1]$ as a $(J+1)$-bit signed integer for $J=\log(M)$.

  The implementation of the operator $A$ in Lemma~\ref{lemma:LCU} does not need any queries and it can be constructed exactly using $\bigO(M(\log(M)+1))$ two-qubit gates, e.g., by the techniques of~\cite{grover2002SuperposEffIntegrProbDistr}. We sketch the basic idea which is based on induction. For $J=1$ the operator $A$ is just a two-qubit unitary. Suppose we proved the claim for bitstrings of length $J$ and want to prove the claim for length $J+1$. Let $a\in\mathbb{R}_+^{2^{J+1}}$ be such that $\nrm{a}_1=1$ and define $\tilde{a}\in\mathbb{R}_+^{2^{J}}$ such that $\tilde{a}_b=a_{b,0}+a_{b,1}$ for all bitstrings $b\in\{0,1\}^{J}$. Then we have a circuit $\tilde{A}$ that uses $\bigO(2^J (J+1))$ gates and satisfies $\sqrt{\tilde{a}_b}=\bra{b}\tilde{A}\ket{0\ldots 0}$ for all $b\in\{0,1\}^{J}$. We can add an extra $\ket{0}$-qubit and implement a controlled rotation gate $R_b$ on it for each $b\in\{0,1\}^{J}$. Let $R_b$ have rotation angle $\arccos\left(\sqrt{a_{b,0}/\tilde{a}_b}\right)$ and be controlled by $b$. It is easy to see that the new unitary $A$ satisfies $\sqrt{a_{b'}}=\bra{b'}A\ket{0\ldots 0}$ for each $b'\in\{0,1\}^{J}$. Each $R_b$ can be implemented using $\bigO(J)$ two-qubit gates and ancilla qubits, justifying the gate complexity and concluding the induction.

  What remains is to implement the operator $V\kern-0.2mm=\sum_{m=-M}^{M-1}\ketbra{m}{m}\otimes \frac{c_m}{|c_m|} e^{im\gamma H}$ from Lemma~\ref{lemma:LCU}. We implement $V=PW$ in two steps, where $P=\sum_{m=-M}^{M-1}\ketbra{m}{m}\otimes \frac{c_m}{|c_m|}I$. This $P$ can be implemented exactly using $\bigO(M(\log(M)+1))$ gates simply by building a controlled gate that adds the right phase for each individual bitstring. Since the bitstring on which we want to do a controlled operation has length  $\log(M)+1$, each controlled operation can be constructed using $\bigO(\log(M)+1)$ gates and ancilla qubits resulting in the claimed gate complexity. We use a circuit implementing controlled $(M,\gamma,\eps)$-simulation of $H$, denoted by~$\tilde{W}$, which is an $\eps$-approximation of $W$ by definition.

  Finally $\tilde{U}:=(A^\dagger\otimes I) P\tilde{W} \left(A\otimes I\right)$. This yields an $\eps$-precise implementation, since
  \begin{align*}
    \nrm{(\bra{0}\otimes I)\tilde{U}(\ket{0}\otimes I)-\frac{f(H)}{\nrm{c}_1}}
    &= \nrm{(\bra{0}\otimes I)\tilde{U}(\ket{0}\otimes I)-(\bra{0}\otimes I)(A^\dagger\otimes I) PW \left(A\otimes I\right)(\ket{0}\otimes I)}\\
    &\leq \nrm{\tilde{U}-(A^\dagger\otimes I) PW \left(A\otimes I\right)}\\
    &= \nrm{(A^\dagger\otimes I) P\tilde{W} \left(A\otimes I\right)-(A^\dagger\otimes I) PW \left(A\otimes I\right)}\\
    &= \nrm{\tilde{W}-W } \leq  \eps.
  \end{align*}
  \vskip -0.5cm
\end{proof}

Now we can state the main result of this appendix, which tells us how to efficiently turn a function (provided with its Taylor series) of a Hamiltonian~$H$, into a quantum circuit by using controlled Hamiltonian simulation.

In the following theorem we assume that the eigenvalues of $H$ lie in a radius-$r$ ball around $x_0$.  The main idea is that if even $r+\delta$ is less than the radius of convergence of the Taylor series, then we can obtain an $\eps$-approximation of $f$ by truncating the series at logarithmically high powers.
$B$ will be an upper bound on the absolute value of the function within the $r+\delta$ ball around $x_0$, in particular $\nrm{f(H)/B}\leq 1$. Therefore we can implement $f(H)/B$ as a block of some larger unitary. It turns out that apart from the norm and sparsity of $H$ and precision parameters, the complexity depends on the ratio of $\delta$ and $r$.
\begin{theorem}[Implementing a smooth function of a Hamiltonian]\label{thm:Taylor}
Let $x_0\in\mathbb{R}$ and $r>0$ be such that $f(x_0+x)=\sum_{\ell=0}^{\infty} a_\ell x^\ell$ for all $x\in\![-r,r]$.
Suppose $B>0$ and $\delta\in(0,r]$ are such that $\sum_{\ell=0}^{\infty}(r+\delta)^\ell|a_\ell|\leq B$.
If $\nrm{H-x_0I}\leq r$ and $\eps\in\!\left(0,\frac{1}{2}\right]$, then we can implement a unitary $\tilde{U}$ such that $\nrm{(\bra{0}\otimes I)\tilde{U}(\ket{0}\otimes I)-\frac{f(H)}{B}}\leq\eps$, using $\bigO\left(r/\delta\log\left(r/(\delta\eps)\right)\log\left(1/\eps\right)\right)$ gates and a single use of a circuit for controlled $\left(\bigO(r\log(1/\eps)/\delta),\bigO(1/r),\eps/2\right)$-simulation of $H$.

Suppose we are given $K$ such that $\nrm{H}\leq K$ and $r=\bigO(K)$. If, furthermore, $H$~is $d$-sparse and is accessed via oracles~\eqref{eq:oracleind}-\eqref{eq:oraclemat}, then the whole circuit can be implemented using
\[
  \bigO\left(\!\frac{Kd}{\delta}\log\left(\frac{K}{\delta\eps}\right)\log\left(\frac{1}{\eps}\right)\!\right)\text{ queries and }
  \bigO\left(\!\frac{Kd}{\delta}\log\left(\frac{K}{\delta\eps}\right)\log\left(\frac{1}{\eps}\right)\!\left[\log(n)\!+\!\log^{\frac{5}{2}}\left(\frac{K}{\delta\eps}\right)\right]\!\right)\text{ gates.}
\]
\end{theorem}
\begin{proof}
  The basic idea is to combine Lemma~\ref{lemma:LowWeightAPX} and Lemma~\ref{lemma:LCUApplied} and apply them to a transformed version of the function.
  First we define $\delta':=\delta/(r+\delta)$, which is at most $1/2$ by assumption. Then, for all $\ell\in\mathbb{N}$ let $b_\ell:=a_\ell(r+\delta)^\ell$ and define the function $g:[-1+\delta',1-\delta'] \rightarrow \mathbb{R}$ by $g(y):=\sum_{\ell=0}^{\infty}b_\ell y^\ell$ so that
  \begin{equation}\label{eq:sclaedF}
    f(x_0+x)=g(x/(r+\delta)) \qquad \text{ for all } x\in\![-r,r].
  \end{equation}
  Now we set $L:=\left\lceil\frac{1}{\delta'}\log\left(\frac{8}{\eps}\right)\right\rceil$. Then for all $y\in\![-1+\delta',1-\delta']$
  \begin{align*}
    \left|g(y)-\sum_{\ell=0}^{L-1}b_\ell y^\ell\right|
    &= \left|\sum_{\ell=L}^{\infty}b_\ell y^\ell\right|\\
    &\leq \sum_{\ell=L}^{\infty}\left|b_\ell (1-\delta')^\ell\right|\\
    &\leq (1-\delta')^L \sum_{\ell=L}^{\infty}\left|b_\ell\right|\\
    &\leq \left(1-\delta'\right)^{L}B\\
    &\leq e^{-\delta'L}B\\
    &\leq \frac{\eps B}{8}.
  \end{align*}
  We would now like to obtain a Fourier-approximation of $g$ for all $y\in[-1+\delta',1-\delta']$, with precision $\eps'=\frac{\eps B}{2}$.
  Let $b':=(b_0,b_1,\ldots,b_{L-1})$ and observe that $\nrm{b'}_1\leq\nrm{b}_1\leq B$. We apply Lemma~\ref{lemma:LowWeightAPX} to the function $g$, using the polynomial approximation corresponding to the truncation to the first $L$ terms, i.e., using the coefficients in $b'$.
  Then we obtain a Fourier $\eps'$-approximation $\tilde{g}(y):=\sum_{m=-M}^{M}\tilde{c}_m e^{\frac{i\pi m}{2}y}$ of $g$, with
  $$
  M=\bigO\left(\frac{1}{\delta'}\log\left(\frac{\nrm{b'}_1}{\eps'}\right)\right)=\bigO\left(\frac{r}{\delta}\log\left(\frac{1}{\eps}\right)\right)
  $$
  such that the vector of coefficients $\tilde{c}\in\mathbb{C}^{2M+1}$ satisfies $\nrm{\tilde{c}}_1\leq\nrm{b'}_1 \leq \nrm{b}_1 \leq B$.
  Let
  $$\tilde{f}(x_0+x):=\tilde{g}\left(\frac{x}{r+\delta}\right)=\sum_{m=-M}^{M}\tilde{c}_m e^{\frac{i\pi m}{2(r+\delta)}x};$$
  by~\eqref{eq:sclaedF} we see that $\tilde{f}$ is an $\eps'$-precise Fourier approximation of $f$ on the interval $[x_0-r,x_0+r]$. To transform this Fourier series to its final form, we note that
  $\tilde{f}(z)=\sum_{m=-M}^{M}\tilde{c}_m e^{\frac{i\pi m}{2(r+\delta)}(z-x_0)}$, so by defining $c_m:=\tilde{c}_me^{-\frac{i\pi m}{2(r+\delta)}x_0}$ we get a Fourier series in $z$, while preserving $\nrm{c}_1=\nrm{\tilde{c}}_1\leq B$.

  In the trivial case, when $c=0$, we choose a unitary $\tilde{U}$, such that it maps the $\ket{0}$ ancilla state to $\ket{1}$, then clearly $(\bra{0}\otimes I)\tilde{U}(\ket{0}\otimes I)=0=\tilde{f}(H)$. Clearly such a $\tilde{U}$ can be implemented using $\bigO(1)$ gates and $0$ queries.
  Otherwise we can apply Lemma~\ref{lemma:LCUApplied} to this modified Fourier series to construct a unitary circuit $\tilde{V}$ implementing an $\frac{\eps}{2}$-approximation of $\tilde{f}(H)/\nrm{c}_1$.
  We can further scale down the amplitude of the $\ket{0}$-part of the output by a factor of $\nrm{c}_1/B\leq 1$, to obtain an approximation of $\tilde{f}(H)/B$ as follows. We simply add an additional ancilla qubit initialized to $\ket{0}$ on which we act with the one-qubit unitary
  $$Rot:=\left(\begin{array}{cc}
     \frac{\nrm{c}_1}{B} & \sqrt{1-\frac{\nrm{c}_1^2}{B^2}}\\
     -\sqrt{1-\frac{\nrm{c}_1^2}{B^2}} & \frac{\nrm{c}_1}{B}
   \end{array}\right).$$
  Finally we define $\tilde{U}:=Rot\otimes \tilde{V}$, and define $\ket{0}\ket{0}$ as the new success indicator, where the first qubit is the new ancilla.
  We show that $\tilde{U}$ implements $f(H)/B$ with $\eps$ precision: (if $c=0$, let us use the definition $\tilde{f}(H)/\nrm{c}_1:=0$)
  \begin{align*}
   \nrm{(\bra{0}\bra{0}\otimes I)\tilde{U}(\ket{0}\ket{0}\otimes I)-\frac{f(H)}{B}}
   &\leq
     \nrm{(\bra{0}\bra{0}\otimes I)\tilde{U}(\ket{0}\ket{0}\otimes I)-\frac{\tilde{f}(H)}{B}}
     + \nrm{\frac{\tilde{f}(H)}{B}-\frac{f(H)}{B}}\\
   &= \frac{\nrm{c}_1}{B}\nrm{(\bra{0}\otimes I)\tilde{V}(\ket{0}\otimes I)-\frac{\tilde{f}(H)}{\nrm{c}_1}}
     + \nrm{\frac{\tilde{f}(H)-f(H)}{B}}\\
   &\leq  \frac{\nrm{c}_1}{B}\frac{\eps}{2}+\frac{\eps'}{B}\\
   &\leq \eps .
  \end{align*}
  Lemma~\ref{lemma:LCUApplied} uses $\bigO\left(M\log(M+1)\right)=\bigO\left(r/\delta\log(1/\eps)\log\left(r/(\delta\eps)\right)\right)$ gates and a single use of a controlled $\left(M,\gamma=\pi/(2r+2\delta),\eps/2\right)$-simulation of $H$. If $\nrm{H}=\bigO(K)$, we can use Lemma~\ref{lemma:controlledHamsin} to conclude
  $\bigO\left(M\gamma Kd\log\left(\frac{1}{\eps}\right)\log\left(\frac{M\gamma K}{\eps}\log\left(\frac{1}{\eps}\right)\right)\right)=\bigO\left(\frac{Kd}{\delta} \log\left(\frac{K}{\delta\eps}\right)\log\left(\frac{1}{\eps}\right)\right)$ query and
  \begin{align*}
   \bigO\left(M\gamma Kd\log\left(\frac{1}{\eps}\right)\!\log\left(\frac{M\gamma K}{\eps}\log\left(\frac{1}{\eps}\right)\!\right)\!\left[\log(n)+\log^{\frac{5}{2}}\!\left(\frac{M\gamma K}{\eps}\log\left(\frac{1}{\eps}\right)\!\right)\!\right]\right)\kern-52mm&\\     &=\bigO\left(\frac{Kd}{\delta}\log\left(\frac{K}{\delta\eps}\right)\log\left(\frac{1}{\eps}\right)\!\left[\log(n)+\log^{\frac{5}{2}}\!\left(\frac{K}{\delta\eps}\right)\!\right]\right)
  \end{align*}
  gate complexity.
  Finally note that the polynomial cost of calculating $c$ that is required by Lemma~\ref{lemma:LowWeightAPX} does not affect the query complexity or the circuit size, it only affects the description of the circuit.
\end{proof}
   \begin{remark}\label{rem:subSpaceFun}
     Note that in the above theorem we can relax the criterion $\nrm{H-x_0I}\leq r$. Suppose we have an orthogonal projector $\Pi$, which projects to eigenvectors with eigenvalues in $[x_0-r,x_0+r]$, i.e., $[H,\Pi]=0$ and $\nrm{\Pi\left(H-x_0I\right)\Pi}\leq r$. Then the circuit $\tilde{U}$ constructed in Theorem~\ref{thm:Taylor} satisfies
     $$\nrm{\Pi\left((\bra{0}\otimes I)\tilde{U}(\ket{0}\otimes I)-\frac{f(H)}{B}\right)\Pi}\leq\eps.$$
   \end{remark}

   The following corollary shows how to implement functions piecewise in small ``patches" using Remark~\ref{rem:subSpaceFun}. The main idea is to first estimate the eigenvalues of $H$ up to $\theta$ precision, and then implement the function using the Taylor series centered around a point close to the eigenvalue.

   This approach has multiple advantages. First, the function may not have a Taylor series that is convergent over the whole domain of possible eigenvalues of $H$. Even if there is such a series, it can have very poor convergence properties, making $B$ large and therefore requiring a lot of amplitude amplification. Nevertheless, for small enough neighborhoods the Taylor series always converges quickly, overcoming this difficulty.

   \begin{corollary}\label{cor:patched}
     Suppose $(x_\ell)\in\mathbb{R}^L$ and $r,\theta\in\mathbb{R_+}$ are  such that the spectrum of $H$ lies in the domain $\bigcup_{\ell=1}^L [x_\ell-(r-2\theta),x_\ell+(r+2\theta)]$.\footnote{This way, even if we make $\theta$ error during the estimation of an eigenvalue $\lambda_i$, the closest $x_\ell$ will still contain $\lambda_i$ in its radius-$r$ neighborhood.} Suppose there exist coefficients $a_k^{(\ell)}\in\mathbb{R}$ such that for all $\ell\in [L]$ and $x\in\![-r,r]$ we have $f(x_\ell+x)=\sum_{k=0}^{\infty} a_k^{(\ell)} x^k$, and  $\sum_{k=0}^{\infty}(r+\delta)^k|a_k^{(\ell)}|\leq B$ for some fixed $\delta\in[0,r]$ and $B>0$. If $\nrm{H}\leq K$ and $\eps\in\!\left(0,\frac{1}{2}\right]$, then we can implement a unitary $\tilde{U}$ such that $$
     \nrm{(\bra{0}\otimes I)\tilde{U}(\ket{0}\otimes I)-\frac{f(H)}{B}}\leq\eps,
     $$
     using $\bigO\left(Lr/\delta\log\left(r/(\delta\eps)\right)\log\left(1/\eps\right)+\log(K/\theta)\log\log(K/(\theta\eps))\right)$ gates, and with $\bigO(\log(1/\eps))$ uses of an $(\bigO(1/\theta),\pi/K,\Omega(\eps^2/\log(1/\eps)))$-simulation of $H$ and a single use of a circuit for controlled $\left(\bigO(r\log(1/\eps)/\delta),\bigO(1/r),\eps/2\right)$-simulation of~$H$.
     If $r=\bigO(K)$, $\theta\leq r/4$, $\theta=\Omega(\delta)$, $\nrm{H}\leq K$, $H$~is $d$-sparse and is accessed via oracles~\eqref{eq:oracleind}-\eqref{eq:oraclemat}, then the circuit can be implemented using
     \vskip-3.5mm
     $$
     \bigO\left(\!\frac{Kd}{\delta}\log\left(\frac{K}{\delta\eps}\right)\log\left(\frac{1}{\eps}\right)\!\right)\text{ queries and }
     \bigO\left(\!\frac{Kd}{\delta}\log\left(\frac{K}{\delta\eps}\right)\log\left(\frac{1}{\eps}\right)\!\left[\log(n)\!+\!\log^{\frac{5}{2}}\left(\frac{K}{\delta\eps}\right)\right]\!\right)\text{ gates.}
     $$
   \end{corollary}

   \begin{proof}[Sketch of the proof.]
     We start by performing phase estimation on $e^{iH}$ with $\approx\theta$ resolution in phase. We boost the success probability by taking the median outcome of $\bigO(\log(1/\eps))$ parallel repetitions, so that we get a worse-than-$\theta$ estimation with probability at most $\bigO(\eps^2)$. This way the boosted phase estimation circuit is $\bigO(\eps)$-close in operator norm to an ``idealized" phase estimation unitary that never makes approximation error greater than~$\theta$ (for more details on this type of argument, see the proof of Lemma~\ref{lemma:normEst}). Phase estimation uses controlled $(\bigO(K/\theta),\pi/K,\Omega(\eps^2/\log(1/\eps)))$-simulation of $H$ and a Fourier transform on $\bigO(\log(K/\theta))$-bit numbers which can be implemented using $\bigO(\log(K/\theta)\log\log(K/(\theta\eps)))$ gates. The probability-boosting uses $\bigO(\log(1/\eps))$ repetitions. Controlled on the phase estimate $\tilde{\lambda}$ that we obtained, we implement a $1/B$-scaled version of the corresponding function ``patch'' $\left.f(x)\right|_{[x_\ell-r,x_\ell+r]}$ centered around $\arg\min|x_\ell-\tilde{\lambda}|$ using Theorem~\ref{thm:Taylor} and Remark~\ref{rem:subSpaceFun}. The additional gate complexities of the ``patches'' add up to $\bigO\left(Lr/\delta\log\left(r/(\delta\eps)\right)\log\left(1/\eps\right)\right)$, but since each ``patch'' uses the same controlled $\left(\bigO(r\log(1/\eps)/\delta),\bigO(1/r),\eps/2\right)$-simulation of $H$, we only need to implement that once. Finally we uncompute phase estimation.\footnote{Note that phase estimation on some eigenvector of $e^{iH}$ can produce a superposition of different estimates of the phase. If some intervals $[x_\ell-r,x_\ell+r]$ overlap for $\ell$ and $\ell'$, those estimates could lead to different implementations of $f(H)$ (one based on the coefficients $a_k^{(\ell)}$ and one based on $a_k^{(\ell')}$). However, this causes no difficulty; since we used the same normalization $1/B$ for all implementations, both implementations lead to essentially the same state after postselecting on the $\ket{0}$ ancilla state.}
     For the final complexity, note that we can assume without loss of generality that $L(r-2\theta)=\bigO(K)$, since otherwise we can just remove some redundant intervals from the domain. Hence $Lr=\bigO(K)$ and Lemma~\ref{lemma:controlledHamsin} implies the stated complexities.
   \end{proof}

           This corollary is essentially as general and efficient as we can hope for. Let $D$ denote the domain of possible eigenvalues of $H$. If we want to implement a reasonably smooth function~$f$, then it probably satisfies the following: there is some $r=\Omega(1)$, such that for each $x\in D$, the Taylor series in the radius-$r$ neighborhood of $x$ converges quickly, more precisely the Taylor coefficients $a_k^{(x)}$ for the $x$-centered series satisfy $\sum_{k=0}^{\infty}|a_k^{(x)}|r^k=\bigO(\nrm{f}_{\infty})$, where we define $\nrm{f}_{\infty}:=\sup_{x\in D}|f(x)|$. If this is the case, covering $D$ with radius-$\bigO(r)$ intervals, choosing $\theta=\Theta(r)$ and $\delta=\Theta(r)$, Corollary~\ref{cor:patched} provides an $\bOt{\nrm{H}d}$ query and gate complexity implementation of $f(H)/B$, where $B=\bigO(\nrm{f}_{\infty})$. The value of $B$ is optimal up to constant factors, since $f(H)/B$ must have norm at most $1$. Also the $\nrm{H}d$ factor in the complexity is very reasonable, and we achieve the logarithmic error dependence which is the primary motivation of the related techniques. An application along the lines of this discussion can be found in Lemma~\ref{lemma:squareGibbs}.

           Also note that in the above corollary we added up the gate complexities of the different ``patches.'' Since these gates prepare the Fourier coefficients of the function corresponding to the different Taylor series at different points, one could use this structure to implement all coefficients with a single circuit. This can potentially result in much smaller circuit sizes, which could be beneficial when the input model allows more efficient Hamiltonian simulation (which then would no longer be the bottleneck in the complexity).

           \subsection{Applications of smooth functions of Hamiltonians}
           In this subsection we use the input model for the $d$-sparse matrix $H$ as described at the start of Section~\ref{sec:upperbounds}.
           We calculate the implementation cost in terms of queries to the input oracles~\eqref{eq:oracleind}-\eqref{eq:oraclemat}, but it is easy to convert the results to more general statements as in the previous subsection.

           The following theorem shows how to efficiently implement the function $e^{-H}$ for some $H\succeq I$. We use this result in the proof of Lemma~\ref{lemma:trPreEst} to estimate expectation values of the quantum state $\rho=e^{-H}/\tr{e^{-H}}$ (for the application we ensure that $H\succeq I$ by adding some multiple of $I$).

           \begin{theorem}\label{thm:emH}
             Suppose that $I\preceq H$ and we are given $K\in\mathbb{R}_+$ such that $\nrm{H}\leq 2K$. If $\eps\in(0,1/3)$,
             then we can implement a unitary $\tilde{U}$ such that $\nrm{(\bra{0}\otimes I)\tilde{U}(\ket{0}\otimes I)-e^{-H}}\leq\eps$ using   $$\bigO\left(Kd\log\left(\frac{K}{\eps}\right)\log\left(\frac{1}{\eps}\right)\right)\text{ queries and }\bigO\left(Kd\log\left(\frac{K}{\eps}\right)\log\left(\frac{1}{\eps}\right)\left[\log(n)+\log^{\frac{5}{2}}\left(\frac{K}{\eps}\right)\right]\right)\text{ gates.}$$
           \end{theorem}
           \begin{proof}
             In order to use Theorem~\ref{thm:Taylor} we set $x_0:=K+1/2$ so that $\nrm{H-x_0 I}\leq K=:r$, and use the function
             $$
             f(x_0+x)=e^{-x_0-x}=e^{-x_0}e^{-x}=e^{-x_0}\sum_{\ell=0}^{\infty}\frac{(-x)^\ell}{\ell!}.
             $$
             We choose $\delta:=1/2$ so that $e^{-x_0}\sum_{\ell=0}^{\infty}\frac{(r+\delta)^\ell}{\ell!}=e^{-x_0}\sum_{\ell=0}^{\infty}\frac{x_0^\ell}{\ell!}=1$, therefore we set $B:=1$.
             Theorem~\ref{thm:Taylor} tells us that we can implement a unitary~$\tilde{U}$, such that $\tilde{f}(H):=(\bra{0}\otimes I)\tilde{U}(\ket{0}\otimes I)$ is an $\eps$-approximation of $f(H)/B= e^{-H}$, using
             \[
               \bigO\left(Kd\log\left(\frac{K}{\eps}\right)\log\left(\frac{1}{\eps}\right)\right)\text{ queries and }\bigO\left(Kd\log\left(\frac{K}{\eps}\right)\log\left(\frac{1}{\eps}\right)\left[\log(n)+\log^{\frac{5}{2}}\left(\frac{K}{\eps}\right)\right]\right) \text{ gates}.
             \]
           \end{proof}

To conclude this appendix, we now sketch the proofs of a few interesting consequences of the earlier results in this appendix. These will, however, not be used in the body of the paper.

First, we show how to use the above subroutine together with amplitude amplification to prepare a Gibbs state with cost depending logarithmically on the precision parameter, as shown by the following lemma. To our knowledge this is the first Gibbs sampler that achieves logarithmic dependence on the precision parameter without assuming access to the entries of $\sqrt{H}$ as in~\cite{chowdhury2016QGibbsSampling}. This can mean a significant reduction in complexity; for more details see the introduction of Section~\ref{sec:estTrArhogeneral}.

           \begin{lemma}\label{lemma:preGibbs}
             We can probabilistically prepare a purified Gibbs state $\ket{\tilde{\gamma}}_{AB}$ such that with high probability $\nrm{\Tr_B\left(\ketbra{\tilde{\gamma}}{\tilde{\gamma}}_{AB}\right)\kern-0.2mm-\kern-0.2mm e^{-H}\!/\tr{e^{-H}}\kern-0.3mm}_1\leq \eps$ holds, using an expected cost $\bOt{\!\sqrt{n / \tr{e^{-H}}}}$ times the complexity of Theorem~\ref{thm:emH}. If we are given a number $z \leq \tr{e^{-H}}$, then we can also prepare $\ket{\tilde{\gamma}}_{AB}$ in a unitary fashion with cost $\bOt{\sqrt{n / z}}$ times the complexity of Theorem~\ref{thm:emH}.
           \end{lemma}

           \begin{proof}[Sketch of proof]
             First we show how to prepare a purified sub-normalized Gibbs state. Then we use the exponential search algorithm of Boyer et al.~\cite{boyer1998TightBoundsOnQuantumSearching} (with exponentially decreasing guesses for the norm $a$ of the subnormalized Gibbs state, and hence exponentially increasing number of amplitude amplification steps) to postselect on this sub-normalized state in a similar fashion as in Algorithm~\ref{alg:genMin}.
             There is a possible caveat here: if we postselect on a state with norm $a$, then it gets rescaled by $1/a$ and its preparation error is rescaled by $1/a$ as well. Therefore during the rounds of the search algorithm we always increase the precision of implementation to compensate for the increased (error) amplification.
             Since the success of postselection in a round is upper bounded by the square of $\bigO\left(a\cdot\#\{\text{amplification steps in the round}\}\right)$, the probability for the postselection to succeed in any of the early rounds is small.

             Now we describe how to prepare a purified sub-normalized Gibbs state. We use the decomposition $H\!=\!\sum_{j=1}^{n}E_j\ketbra{\phi_j}{\phi_j}$, where $\{\ket{\phi_j}:j\in[n]\}$ is an orthonormal eigenbasis of $H$. Due to the invariance of maximally entangled states under transformations of the form $W\otimes W^*$ for unitary $W$, we have
             \begin{equation}\label{eq:preGibbsMaxEnt}
               \frac{1}{\sqrt{n}}\sum_{j=0}^{n-1}\ket{j}_A\ket{j}_B=\frac{1}{\sqrt{n}}\sum_{j=1}^{n}\ket{\phi_j}_A\ket{\phi^*_j}_B.
             \end{equation}
             Suppose we can implement a unitary $U$ such that $(\bra{0}\otimes I)U(\ket{0}\otimes I)=e^{-H/2}$.
             If we apply $U$ to the A-register of the state~\eqref{eq:preGibbsMaxEnt}, then we get a state $\ket{\gamma}$ such that $\Tr_B\left((\bra{0}\otimes I)\ketbra{\gamma}{\gamma}(\ket{0}\otimes I)\right)=e^{-H}/n$.

             If we implement $e^{-H/2}$ 
             with sufficient precision using Theorem~\ref{thm:emH} in the 
             exponential search algorithm, then after $\bigO\left(\sqrt{n/\tr{e^{-H}}}\right)$ rounds of amplitude amplification, with high probability we obtain a Gibbs state using the claimed expected runtime.

             If we also know a lower bound $z \leq \tr{e^{-H}}$, then we have an upper bound on the expected runtime, therefore we can turn the procedure into a unitary circuit using standard techniques.
           \end{proof}

  We can also recover the Gibbs sampler of Chowdhury and Somma~\cite{chowdhury2016QGibbsSampling}: if we apply our Corollary~\ref{cor:patched} to the function $e^{-x^2}$ assuming access to $\sqrt{H}$ for some psd matrix~$H$, then we get a Gibbs sampler for the state $e^{-H}/\Tr(e^{-H})$, similar to~\cite{chowdhury2016QGibbsSampling}. The advantage of the presented approach is that it avoids the usage of involved integral transformations, and can be presented without writing down a single integral sign, also due to our general results the proof is significantly shorter. Before we prove the precise statement in Lemma~\ref{lemma:squareGibbs}, we need some preparation for the application of Corollary~\ref{cor:patched}:
  \begin{lemma}\label{lemma:emx2}
    For all $k\in\mathbb{N}$ we have
    \begin{equation}\label{eq:emx2}
      \partial_x^{k+1} e^{-x^2}=-2x\partial_x^{k} e^{-x^2}-2k\partial_x^{k-1} e^{-x^2}
    \end{equation}
    and for all $x\in\mathbb{R}$
    \begin{equation}\label{eq:emx2n}
      \left|\partial_x^{k} e^{-x^2}\right|\leq (2|x|+2k)^k e^{-x^2}.
    \end{equation}
    Therefore
    \begin{equation}\label{eq:emx2Taylor}
      \sum_{k=0}^{\infty}\frac{\left|\partial_x^{k} e^{-x^2}\right|}{k!}\left(\frac{1}{8e}\right)^k\leq 2.
    \end{equation}
           \end{lemma}
           \begin{proof}
             We prove both claims by induction. $\partial_x^{0} e^{-x^2}=e^{-x^2}$, $\partial_x^{1} e^{-x^2}=-2xe^{-x^2}$ and $\partial_x^{2} e^{-x^2}=4x^2e^{-x^2}-2e^{-x^2}$, so~\eqref{eq:emx2} holds for $k=1$. Suppose~\eqref{eq:emx2} holds for $k$, we prove the inductive step as follows:
             \begin{align*}
               \partial_x^{k+2} e^{-x^2} \!
               =\partial_x\left(\partial_x^{k+1} e^{-x^2}\right)
               \overset{\eqref{eq:emx2}}{=}\partial_x\left(-2x\partial_x^{k} e^{-x^2}\!-2k\partial_x^{k-1} e^{-x^2}\right)
               =-2x\partial_x^{k+1} e^{-x^2}\!-2(k+1)\partial_x^{k} e^{-x^2}.
             \end{align*}
             Similarly, observe that~\eqref{eq:emx2n} holds for $k=0$ and $k=1$. Suppose~\eqref{eq:emx2n} holds for $k$, then we show the induction step as follows:
             \begin{align*}
               \left|\partial_x^{k+1} e^{-x^2}\right|
               \kern+1.2mm&\kern-1.2mm\overset{\eqref{eq:emx2}}{=}\left|-2x\partial_x^{k} e^{-x^2}-2k\partial_x^{k-1} e^{-x^2}\right|\\
                          &\leq\left|2x\partial_x^{k} e^{-x^2}\right|+\left|2k\partial_x^{k-1} e^{-x^2}\right|\\
               \kern+1.2mm&\kern-1.2mm\overset{\eqref{eq:emx2n}}{\leq}2|x|(2|x|+2k)^k e^{-x^2}+2k(2|x|+2(k-1))^{k-1} e^{-x^2}\\
                          &\leq (2|x|+2(k+1))^{k+1} e^{-x^2}.
             \end{align*}
             Finally, using the previous two statements we can prove~\eqref{eq:emx2Taylor} by the following calculation:
             \begin{align*}
               \sum_{k=0}^{\infty}\frac{\left|\partial_x^{k} e^{-x^2}\right|}{k!}\left(\frac{1}{8e}\right)^k
               \kern+1.2mm&\kern-1.2mm\overset{\eqref{eq:emx2n}}{\leq} \sum_{k=0}^{\infty}\frac{(2|x|+2k)^k e^{-x^2}}{k!}\left(\frac{1}{8e}\right)^k\\
                          &\leq \sum_{k=0}^{\infty}\frac{(4|x|)^k e^{-x^2}}{k!}\left(\frac{1}{8e}\right)^k
                            +\sum_{k=1}^{\infty}\frac{(4k)^k e^{-x^2}}{k!}\left(\frac{1}{8e}\right)^k\\
                          &\leq e^{-x^2}\left(\sum_{k=0}^{\infty}\frac{1}{k!}\left(\frac{4|x|}{8e}\right)^k
                            +\sum_{k=1}^{\infty}\frac{1}{k!}\left(\frac{4k}{8e}\right)^k\right)\\
                          &\leq e^{-x^2}\left(e^{\frac{|x|}{2e}}
                            +\sum_{k=1}^{\infty}\frac{1}{\sqrt{2\pi}}\left(\frac{e}{k}\right)^k\left(\frac{k}{2e}\right)^k\right)\\
                          &=	e^{-\left(|x|-\frac{1}{4e}\right)^2}e^{\left(\frac{1}{4e}\right)^2} + \frac{e^{-x^2}}{\sqrt{2\pi}}\\
                          &\leq e^{\left(\frac{1}{4e}\right)^2} + \frac{1}{\sqrt{2\pi}}\\
                          &\leq 2.
             \end{align*}
           \end{proof}

           \begin{lemma}\label{lemma:squareGibbs}
             Suppose we know a $K>1$ such that $\nrm{H}\leq K$. If $\eps\in(0,1/3)$,
             then we can implement a unitary $\tilde{U}$ such that $\nrm{(\bra{0}\otimes I)\tilde{U}(\ket{0}\otimes I)-e^{-H^2}/2}\leq\eps$ using   $$\bigO\left(Kd\log\left(\frac{K}{\eps}\right)\log\left(\frac{1}{\eps}\right)\right)\text{ queries and }\bigO\left(Kd\log\left(\frac{K}{\eps}\right)\log\left(\frac{1}{\eps}\right)\left[\log(n)+\log^{\frac{5}{2}}\left(\frac{K}{\eps}\right)\right]\right)\text{ gates.}$$
           \end{lemma}

           \begin{proof}
             We apply Corollary~\ref{cor:patched} to the function $e^{-x^2}$. For this let $L_0:=\lceil 32K \rceil$, $L:=2L_0+1$ and let  $x_\ell:=(\ell-1-L_0)/32$ for all $\ell\in[L]$. We choose $r:=1/32$, $\delta:=\theta:=1/128$ and $B=2$ so that the conditions of Corollary~\ref{cor:patched} are satisfied, as shown by Lemma~\ref{lemma:emx2}. Indeed $r+\delta\leq1/(8e)$, hence for all $\ell\in[L]$ we have $\sum_{k=0}^{\infty}a_k^{(\ell)}(r+\delta)^k\leq 2=B$ as we can see by~\eqref{eq:emx2Taylor}. Since $\delta=\Theta(1)$, Corollary~\ref{cor:patched} provides the desired complexity.
           \end{proof}

           We can use the above lemma to prepare Gibbs states in a similar way to Lemma~\ref{lemma:preGibbs}.
           In case we have access to $\sqrt{H}$, the advantage of this method is that the dependence on $\nrm{H}$ is reduced~to~$\sqrt{\nrm{H}}$.

           \paragraph{Improved HHL algorithm.} Our techniques can also be applied to the Harrow-Hassidim-Lloyd (HHL) algorithm~\cite{harrow2009QLinSysSolver} in order to gain improvements in a similar manner to Childs et al.~\cite{childs2015QLinSysExpPrec}. The problem the HHL algorithm solves is the following. Suppose we have a circuit $U$ preparing a quantum state~$\ket{b}$ (say, starting from $\ket{0}$), and have $d$-sparse oracle access to a non-singular Hamiltonian $H$. The task is to prepare a quantum state, that $\eps$-approximates $H^{-1}\ket{b}/\nrm{H^{-1}\ket{b}}$. For simplicity here we only count the number of uses of $U$ and the number of queries to $H$. Childs et al.~\cite{childs2015QLinSysExpPrec} present two different methods for achieving this, one based on Hamiltonian simulation, and another directly based on quantum walks. Under the conditions $\nrm{H}\leq 1$ and $\nrm{H^{-1}}\leq\kappa$, the former makes $\bigO\left(\kappa\sqrt{\log(\kappa/\eps)}\right)$ uses of $U$ and has query complexity $\bigO(d\kappa^2\log^{2.5}(\kappa/\eps))$. The latter makes $\bigO\left(\kappa\log(d\kappa/\eps)\right)$ uses of $U$ and has query complexity $\bigO(d\kappa^2\log^{2}(d\kappa/\eps))$.

           Now we provide a sketch of how to solve the HHL problem with $\bigO(\kappa)$ uses  of $U$ and with query complexity $\bigO(d\kappa^2\log^{2}(\kappa/\eps))$ using our techniques. The improvement on both previously mentioned results is not very large, but our proof is significantly shorter thanks to our general Theorem~\ref{thm:Taylor} and Corollary~\ref{cor:patched}.

           To solve the HHL problem we need to implement the function  $H^{-1}$, i.e., apply the function $f(x)=1/x$ to~$H$. Due to the constraints on $H$, the eigenvalues of $H$ lie in the union of the intervals $[-1,-1/\kappa]$ and $[1/\kappa,1]$. We first assume that the eigenvalues actually lie in $[1/\kappa,1]$. In this case we can easily implement the function $1/x$ by Theorem~\ref{thm:Taylor} using the Taylor series around $1$:
           \begin{equation}\label{eq:1oz}
             (1+z)^{-1}=\frac{1}{1+z}=\sum_{k=0}^{\infty}(-1)^kz^k.
           \end{equation}
           As $H^{-1}=(I+(H-I))^{-1}$, we are interested in the eigenvalues of $H-I$. The eigenvalues of $H-I$ lie in the interval $[-1+1/\kappa,0]$, so we choose $r:=1-1/\kappa$ and $\delta:=1/(2\kappa)$. By substituting $z:=-1+1/(2\kappa)$ in~\eqref{eq:1oz}, we can see that $B:=2\kappa$ satisfies the conditions of Theorem~\ref{thm:Taylor}. Let $\eps'\in(0,1/2)$, then Theorem~\ref{thm:Taylor} provides an $\bigO\left(d\kappa \log\left(\kappa/\eps'\right)\log\left(1/\eps'\right)\right)$-query implementation of an $\eps'$-approximation of the operator $H^{-1}/(2\kappa)$, since $\nrm{H}\leq 1$. We can proceed similarly when the eigenvalues of $H-I$ lie in the interval $[-1,-1/\kappa]$, and we can combine the two cases using Corollary~\ref{cor:patched}.

           Setting $\eps':=c\eps/\kappa$ for an appropriate constant $c$ and using amplitude amplification, we can prepare an $\eps$-approximation of the state $H^{-1}\ket{b}/\nrm{H^{-1}\ket{b}}$ as required by HHL using $\bigO(\kappa)$ amplitude amplification steps. Therefore we use $U$ at most $\bigO(\kappa)$ times and make $\bigO\left(d\kappa^2 \log^2\left(\kappa/\eps\right)\right)$ queries.

           \section{Generalized minimum-finding algorithm}\label{app:genMinFind}
           In this appendix we describe our generalized quantum minimum-finding algorithm, which we are going to apply to finding an approximation of the ground state energy of a Hamiltonian. This algorithm generalizes the results of D{\"u}rr and H{\o}yer~\cite{durr1996QMinimumFinding} in a manner similar to the way amplitude amplification~\cite{brassard2002AmpAndEst} generalizes Grover search: we do not need to assume the ability to query individual elements of the search space, we just need to be able to generate a superposition over the search space. The algorithm also has the benefit over binary search that it removes a logarithmic factor from the complexity.

           The backbone of our analysis will be the meta-algorithm below from~\cite{durr1996QMinimumFinding}. The meta-algorithm finds the minimal element in the range of the random variable $X$ by sampling, where by ``range'' we mean the values which occur with non-zero probability. We assume $X$ has finite range.

           \begin{metaalgorithm}[H]
             \begin{description}
             \item[Input] A discrete random variable $X$ with finite range.
             \item[Output] The minimal value $x_{\min}$ in the range of $X$.
             \end{description}
             \begin{algorithmic}
               \State \textbf{Init} $t\leftarrow 0$; $s_0\leftarrow \infty$
               \State \textbf{Repeat} until $s_t$ is minimal in the range of $X$
               \begin{enumerate}
               \item $t\leftarrow t+1$
               \item Sample a value $s_t$ according to the conditional distribution $\mathrm{Pr}(X=s_t \mid X<s_{t-1})$.
               \end{enumerate}
             \end{algorithmic}
             \caption{Minimum-finding}
             \label{alg:metaMin}
           \end{metaalgorithm}

           Note that the above algorithm will always find the minimum, since the obtained samples are strictly decreasing.
           \begin{lemma}\label{lemma:metaMin}
             Let $X$ be a finite discrete random variable whose range of values is $x_1<x_2<\ldots<x_N$.
             Let $S(X)=\{s_1,s_2,\dots\}$ denote the random set of values obtained via sampling during a run of Meta-Algorithm~\ref{alg:metaMin} with input random variable $X$. If $k\in[N]$, then
             $$\mathrm{\mathrm{Pr}}(x_k\in S(X))=\frac{\mathrm{\mathrm{Pr}}(X=x_k)}{\mathrm{\mathrm{Pr}}(X\leq x_k)}.$$
           \end{lemma}
           \begin{proof}
             The intuition of the proof is to show that, whenever $\mathrm{Pr}(s_{t-1}>x_k)>0$, we have
             \begin{equation}\label{eq:conditional}
               \mathrm{Pr}(s_t=x_k \mid t\in[N]\text{ is the first time such that } s_t\leq x_k)=\frac{\mathrm{Pr}(X=x_k)}{\mathrm{Pr}(X\leq x_k)}.
             \end{equation}
             To formulate the statement more precisely\footnote{Throughout this proof, whenever a fraction is $0/0$, we simply interpret it as $0$. Therefore we also interpret conditional probabilities, conditioned on events that happen with probability $0$, as $0$.} we consider a fixed value $t\in[N]$. For notational convenience let $x:=x_k$, $x_{N+1}:=\infty$, we prove~\eqref{eq:conditional} by:
             \begin{align*}
               \mathrm{Pr}\left(s_t=x \mid s_t\leq x \kern-0.2mm \wedge \kern-0.2mm s_{t-1}>x\right)
               &=\frac{\mathrm{Pr}\left(s_t=x\right)}{\mathrm{Pr}\left(s_t\leq x \wedge s_{t-1}>x\right)}\\
               &=\sum_{x_\ell>x}\frac{\mathrm{Pr}\left(s_t=x\wedge s_{t-1}=x_\ell\right)}{\mathrm{Pr}\left(s_t\leq x \wedge s_{t-1}>x\right)}\\
               &=\sum_{x_\ell>x}\frac{\mathrm{Pr}\left(s_t=x\wedge s_{t-1}=x_\ell\right)}{\mathrm{Pr}\left(s_t\leq x \wedge s_{t-1}=x_\ell\right)}\frac{\mathrm{Pr}\left(s_t\leq x\wedge s_{t-1}=x_\ell\right)}{\mathrm{Pr}\left(s_t\leq x \wedge s_{t-1}>x\right)}\\
               &=\sum_{x_\ell>x}\frac{\mathrm{Pr}\left(s_t=x\mid s_{t-1}=x_\ell\right)\mathrm{Pr}\left(s_{t-1}=x_\ell\right)}{\mathrm{Pr}\left(s_t\leq x \mid s_{t-1}=x_\ell\right)\mathrm{Pr}\left(s_{t-1}=x_\ell\right)}\frac{\mathrm{Pr}\left(s_t\leq x\wedge s_{t-1}\!=\kern-0.2mm x_\ell\right)}{\mathrm{Pr}\left(s_t\leq x \wedge s_{t-1}>x\right)}\\
               &=\sum_{x_\ell>x}\frac{\mathrm{Pr}\left(X=x\right)}{\mathrm{Pr}\left(X\leq x\right)}\frac{\mathrm{Pr}\left(s_{t-1}=x_\ell\right)}{\mathrm{Pr}\left(s_{t-1}=x_\ell\right)}\frac{\mathrm{Pr}\left(s_t\leq x\wedge s_{t-1}=x_\ell\right)}{\mathrm{Pr}\left(s_t\leq x \wedge s_{t-1}>x\right)}\\
               &=\frac{\mathrm{Pr}\left(X=x\right)}{\mathrm{Pr}\left(X\leq x\right)}\frac{\mathrm{Pr}\left(s_{t-1}>x\right)}{\mathrm{Pr}\left(s_{t-1}>x\right)}.
             \end{align*}
             This is enough to conclude the proof, since there is always a smallest $t\in[N]$ such that $s_t\leq x$, as the algorithm always finds the minimum in at most $N$ steps. So we can finish the proof by
             \begin{align*}
               \mathrm{Pr}\left(x\in S(X)\right)&=\sum_{t=1}^{N}\mathrm{Pr}\left(s_t=x \right)\\
                                                &=\sum_{t=1}^{N}\mathrm{Pr}\left(s_t=x \mid s_t\leq x \wedge s_{t-1}>x\right)\mathrm{Pr}\left(s_t\leq x \wedge s_{t-1}>x\right)\\
                                                &=\frac{\mathrm{Pr}\left(X=x\right)}{\mathrm{Pr}\left(X\leq x\right)}\sum_{t=1}^{N}\mathrm{Pr}\left(s_t\leq x \wedge s_{t-1}>x\right)\\
                                                &=\frac{\mathrm{Pr}\left(X=x\right)}{\mathrm{Pr}\left(X\leq x\right)}\mathrm{Pr}\left(\exists t\in [N]\!:\,s_t\leq x \wedge s_{t-1}>x\right)\\
                                                &=\frac{\mathrm{Pr}\left(X=x\right)}{\mathrm{Pr}\left(X\leq x\right)}.
             \end{align*}
             \vskip-5mm
           \end{proof}

           Now we describe our generalized minimum-finding algorithm which is based on Meta-Algorithm~\ref{alg:metaMin}. We take some unitary $U$, and replace $X$ by the distribution obtained if we measured the second register of $U\ket{0}$. We implement conditional sampling via amplitude amplification and the use of the exponential search algorithm of Boyer et al.~\cite{boyer1998TightBoundsOnQuantumSearching}. If a unitary prepares the state $\ket{0}\ket{\phi}+\ket{1}\ket{\psi}$ where $\nrm{\phi}^2+\nrm{\psi}^2=1$, then this exponential search algorithm built on top of amplitude amplification prepares the state $\ket{1}\ket{\psi}$ probabilistically using an expected number of $\mathcal{O}\left(1/\nrm{\psi}\right)$ applications of $U$ and~$U^{-1}$ (we will skip the details here, which are straightforward modifications of~\cite{boyer1998TightBoundsOnQuantumSearching}).

           \begin{algorithm}[H]
             \begin{description}
             \item[Input] A number $M$ and a unitary $U$, acting on $q$ qubits, such that $U\ket{0}=\sum_{k=1}^{N}\ket{\psi_k}\ket{x_k}$, where $x_k$ is a binary string representing some number and $\ket{\psi_k}$ is an unnormalized quantum state on the first register. Let $x_1<x_2<\ldots <x_N$ and define $X$ to be the random variable with $\mathrm{Pr}(X=x_k)=\nrm{\psi_k}^2$.
             \item[Output] Some $\ket{\psi_k}\ket{x_{k}}$ for a (hopefully) small $k$.
             \end{description}
             \begin{algorithmic}
               \State \textbf{Init} $t\leftarrow 0$; $s_0\leftarrow \infty$
               \State \textbf{While} the total number of applications of $U$ and $U^{-1}$ does not exceed $M$:
               \begin{enumerate}
               \item $t\leftarrow t+1$
               \item Use the exponential search algorithm with amplitude amplification on states such that $x_k< s_{t-1}$ to obtain a sample $\ket{\psi_k}\ket{x_k}$.
               \item $s_t\leftarrow x_k$
               \end{enumerate}
             \end{algorithmic}
             \caption{Generalized minimum-finding}
             \label{alg:genMin}
           \end{algorithm}

           \begin{lemma}\label{lemma:genMinExp}
             There exists $C\in\mathbb{R_+}$, such that if we run Algorithm~\ref{alg:genMin} indefinitely (setting $M=\infty$), then for every $U$ and $x_k$ the expected number of uses of $U$ and $\kern0.2mm U^{-1}$ before obtaining a sample $x\leq x_k$ is at most $\frac{C}{\sqrt{\mathrm{Pr}(X\leq x_k)}}$.
           \end{lemma}
           \begin{proof}
             Let $X_{<x_\ell}$ denote the random variable for which $\mathrm{Pr}(X_{<x_\ell}=x)=\mathrm{Pr}(X=x\mid X<x_\ell)$. The~expected number of uses of $U^{\pm 1}$ in Algorithm~\ref{alg:genMin} before obtaining a value $x\leq x_k$ is
             \begin{align*}
               \mathbb{E}\left[\#\text{uses of }U^{\pm 1}\text{ for finding }x\!\leq\! x_k\right]\!
               &=\sum_{i=0}^{N-1}\mathbb{E}\left[\#\text{uses of }U^{\pm 1}\text{ in the }i\text{-th round before }x\!\leq\! x_k\right]\\
               &=\!1\!+\!\sum_{i=1}^{N\kern-0.3mm-\kern-0.3mm1}\kern-4mm\sum_{\kern4mm\ell=k+1}^N \!\!\!\!\!\!\mathrm{Pr}(s_i\!=\!x_\ell)\mathbb{E}\left[\#\text{uses of }U^{\pm 1}\text{ for sampling }X_{< x_\ell}\right]\\               
               &=\!1\!+\!\!\!\!\sum_{\ell=k+1}^{N}\!\!\mathrm{Pr}(x_\ell\!\in\! S(X))\,\mathbb{E}\left[\#\text{uses of }U^{\pm 1}\text{ for sampling }X_{< x_\ell}\right]\\
                &=\!1\!+\!\!\!\sum_{\ell=k+1}^{N}\frac{\mathrm{Pr}(X=x_\ell)}{\mathrm{Pr}(X\leq x_\ell)}\bigO\left(\frac{1}{\sqrt{\mathrm{Pr}(X< x_\ell)}}\right)\\
                 &=\bigO\left(\sum_{\ell=k+1}^{N}\frac{\mathrm{Pr}(X=x_\ell)}{\mathrm{Pr}(X\leq x_\ell)}\frac{1}{\sqrt{\mathrm{Pr}(X< x_\ell)}}\right)\\
                 &=\bigO\left(\frac{1}{\sqrt{\mathrm{Pr}(X\leq x_k)}}\right),
             \end{align*}
             where the last equality follows from Equation~\eqref{eq:intApx} below. The constant $C$ from the lemma is the constant hidden by the $\bigO$. The remainder of this proof consists of proving~\eqref{eq:intApx} using elementary calculus. Let us introduce the notation $p_0:=\mathrm{Pr}(X\leq x_k)$ and for all $j\in[N-k]$ let $p_j:=\mathrm{Pr}(X=x_{k+j})$.
             Then the expression inside the $\bigO$ on the second-to-last line above becomes
             \begin{equation}\label{eq:pSimplified}
               \sum_{\ell=k+1}^{N}\frac{\mathrm{Pr}(X=x_\ell)}{\mathrm{Pr}(X\leq x_\ell)}\sqrt{\frac{1}{\mathrm{Pr}(X< x_\ell)}}
               =\sum_{j=1}^{N-k}\frac{p_j}{\sum_{i=0}^{j}p_i}\sqrt{\frac{1}{\sum_{i=0}^{j-1}p_i}}.
             \end{equation}
             The basic idea is that we treat the expression on the right-hand side of~\eqref{eq:pSimplified} as an integral approximation sum for the integral $\int_{p_0}^{1}z^{-3/2}dz$, and show that it is actually always less than the value of this integral. We proceed by showing that subdivision always increases the sum.

             Let us fix some $\ell\in[N-k]$ and define
             \[
               p'_i =
               \begin{cases}
                 p_i &\text{for } i\in\{0,1,\ldots, \ell-1\}\\
                 p_i/2 & \text{for } i\in\{\ell, \ell+1\}\\
                 p_{i-1} & \text{for } i\in\{\ell+2,\ldots,N-k+1\}
               \end{cases}
             \]
             and observe that
             \begin{align}
               &\sum_{j=1}^{N-k+1}\frac{p'_j}{\sum_{i=0}^{j}p'_i}\sqrt{\frac{1}{\sum_{i=0}^{j-1}p'_i}}
                 -\sum_{j=1}^{N-k}\frac{p^{\phantom{'}}_{\kern-0.2mmj}}{\sum_{i=0}^{j}p_i}\sqrt{\frac{1}{\sum_{i=0}^{j-1}p_i}}\nonumber\\
               =&\frac{p'_\ell}{\sum_{i=0}^{\ell}p'_i}\sqrt{\frac{1}{\sum_{i=0}^{\ell-1}p'_i}}+\frac{p'_{\ell+1}}{\sum_{i=0}^{\ell+1}p'_i}\sqrt{\frac{1}{\sum_{i=0}^{\ell}p'_i}}-\frac{p_\ell}{\sum_{i=0}^{\ell}p_i}\sqrt{\frac{1}{\sum_{i=0}^{\ell-1}p_i}}\nonumber\\
               =&\frac{p_\ell/2}{p_\ell/2+\sum_{i=0}^{\ell-1}p_i}\sqrt{\frac{1}{\sum_{i=0}^{\ell-1}p_i}}+\frac{p_\ell/2}{p_\ell+\sum_{i=0}^{\ell-1}p_i}\sqrt{\frac{1}{p_\ell/2+\sum_{i=0}^{\ell-1}p_i}}-\frac{p_\ell}{p_\ell+\sum_{i=0}^{\ell-1}p_i}\sqrt{\frac{1}{\sum_{i=0}^{\ell-1}p_i}}.\label{eq:abDiff}
             \end{align}
             We show that~\eqref{eq:abDiff} is $\geq 0$ after simplifying the expression by substituting $a:=\sum_{i=0}^{\ell-1}p_i$ and $b:=p_\ell/2$:
             \begin{align}
               \frac{b}{a+b}\sqrt{\frac{1}{a}}+\frac{b}{a+2b}\sqrt{\frac{1}{a+b}}-\frac{2b}{a+2b}\sqrt{\frac{1}{a}}=\frac{\left(a+b-\sqrt{a} \sqrt{a+b}\right)b }{(a+b)^{3/2} (a+2 b)}\geq 0.
             \end{align}
             Let us fix some parameter $\delta>0$. Recursively applying this halving procedure
             for different indices, we can find some $J\in\mathbb{N}$ and $\tilde{p}\in\mathbb{R}_+^{J+1}$ such that $\sum_{j=0}^{J}\tilde{p}_j=1$, $\tilde{p}_0=p_0$ and $\tilde{p}_j\leq \delta$ for all $j\in [J]$, moreover
             \begin{align*}
               &	\sum_{j=1}^{N-k}\frac{p_j}{\sum_{i=0}^{j}p_i}\sqrt{\frac{1}{\sum_{i=0}^{j-1}p_i}}
                 \leq \sum_{j=1}^{J}\frac{\tilde{p}_j}{\sum_{i=0}^{j}\tilde{p}_i}\sqrt{\frac{1}{\sum_{i=0}^{j-1}\tilde{p}_i}} .
             \end{align*}
             Observe that for all $j\in [J]$
             \begin{align}\label{eq:deltaMore}
               \sqrt{\frac{1}{\sum_{i=0}^{j-1}\tilde{p}_i}}
               =\sqrt{\frac{1}{\sum_{i=0}^{j}\tilde{p}_i}}\sqrt{\frac{\sum_{i=0}^{j}\tilde{p}_i}{\sum_{i=0}^{j-1}\tilde{p}_i}}
               =&\sqrt{\frac{1}{\sum_{i=0}^{j}\tilde{p}_i}}\sqrt{1+\frac{\tilde{p}_j}{\sum_{i=0}^{j-1}\tilde{p}_i}}\nonumber	\\
               \leq& \sqrt{\frac{1}{\sum_{i=0}^{j}\tilde{p}_i}}\sqrt{1+\frac{\delta}{\tilde{p}_0}}
                     \leq \sqrt{\frac{1}{\sum_{i=0}^{j}\tilde{p}_i}}\left(1+\frac{\delta}{p_0}\right).
             \end{align}
             Therefore
             \begin{align*}
               \sum_{j=1}^{N-k}\frac{p_j}{\sum_{i=0}^{j}p_i}\sqrt{\frac{1}{\sum_{i=0}^{j-1}p_i}}
               \leq& \sum_{j=1}^{J}\frac{\tilde{p}_j}{\sum_{i=0}^{j}\tilde{p}_i}\sqrt{\frac{1}{\sum_{i=0}^{j-1}\tilde{p}_i}}\\
               \text{(by~\eqref{eq:deltaMore}) }\leq& \sum_{j=1}^{J}\frac{\tilde{p}_j}{\sum_{i=0}^{j}\tilde{p}_i}\sqrt{\frac{1}{\sum_{i=0}^{j}\tilde{p}_i}}\left(1+\frac{\delta}{p_0}\right)\\
               =& \left(1+\frac{\delta}{p_0}\right)\sum_{j=1}^{J}\frac{\tilde{p}_j}{\left(\sum_{i=0}^{j}\tilde{p}_i\right)^{3/2}}\\
               =& \left(1+\frac{\delta}{p_0}\right)\sum_{j=1}^{J}\int_{\sum_{i=0}^{j-1}\tilde{p}_i}^{\sum_{i=0}^{j}\tilde{p}_i}\frac{1}{\left(\sum_{i=0}^{j}\tilde{p}_i\right)^{3/2}} dz\\
               \leq& \left(1+\frac{\delta}{p_0}\right)\sum_{j=1}^{J}\int_{\sum_{i=0}^{j-1}\tilde{p}_i}^{\sum_{i=0}^{j}\tilde{p}_i}z^{-\frac{3}{2}} dz\\
               =& \left(1+\frac{\delta}{p_0}\right)\int_{p_0}^{1}z^{-\frac{3}{2}} dz\\
               =& \left(1+\frac{\delta}{p_0}\right)\left[-2 z^{-\frac{1}{2}}\right]_{p_0}^1 \\
               \leq& \left(1+\frac{\delta}{p_0}\right)\left(\frac{2}{\sqrt{p_0}}\right).
             \end{align*}
             Since this inequality holds for every $\delta>0$, we can conclude using~\eqref{eq:pSimplified} that
             \begin{align}\label{eq:intApx}
               \sum_{\ell=k+1}^{N}\frac{\mathrm{Pr}(X=x_\ell)}{\mathrm{Pr}(X\leq x_\ell)}\sqrt{\frac{1}{\mathrm{Pr}(X< x_\ell)}}
               \leq\frac{2}{\sqrt{p_0}}=\frac{2}{\sqrt{\mathrm{Pr}(X\leq x_k)}}.
             \end{align}
             \vskip-6mm
           \end{proof}

           It is not too hard to work out the constant by following the proof of~\cite{boyer1998TightBoundsOnQuantumSearching} providing something like $C\approx 25$. The following theorem works with any $C$ satisfying Lemma~\ref{lemma:genMinExp}.

           \begin{theorem}[Generalized Minimum-Finding]\label{thm:genMin}
             If we run Algorithm~\ref{alg:genMin} with input satisfying $M\geq 4C/\sqrt{\mathrm{Pr}(X\leq x)}$ for $C$ as in Lemma~\ref{lemma:genMinExp} and a unitary $U$ that acts on $q$ qubits, then at termination we obtain an $x_i$ from the range of $X$ that satisfies $x_i\leq x$ with probability at least $\frac{3}{4}$.
             Moreover the success probability can be boosted to at least $1-\delta$ with $\bigO(\log(1/\delta))$ repetitions. This uses at most $M$ applications of $U$ and $U^{-1}$ and $\bigO(qM)$ other gates.
           \end{theorem}
           \begin{proof}
             Let $x_k$ be the largest value in the range of $X$ such that $x_k\leq x$. Then Lemma~\ref{lemma:genMinExp} says that the expected number of applications of $U$ and $U^{-1}$ before finding a value $x_i\leq x_k$ is at most $C/\sqrt{\mathrm{Pr}(X\leq x_k)}=C/\sqrt{\mathrm{Pr}(X\leq x)}$, therefore by the Markov inequality we know that the probability that we need to use $U$ and $U^{-1}$ at least $4C/\sqrt{\mathrm{Pr}(X\leq x)}$ times is at most $1/4$.
             The boosting of the success probability can be done using standard techniques, e.g., by repeating the whole procedure $\bigO(\log(1/\delta))$ times and taking the minimum of the outputs.

             The number of applications of $U$ and $U^{-1}$ follows directly form the algorithms description. Then, for the number of other gates, each amplitude amplification step needs to implement a binary comparison and a reflection through the $\ket{0}$ state, both of which can be constructed using $\bigO(q)$ elementary gates, giving a total of $\bigO(qM)$ gates.
           \end{proof}

           Note that this result is a generalization of D{\"u}rr and H{\o}yer~\cite{durr1996QMinimumFinding}: if we can create a uniform superposition over $N$ values $x_1<x_2<\ldots<x_N$, then $\mathrm{Pr}(X\leq x_1)=1/N$ and therefore Theorem~\ref{thm:genMin} guarantees that we can find the minimum with high probability with $\bigO(\sqrt{N})$ steps.

           Now we describe an application of this generalized search algorithm that we need in the paper.

           This final lemma in this appendix describes how to estimate the smallest eigenvalue of a Hamiltonian. A similar result was shown by Poulin and Wocjan~\cite{poulin2009PrepGndStateManyBody}, but we improve on the analysis to fit our framework better.
           We assume sparse oracle access to the Hamiltonian $H$ as described in Section~\ref{sec:upperbounds}, and will count queries to these oracles. We use some of the techniques introduced in Appendix~\ref{apx:LowWeight}.

           \begin{lemma}\label{lemma:normEst}
             If $H\!=\!\sum_{j=1}^{n}E_j\ketbra{\phi_j}{\phi_j}$, with eigenvalues $E_1\leq E_2 \leq \ldots \leq E_n$, is such that $\nrm{H}\leq K$, $\eps\leq K/2$, and $H$ is given in $d$-sparse oracle form, then we can obtain an estimate $E$ such that  $\left|E_1-E\right|\leq \eps$, with probability at least $2/3$, using
             $$
             \bigO\left(\frac{Kd\sqrt{n}}{\eps}\log^2\left(\frac{Kn}{\eps}\right)\right)\text{ queries and }\bigO\left(\frac{Kd\sqrt{n}}{\eps}\log^{\frac{9}{2}}\left(\frac{Kn}{\eps}\right)\right) \text{ gates}.
             $$
           \end{lemma}
           \begin{proof}
             The general idea is as follows: we prepare a maximally entangled state on two registers, and apply phase estimation~\cite[Section~5.2]{nielsen2002QCQI}\cite{cleve1997QAlgsRevisited} to the first register with respect to the unitary $e^{\pi iH/K}$. We then use Theorem~\ref{thm:genMin} to find the minimal phase. In order to guarantee correctness we need to account for all the approximation errors coming from approximate implementations. This causes some technical difficulty, since the approximation errors can introduce phase estimates that are much less than the true minimum. We need to make sure that the minimum-finding algorithm finds these faulty estimates only with a tiny probability.

             We first initialize two $\log(n)$-qubit registers in a maximally entangled state $\smash{\frac{1}{\sqrt{n}}\sum_{j=0}^{n-1}\ket{j}\ket{j}}$. This can be done for example using $\log(n)$ Hadamard and CNOT gates, when $n$ is a power of two. Due to the invariance of maximally entangled states under transformations of the form $W\otimes W^*$ for unitary $W$, we have that
             $$
             \frac{1}{\sqrt{n}}\sum_{j=0}^{n-1}\ket{j}\ket{j}=\frac{1}{\sqrt{n}}\sum_{j=1}^{n}\ket{\phi_j}\ket{\phi^*_j}.
             $$

             Let $T:=2^{\left\lceil\log\left(\frac{K}{\eps}\right)+2\right\rceil}$ 
             and first assume that we have access to a perfect unitary $V$ which implements $V=\sum_{t=0}^{T-1}\ketbra{t}{t}\otimes e^{\pi t i H/K}$. Let $e_j:=E_jT/2K$. If we apply phase estimation to the quantum state $\ket{\phi_j}$ using $V$, then we get some superposition of phase estimates $\ket{e}$ in the first register. This superposition has the property, that a measurement in the computational basis reveals an $e$ such that $|e-e_j|\leq 3$ with high probability, so that the estimate $E:= e 2 K /T$ would satisfy $|E-E_j|=|e-e_j|2K/T\leq 3\eps/4<\eps$.
             If we repeat phase estimation $\bigO(\log(n))$ times (on the same state $\ket{\phi_j}$), and take the median of the estimates (in superposition), then we obtain a more concentrated superposition of estimates $e$ such that a measurement would reveal an $|e-e_j|\leq 3$ with probability at least $1-b/n$, for some $b=\Theta(1)$.

             Since in our maximally entangled state $\ket{\phi_j}$ is entangled with $\ket{\phi^*_j}$ on the second register, applying phase estimation to the first register in superposition does not cause interference.
             Let us denote by $U$ the unitary corresponding to the above preparation-estimation-boost procedure. Define $\Pi$ to be the projector which projects to the subspace of estimation values $\ket{e}$ such that there is a $j\in [n]$ with $|e-e_j|\leq \eps$. By the non-interference argument we can see that, after applying $U$, the probability that we get an estimation $e$ such that $|e-e_j|>3$ is at most $b/n$ for all $j\in [n]$, moreover $\nrm{(I-\Pi)U\ket{0}}^2\leq b/n$. Also let $\Pi_1$ denote the projector which projects to phase estimates that yield $e$ such that $|e-e_1|\leq 3$. It is easy to see that $\nrm{\Pi_1 U\ket{0}}^2\geq 1/n-b/n^2$.

             Now let us replace $V$ by $\tilde{V}$ implemented via Lemma~\ref{lemma:controlledHamsin}, such that $\nrm{\smash{V-\tilde{V}}}\leq c'/(n\log(n))$ for some $c'=\Theta(1)$. Let $\tilde{U}$ denote the circuit that we obtain from $U$ by replacing $V$ with $\tilde{V}$. Since in the repeated phase-estimation procedure we use $V$ in total $\bigO(\log(n))$ times, by using the triangle inequality we see that $\nrm{\smash{U-\tilde{U}}}\leq c/(2n)$, where $c=\Theta(1)$. We use the well-known fact that if two unitaries are $\delta$-close in operator norm, and they are applied to the same quantum state, then the measurement statistics of the resulting states are $2\delta$-close. Therefore we can upper bound the difference in probability of getting outcome $(I-\Pi)$:
             $$\nrm{\smash{(I-\Pi)\tilde{U}\ket{0}}}^2-\nrm{(I-\Pi)U\ket{0}}^2\leq 2\nrm{\smash{U\ket{0}-\tilde{U}\ket{0}}}\leq c/n,$$
             hence $\nrm{\smash{(I-\Pi)\tilde{U}\ket{0}}}^2\leq (b+c)/n$, and we can prove similarly that $\nrm{\smash{\Pi_1 \tilde{U}\ket{0}}}^2\geq 1/n-(b+c)/n$.

             Now let $\ket{\psi}:=(I-\Pi)\tilde{U}\ket{0}/\nrm{\smash{(I-\Pi)\tilde{U}\ket{0}}}$ be the state that we would get after post-selecting on the $(I-\Pi)$-outcome of the projective measurement $\Pi$. For small enough $b,c$ we have that $\nrm{\smash{\ket{\psi}-\tilde{U}\ket{0}}}= \bigO(\sqrt{(b+c)/n})$ by the triangle inequality.
             Thus there exists an idealized unitary $U'$ such that $\ket{\psi}=U'\ket{0}$, and $\nrm{\smash{\tilde{U}-U'}}= \bigO(\sqrt{(b+c)/n})$.
             Observe that $\nrm{\Pi_1 U'\ket{0}}^2=\nrm{\ket{\psi}}^2\geq \nrm{\smash{\Pi_1 \tilde{U}\ket{0}}}^2\geq 1/n-(b+c)/n$.

             Now suppose $(b+c)\leq 1/2$ and we run the generalized minimum-finding algorithm of Theorem~\ref{thm:genMin} using $U'$ with $M=6C\sqrt{n}$. Since
             $$\mathrm{Pr}(e\leq e_1+3)\geq\nrm{\Pi_1 U'\ket{0}}^2 \geq (1-b-c)/n \geq 1/(2n)>4/(9n)$$
             we will obtain an estimate $e$ such that $e\leq e_1+3$, with probability at least $3/4$. But since $\Pi\ket{\psi}=\ket{\psi}$, we find that any estimate that we might obtain satisfies $e\geq e_1-3$. So an estimate $e\leq e_1+3$ always satisfies $|e-e_1|\leq 3$.

             The problem is that we only have access to $\tilde{U}$ as a quantum circuit. Let $C_{MF}(\tilde{U})$ denote the circuit that we get from Theorem~\ref{thm:genMin} when using it with $\tilde{U}$ and define similarly $C_{MF}(U')$ for $U'$. Since we use $\tilde{U}$ a total of $\bigO(\sqrt{n})$ times in $C_{MF}(\tilde{U})$ and
             $$\nrm{\tilde{U}-U'}= \bigO(\sqrt{(b+c)/n})\text{, we get that } \nrm{C_{MF}(\tilde{U})-C_{MF}(U')}= \bigO(\sqrt{b+c}).
             $$
             Therefore the measurement statistics of the two circuits differ by at most $\bigO(\sqrt{b+c})$. Choosing $b,c$ small enough constants ensures that $C_{MF}(\tilde{U})$ outputs a proper estimate $e$ such that $|e-e_1|\leq 3$ with probability at least $2/3$. As we have shown at the beginning of the proof, such an $e$ yields an $\eps$-approximation of $E_1$ via $E:=e 2 K /T$.

             The query complexity has an $\bigO(Td\log(Tn))=\bigO(Kd/\eps\log(Kn/\eps))$ factor coming from the implementation of $\tilde{V}$ by Lemma~\ref{lemma:controlledHamsin}. This gets multiplied with $\bigO(\log(n))$ by the boosting of phase estimation, and by $\bigO(\sqrt{n})$ due to the minimum-finding algorithm. The gate complexity is dominated by the cost $\bigO(Kd/\eps\log^{7/2}(Kn/\eps))$ of implementing $\tilde{V}$, multiplied with the $\bigO(\sqrt{n}\log(n))$ factor as for the query complexity.
           \end{proof}
           Note that the minimum-finding algorithm of Theorem~\ref{thm:genMin} can also be used for state preparation. If we choose $2\eps$ less than the energy-gap of the Hamiltonian, then upon finding the approximation of the ground state energy we also prepare an approximate ground state. The precision of this state preparation can be improved with logarithmic cost, as can be seen from the proof of Lemma~\ref{lemma:normEst}.
           \section{Sparse matrix summation}\label{app:sparsematrixsum}
           As seen in Section~\ref{sec:upperbounds}, the Arora-Kale algorithm requires an approximation of $\exp(-\eta H^{(t)})$ where $H^{(t)}$ is a sum of matrices. To keep this section general we simplify the notation. Let $H$ be the sum of $k$ different $d$-sparse matrices $M$:
           \[
             H = \sum^{k}_{i=1} M_i
           \]
           In this section we study the complexity of one oracle call to $H$, given access to respective oracles for the matrices $M_1, \ldots, M_k$. Here we assume that the oracles for the $M_i$ are given in sparse matrix form, as defined in Section~\ref{sec:upperbounds}. In particular, the goal is to construct a procedure that acts as a similar sparse matrix oracle for $H$. We will only focus on the oracle that computes the non-zero indices of $H$, since the oracle that gives element access is easy to compute by summing the separate oracles.

           In the remainder of this section we only consider one row of $H$. We denote this row by $R_H$ and the corresponding rows of the matrices $M_i$ by $R_i$.
           Notice that such a row is given as an ordered list of integers, where the integers are the non-zero indices in $R_i$. Then $R_H$ will again be an ordered list of integers, containing all integers in the $R_i$ lists once (i.e., $R_H$ does not contain duplicates).

           \subsection{A lower bound}
           We show a lower bound on the query complexity of the oracle $O^I_H$ described above by observing that determining the number of elements in a row of $H$ solves the majority function.
           Notice that, given access to $O^I_H$, we can decide whether there are at least a certain number of non-zero elements in a row of $H$.

           \begin{lemma}
             Given $k+1$ ordered lists of integers $R_0,\dots,R_k$, each of length at most $d$. Let $R_H$ be the merged list that is ordered and contains every element in the lists $R_i$ only once (i.e., we remove duplicates).
             Deciding whether $|R_H| \leq  d+ \frac{dk}{2}$ or $|R_H| \geq d+\frac{dk}{2}+1$ takes $\Omega\left(dk\right)$ quantum queries to the input lists in general.
           \end{lemma}
           \begin{proof}
             We prove this by a reduction from MAJ on $dk$ elements. Let $Z\in \{0,1\}^{d\times k}$ be a Boolean string. It is known that it takes at least $\Omega(dk)$ quantum queries to $Z$ to decide whether $|Z|\leq \frac{dk}{2}$ or $|Z|\geq \frac{dk}{2}+1$.
             Now let $R_0,R_1,\dots,R_k$ be lists of length $d$ defined as follows:
             \begin{itemize}
             \item $R_0\lbrack j \rbrack = j(k+1)$ for $j = 1, \ldots, d$.
             \item $R_i\lbrack j \rbrack = j(k+1)+j Z_{ij}$ for $j = 1, \ldots, r$ and $i = 1,\ldots,k$.
             \end{itemize}
             By construction, if $Z_{ij} = 1$, then the value of the entry $R_i \lbrack j \rbrack$ is unique in the lists $R_0, \ldots, R_k$, and if $Z_{ij} = 0$ then $R_i [ j ] = R_0\lbrack j \rbrack$.
             So in $R_H$ there will be one element for each element in $R_0$ and one element for each bit in $Z_{ij}$ that is one. The length of $R_H$ is therefore $d+|Z|$. Hence, distinguishing between $|R_H| \leq d + \frac{dk}{2}$ and $|R_H| \geq d + \frac{dk}{2} + 1$ would solve the MAJ problem and therefore requires at least $\Omega(d k)$ queries to the lists in general.
           \end{proof}
           \begin{corollary}
             Implementing a query to a sparse matrix oracle $O^I_H$ for
             \[
               H = \sum_{j=i}^k M_j
             \]
             where each $M_j$ is $d$-sparse, requires $\Omega\left(dk\right)$ queries to the $O^I_{M_j}$ in general.
           \end{corollary}

           \subsection{An upper bound}
           We first show that an oracle for the non-zero indices of $H$ can be constructed efficiently classically. The important observation is that in the classical case we can write down the oracle once and store it in memory. Hence, we can create an oracle for $H$ as follows. We start from the oracle of $M_1$ and then we ``add'' the oracle of $M_2$, then that of $M_2$, etc. By ``adding'' the oracle $M_i$ to $H = \sum_{j=1}^{i-1} M_j$, we mean that, per row, we insert the non-zero indices in the list of $M_i$ into that of $H$ (if it is not already there). When an efficient data structure (for example a binary heap) is used, then such insertions can be done in polylog time. This shows that in the classical case such an oracle can be made in time $\bOt{ndk}$. Note that in the application we are interested in, Meta-algorithm~\ref{alg:AKSDP}, in each iteration $t$ only one new matrix $M^{(t)}$ `arrives', hence from the oracle for $H^{(t-1)}$ an oracle for $H^{(t)}$ can be constructed in time $\bOt{nd}$.

           The quantum case is similar, but we need to add all the matrices together each time a query to $O^I_H$ is made, since writing down each row of $H$ in every iteration would take $\Omega(n)$ operations.

           To implement one such query to $O^I_H$, in particular to the $t$th entry of $R_H$, start with an empty heap and add all elements of $R_1$ to it. Continue with the elements of $R_2$, but this time, for each element first check if it is already present, if not, add it, if it is, just continue.
           Overall this will take $\bigO(dk)$ insertions and searches in the data structure and hence $\bOt{d k}$ operations. We end up with a full description of $R_H$. We can then find the index of the $t$th non-zero element from this and uncompute the whole description.
           Similarly, we need to be able to compute the inverse function since we need an in-place calculation. Given an index $i$ of a non-zero element in~$R_H$, we can compute all the indices for $R_H$ as above, and find where $i$ is in the heap to find the corresponding value of~$t$.

           \section{Equivalence of \texorpdfstring{$R$, $r$, and $\eps^{-1}$}{R, r and 1 over epsilon}}\label{app:reductions}
           In this section we will prove the equivalence of the three parameters $R$, $r$ and $\eps^{-1}$ in the Arora-Kale meta-algorithm. That is, we will show any two of the three parameters can be made constant by increasing the third. Therefore, $\frac{Rr}{\eps}$ as a whole is the interesting parameter. This appendix is structured as a set of reductions, in each case we will denote the parameters of the new SDP with a tilde.
           \begin{lemma}
             Let $0 < \eps \leq 1$. For every SDP with $R,r \geq 1$, there is an SDP with parameters $\tilde{R}=1$ and $\tilde{r} = r$ such that a solution to that SDP with precision $\tilde{\eps} = \frac{\eps}{R}$ provides a solution with precision~$\eps$ to the original SDP.
           \end{lemma}
           \begin{proof}
             Let $\tilde{A}_j = A_j$, $\tilde{C} = C$ and $\tilde{b} = \frac{b}{R}$. Now clearly $\tilde{R} = 1$, but $\widetilde{\opt} = \opt/R$. Hence determining $\widetilde{\opt}$ up to additive error $\tilde{\eps}=\frac{\eps}{R}$ will determine $\opt$ up to additive error $\eps$. Notice that the feasible region of the dual did not change, so $\tilde{r} = r$.
           \end{proof}
           \begin{lemma}
             Let $0 < \eps \leq 1$. For every SDP with $R,r\geq 1$, there is an SDP with parameters $\tilde{R} = \frac{R}{\eps}$ and $\tilde{r} = r$ such that a solution to that SDP with precision $\tilde{\eps}=1$ provides a solution with precision~$\eps$ to the original SDP.
           \end{lemma}
           \begin{proof}
             Let $\tilde{A}_j = A_j$, $\tilde{C} = C$ and $\tilde{b} = \frac{b}{\eps}$. Now $\tilde{R} = \frac{R}{\eps}$ and $\widetilde{\opt} = \opt/\eps$. Hence determining $\widetilde{\opt}$ up to additive error $\tilde{\eps}=1$ will determine $\opt$ up to additive error $\eps$. Notice that again the feasible region of the dual did not change, so $\tilde{r} = r$.
           \end{proof}
           \begin{lemma}
             Let $0 < \eps \leq 1$. For every SDP with $R,r \geq 1$, there is an SDP with parameters $\tilde{R} = R$ and $\tilde{r}=1$ such that a solution to that SDP with precision $\tilde{\eps} = \frac{\eps}{r}$ provides a solution with precision~$\eps$ to the original SDP.
           \end{lemma}
           \begin{proof}
             Let $\tilde{A}_j = A_j$, $\tilde{b} = b$ and $\tilde{C} = \frac{1}{r} C$. Now $\tilde{r} = 1$ and $\widetilde{\opt} = \opt/r$. Hence determining $\widetilde{\opt}$ up to additive error $\tilde{\eps}=\frac{\eps}{r}$ will determine $\opt$ up to additive error $\eps$. Since $r \geq 1$ and $\nrm{C}\leq 1$, we find $\nrm{\smash{\tilde{C}}} \leq 1$ as required. Notice that the feasible region of the primal did not change, so $\tilde{R} = R$.
           \end{proof}
           At this point we would like to state the last reduction by setting $\tilde{C} = \frac{1}{\eps}C$, but this would not guarantee that $\nrm{\smash{\tilde{C}}} \leq 1$. Instead we give an algorithm that performs multiple calls to an SDP-solver, each of which has constant $\eps$ but higher $r$.

           \begin{lemma}
             Assume an SDP-solver that costs
             \[
               C(n,m,s,R,\eps,r)
             \]
             to solve one SDP with precision $\eps$, and assume that $C$ is non-decreasing in $r$. Every SDP with parameters $R,r \geq 1$ can be solved up to precision $0<\eps \leq 1$ with cost
             \[
               \sum_{k=1}^{\log\left(\frac{1}{\eps}\right)} C\left(n + 1,m+ 1,s,R+ 4 \log\left(\frac{1}{\eps}\right),1,2^k(r + 1)\right),
             \]
             by solving $\log\left(\frac{1}{\eps}\right)$ SDPs, where the $k$-th SDP is solved up to precision $\tilde \eps =1$ and has parameters
             \[
               \tilde{n} = n+1, \  \tilde{m} = m+1, \  \tilde{R} = \bigO\left(R+ 4 \log\left(\frac{1}{\eps}\right)\right), \  \tilde{r} \leq 2^k(r + 1),
             \]
             and input matrices with elements described by bitstrings of length $\text{poly}(\log n,\log m, \log\left(\frac{1}{\eps}\right))$.
             Furthermore, if $C(n,m,s,R,1,r) = \text{poly}(n,m,s,R,r)$, then the above cost is $\text{poly}(n,m,s,R,r/eps)$.
           \end{lemma}
           \begin{proof}
             The high-level idea is that we want to learn a small interval in which the optimum lies, whilst using a ``big'' precision of $1$. We do so as follows: given an interval $[L,U]$ with the promise that $\opt \in [L,U]$, we formulate another SDP for which a $1$-approximation of the optimum learns us a new, smaller, interval $[L',U']$ such that $\opt \in [L',U']$. We will moreover have $U'-L' \leq \frac{1}{2}(U-L)$. In the remainder of the proof we first show how to do this reformulation, we then use this technique to prove the lemma.

             \medskip

             Given an SDP $p$, given in the form of Equation~\eqref{eq:SDP}, of which we know an interval $[L,U]$ such that $\opt \in [L,U]$ (with $0 < U-L \leq 1$), we can write down an equivalent SDP $p'$ such that an optimal solution of $p$ corresponds one-to-one to an optimal solution of $p'$, and the optimum of $p'$ lies in $[0,4]$:
             \begin{align*}
               (p') \qquad \qquad  \max \quad & \tr{\begin{pmatrix} 0 & 0\\ 0 & 1 \end{pmatrix} \tilde{X}} \\
               \text{s.t.}\ \ \ & \tr{\begin{pmatrix} -C & 0 \\ 0 & \frac{U-L}{4} \end{pmatrix} \tilde{X}}  \leq -L,\\
                                              &\tr{\begin{pmatrix} A_j & 0 \\ 0 & 0 \end{pmatrix}  \tilde{X}}  \leq b_j \quad \text{ for all } j \in [m], \\
                                              &\tilde{X} \succeq 0,
             \end{align*}
             here the variable $\tilde{X}$ is of size $(n+1) \times (n+1)$, and should be thought of as $\tilde{X} = \begin{pmatrix} X & \cdot \\ \cdot & z\end{pmatrix}$, where $X$ is the variable of the original SDP: an $n \times n$ positive semidefinite matrix.
             Observe that by assumption, for every feasible $X$ of the original SDP, $L \leq \tr{CX} \leq U$. Therefore, the first constraint implies $0 \leq z \leq 4$ and hence the new optimum lies between $0$ and $4$, and the new trace bound is $\tilde{R} = R + 4$. We now determine~$\tilde{r}$.
             The dual of the above program is given by:
             \begin{align*}
               (d') \qquad \qquad \min \quad & -L y_0  + \sum_{j=1}^m b_j y_j \\
               \text{s.t.}\ \ \ & \begin{pmatrix} -C & 0 \\ 0 & \frac{U-L}{4} \end{pmatrix} y_0 + \sum_{j=1}^m \begin{pmatrix} A_j & 0 \\ 0 & 0 \end{pmatrix} y_j \succeq \begin{pmatrix} 0 & 0 \\ 0 & 1 \end{pmatrix} \\
                                             &y \geq 0.
             \end{align*}
             \begin{claim} \label{claim33}
               An optimal solution $\tilde{y}$ to $d'$ is of the form $\tilde{y} = \frac{4}{U-L}(1, y)$ where $y$ is an optimal solution to $d$, the dual of  $p$.
             \end{claim}
             The proof of the claim is deferred to the end of this section. The claim implies that $\tilde{r} = \frac{4}{U-L} (1+ r)$. We also have
             \[
               \opt = L + \opt' \frac{U-L}{4},
             \]
             and hence, a $1$-approximation to $\opt'$ gives a $\frac{U-L}{4}$-approximation to $\opt$.

             \medskip

             We now use the above technique to prove the lemma. Assume an SDP of the form~\eqref{eq:SDP} is given. By assumption, $\nrm{C} \leq 1$ and therefore $\opt \in [-R,R]$. Calling the SDP-solver on this problem with $\eps = 1$ will give us an estimate of $\opt$ up to additive error $1$. Call this estimate $\opt_0$, then $\opt \in [\opt_0 - 1, \opt_0 +1] =: [L_0,U_0]$. We now define a new SDP $p'$ as above, with $U = U_0, L= L_0$ (notice $U_0 - L_0 \leq 2$). By the above, solving $p'$ with $\tilde{\eps} = 1$ determines a new interval $[L_1, U_1]$ of length at most $2 \frac{U_0 - L_0}{4}$ such that $\opt \in [L_1, U_1]$.
             We use the interval $[L_1, U_1]$ to build a new SDP $p'$ and solve that with $\tilde{\eps} = 1$ to get an interval $[L_2, U_2]$. Repeating this procedure $k = \log\left(\frac{1}{\eps}\right)+1$ times determines an interval $[L_k, U_k]$ of length at most $\frac{1}{2^k}(U_0-L_0) \leq \eps$. Hence, we have determined the optimum of $p$ up to an additive error of $\eps$.
             The total time needed for this procedure is at most
             \[
               \sum_{k=1}^{\log\left(\frac{1}{\eps}\right)} C\left(n + 1,m+ 1,s,R+ 4 \log\left(\frac{1}{\eps}\right),1,2^k(r + 1)\right). \qedhere
             \]
           \end{proof}
           \begin{proof}[Proof of Claim~\ref{claim33}]
             First observe that the linear matrix inequality in $d'$ implies the inequality $y_0 \frac{U-L}{4} \geq 1$ and hence $y_0 \geq \frac{4}{U-L}$.
             Suppose we fix a value $y_0$ (with $y_0 \geq  \frac{4}{U-L}$) in $d'$, then the resulting SDP is of the form
             \begin{align*}
               (d'') \qquad \qquad -L y_0 +  \min \quad &   \sum_j b_j y_j \\
               \text{s.t.}\ \ \ & \sum_{j=1}^m y_j A_j \succeq y_0 C, \\
                                                        &y \geq 0.
             \end{align*}
             and hence, an optimal solution $\tilde y$ to $d''$ is of the form $\tilde y = y_0 y$ where $y$ is an optimal solution to~$d$. It follows that an optimal solution to  $d'$ is of the form $(y_0, y_0 y)$ where $y$ is an optimal solution to~$d$. Observe that the optimal value of $d'$ as a function of $y_0$ is of the form $y_0 \cdot (-L + \opt)$. Since by assumption $\opt \geq L$, the objective is increasing (linearly) with $y_0$ and hence $y_0 =  \frac{4}{U-L}$ is optimal.
           \end{proof}

\end{document}

%% file: Ggraph.txt
\begin{tikzpicture}[
        scale=\scale,
        axis/.style={thick, -, >=stealth'},
        arrow/.style={thick, ->, >=stealth'},
		qu/.style={very thick, blue,dashed},
        corner node/.style={color=red,large}
     ]

        \coordinate (corner) at (\cx,\cy);     
        \coordinate (corner3) at ($3*(corner)$);
        \coordinate (corner2) at ($2*(corner)$);
        \coordinate (UR) at (3,3);
		\coordinate (UL) at (-3,3);
        \coordinate (LR) at (3,-3);
		\coordinate (LL) at (-3,-3);
		\coordinate (cornerL) at (-3,\cy);
		\coordinate (cornerU) at (\cx,3);
        \coordinate (corner2L) at ($(-3,2*\cy)$);
		\coordinate (corner2U) at ($(2*\cx,3)$);
		\coordinate (corner3L) at ($(-3,3*\cy)$);
		\coordinate (corner3U) at ($(3*\cx,3)$);
	\fill[fill=blue!10, very thick] (corner) -- (corner3) -- (corner3U) -- (UL) -- (cornerL);
    \fill[fill=blue!10, very thick] (corner) -- (cornerU) -- (UL) -- (corner3);
    \fill[fill=blue!10, very thick] (corner3) -- (corner3U) -- (UL) -- (corner3L);
    \draw[axis] (-3,0)  -- (3,0) node(xline)[right] {};
    \draw[axis] (0,-3)  -- (0,3) node(xline)[right] {};
    \draw[arrow] (corner) -- (corner3);
    \draw[qu] (cornerL) -- (corner) -- (cornerU);
    \draw[qu] (corner2L) -- (corner2) -- (corner2U);
    \fill[red] (corner) circle (2pt) node[right] {};
\end{tikzpicture}

%% file: angles.txt
\begin{tikzpicture}[
        scale=\scale,
        line/.style={thick, -, >=stealth'},
	combi/.style={thick, dashed, >=stealth'},
corner node/.style={color=red,large}
     ]

        \coordinate (corner) at (2,-2);     
	\coordinate (pj) at ($(corner) + (-6,-3)$);
	\coordinate (pk) at ($(corner) + (5,7)$);
	\coordinate (UR) at (6,6);
	\coordinate (UL) at (-6,6);
        \coordinate (LR) at (6,-6);
	\coordinate (LL) at (-6,-6);
	
	\coordinate (cornerL) at (-6,-2);
	\coordinate (cornerU) at (2,6);

    \fill[fill=blue!10, very thick] (corner) -- (cornerU) -- (UL) -- (cornerL);

    \draw[combi] (pk) -- (pj);
    \draw (cornerU) coordinate (l2) node[right] {$L_2$}
    -- (corner)
    -- (cornerL) coordinate (l1) node[below] {$L_1$};

    \draw (pj) coordinate (pj) node[below] {$p_j$}
    -- (corner)
    -- (pk) coordinate (pk) node[right] {$p_k$};	

    \pic [draw, -, "$\angle L_2\ell_k$", angle eccentricity=3] {angle = pk--corner--l2};
    \pic [draw, -, "$\angle L_1 L_2$", angle eccentricity=2] {angle = l2--corner--l1};	
    \pic [draw, -, "$\angle \ell_j L_1$", angle eccentricity=3] {angle = l1--corner--pj};
    
    \node[text width=3cm] at ($(corner) + (1.8,-.3)$) {$(\alpha/r,\cp/r)$};

    \fill[red] (corner) circle (2pt) node[right] {};
\end{tikzpicture}